\documentclass[10pt]{article}
\usepackage{amsmath, amsthm, amssymb, pdfsync, psfrag, verbatim,color} 
\usepackage{hyperref}
\usepackage[round]{natbib}
\usepackage{breakcites}
\usepackage[affil-it]{authblk}
\usepackage{slantsc}

\newcommand{\be}{\begin{equation}}
	\newcommand{\ee}{\end{equation}}
\newcommand{\bea}{\begin{eqnarray}}
	\newcommand{\eea}{\end{eqnarray}}
\newcommand{\beas}{\begin{eqnarray*}}
	\newcommand{\eeas}{\end{eqnarray*}}
\newcommand{\eq}[1]{\begin{equation}\begin{aligned}#1\end{aligned}\end{equation}}
\newcommand{\eqn}[1]{\begin{align*}#1\end{align*}}

\newcommand{\bbE}{\mathbb E}
\newcommand{\bbF}{\mathbb F}

\newcommand{\bbN}{\mathbb N}

\newcommand{\bbP}{\mathbb P}

\newcommand{\bbR}{\mathbb R}

\newcommand{\scF}{\mathcal F}

\newcommand{\scL}{\mathcal L}
\newcommand{\scM}{\mathcal M}

\newcommand{\scR}{\mathcal R}

\newcommand{\scU}{\mathcal U}

\newcommand*\widebar[1]{%
	\hbox{%
		\vbox{%
			\hrule height 0.5pt 
			\kern0.5ex
			\hbox{%
				\kern-0.1em
				\ensuremath{#1}%
				\kern-0.1em
			}%
		}%
	}%
} 

\newcommand{\crl}[1]{\ensuremath{ \left\{ #1 \right\} }}
\newcommand{\edg}[1]{\ensuremath{ \left[ #1 \right] }}
\newcommand{\brak}[1]{\ensuremath{\left( #1 \right)}}

\newtheoremstyle{plain}%
  {9pt}
  {9pt}
  {\itshape}
  {}
  {\scshape}
  { }
  {0.5em}
  {}
\theoremstyle{plain}
\newtheorem{assumption}{Assumption}
\newtheorem{theorem}{Theorem}[section]

\newtheorem{proposition}[theorem]{Proposition}

\newtheorem{remark}[theorem]{Remark}
\newtheorem{example}[theorem]{Example}
\newtheorem{examples}[theorem]{Examples}
\newtheorem{foo}[theorem]{Remarks}

\newenvironment{Remark}{\begin{remark}\rm}{\end{remark}}

\usepackage{graphicx} 
\usepackage{float} 
\usepackage{subfigure} 
\allowdisplaybreaks[4]
\usepackage[title]{appendix}
\usepackage{abstract}

\usepackage{sectsty}
\sectionfont{\fontsize{10}{15}\selectfont}
\subsectionfont{\fontsize{10}{15}\selectfont\itshape}
\newcommand\FigInsert[1]{%
  \begin{center}
  \framebox{Figure \ref{#1} near here}
  \end{center}}
 \newcommand\TabInsert[1]{%
  \begin{center}
  \framebox{Table \ref{#1} near here}
  \end{center}}

\title{\textbf{Robust Wasserstein Optimization and its
Application in Mean-CVaR}}

\author{XIN HAI\thanks{Corresponding author.
Email: xin.hai@monash.edu} $\dag$ and KIHUN NAM${\ddag}$}
\affil{$\dag\ddag$Monash University, Clayton, VIC 3800, Australia}
\begin{document}
\date{}
\maketitle
\begin{abstract}
\noindent We refer to recent inference methodology and formulate a framework for solving the distributionally robust optimization problem, where the true probability measure is inside a Wasserstein ball around the empirical measure and the radius of the Wasserstein ball is determined by the empirical data. We transform the robust optimization into a non-robust optimization with a penalty term and provide the selection of the Wasserstein ambiguity set's size. Moreover, we apply this framework to the robust mean-CVaR optimization problem and the numerical experiments of the US stock market show impressive results compared to other popular strategies.\\[4mm]
\end{abstract}
\setcounter{equation}{0}
\section{Introduction}
In this article, we use a robust Wasserstein optimization (RWO) framework to solve the distributionally robust optimization problem, where the true probability measure is inside a Wasserstein ball centered at the empirical probability measure. We transform the robust optimization problem into a non-robust minimization problem with a penalty term and provide an appropriate selection of the Wasserstein ambiguity set's size. Then we apply this RWO framework to the robust mean-CVaR (conditional value-at-risk) optimization and obtain its dual non-robust minimization with a penalty term. Next, we apply the robust mean-CVaR model to the US stock market on five different 10-year intervals between 2002 and 2019, which provides the portfolio allocation, the Sharpe ratio, the mean/CVaR, and the cumulative wealth competitive with the classical mean-CVaR portfolio and various robust portfolio strategies.

It is not hard to solve the non-robust optimization problem. Although the true probability measure is not precisely known in practice, it is often estimated through a finite set of historical data. Unfortunately, the non-robust model is extremely sensitive to the underlying parameters and the empirical version can deviate significantly from the true one.

In order to address this issue, robust optimization has been widely used. The robust formulations consider the case where the true probability differs from the empirical one. A robust solution is obtained by considering the worst-case scenario inside the ambiguity set, which represents a class of models that are plausible variations of the empirical measure. Therefore, it is important to be clear about the ambiguity set when studying robust optimization problems. There are some typical choices for the ambiguity set. \cite{hansen2001robust} considered a set of measures that are equivalent to the empirical measure. \cite{lobo2000worst}, and \cite{tutuncu2004robust} analyzed an inverse image of the measurable function, such as the mean and variance can be allowed to be in certain specified intervals.  \citealp[Chapter 12]{fabozzi2007robust}, \citealp{hota2019data}, \cite{fournier2015rate} considered a topological neighborhood of the empirical probability measure, such as a Wasserstein ball that is around the empirical measure.

The size of the ambiguity set is crucial as well. If the ambiguity set's size is too large, then the available data becomes less relevant and the result is too conservative. If the size is too small, then the effect of robustness will be lost. Therefore, it is necessary to choose an appropriate size of the ambiguity set based on the context of one's problem. Recently, \cite{delage2010distributionally}, \cite{esfahani2018data}, and \cite{blanchet2019data} studied the data-driven problem to provide some approaches for the choice of the ambiguity set's size. In particular, \cite{blanchet2019robust} studied the ambiguity set around the empirical measure and the distance between two probability measures is dictated by the Wasserstein metric. They provided a novel method, Robust Wasserstein Profile Inference (RWPI), to choose an appropriate Wasserstein ball for a given confidence level.

The RWO framework is mainly concerned with the selection of the Wasserstein ambiguity set's size and the dual problem of robust optimization. First, we define the non-robust optimization problem and robust Wasserstein profile (RWP) function to obtain the Wasserstein ambiguity set's size by adapting and modifying the RWPI method. Then We transform the robust optimization into a non-robust optimization with a penalty term. Finally, we consider the application of the RWO framework to the portfolio optimization problem and compare its performance with other strategies.

Recently, VaR (value-at-risk) and CVaR become the most popular risk measures. VaR reflects the maximum potential loss of an asset for a given confidence level and period. However, VaR has some drawbacks, such as the lack of subadditivity. Hence it is not a coherent risk measure (\cite{artzner1999coherent}). To overcome the limitations of VaR, \cite{rockafellar2000optimization} studied a modified version of CVaR, which was defined as the mean of the tail distribution exceeding VaR. As a coherent risk measure, CVaR is currently more widely recommended than VaR by theoreticians and market practitioners.

We consider the robust mean-CVaR optimization problem with a cost function $c(u,w)=\|u-w\|_2^\kappa$, which corresponds to the Wasserstein metric of order $\kappa$. Based on our RWO framework, we transform the robust mean-CVaR optimization problem into a non-robust minimization problem with a penalty term for $\kappa=1$ and $2$. Next, we choose an appropriate Wasserstein ambiguity set's size using a data-driven method proposed by \cite{blanchet2019robust} for the mean-variance optimization. Since the mean-CVaR problem is not regular enough to apply \cite{blanchet2019robust} method, we regularise the equation using the mollifier used in \cite{peng1999non} and \cite{tong2008available}. We obtain the explicit expression of the Wasserstein ambiguity set's size for the robust mean-CVaR model when $\kappa=1$ and $2$.

Then we select the top 100 US stock price data in the five ten-year periods 2002-2012, 2004-2014, 2006-2016, 2008-2018, and 2009-2019 to perform numerical simulations with and without transaction costs. We choose the first two years as the in-sample to estimate the parameters, and the last eight years as the out-of-sample to test the model. We analyze the influence of the smooth parameter in the smooth mean-CVaR problem compared with the non-robust problem. We show the convergence of the Wasserstein ambiguity set's size for the robust mean-CVaR models as the in-sample size grows. We compared the performance (the portfolio percentage allocation, the Sharpe ratio, the mean/CVaR, and the cumulative wealth) of our robust mean-CVaR models for $\kappa=1$ and $2$ with the classical mean-CVaR model in \cite{rockafellar2000optimization} and \cite{rockafellar2002conditional}, the robust mean-CVaR model under box uncertainty in \cite{zhu2009worst}, and the robust mean-CVaR model under distribution ambiguity in \cite{kang2019data}.

The numerical experiments showed impressive results. Our robust models perform better than non-robust and other robust strategies in most cases except in the 2004-2014 and 2008-2018 simulations, where our robust models showed moderate performance on the mean/CVaR. Moreover, our robust models have the most diversified portfolios among all strategies.

The rest of the paper is organized as follows. In Section \ref{section2} we formulate the RWO framework to solve the robust Wasserstein optimization problem. Besides, we provide an appropriate selection of the Wasserstein ambiguity set's size. Section \ref{section3} is our application of the RWO framework, the robust mean-CVaR problem. Section \ref{section4} presents the implementation of our investment strategy. Section \ref{section5} is
concerned with the numerical experiments of our robust strategies compared to several popular strategies. Finally, Section \ref{section6} concludes this paper. 

\section{RWO framework}
\label{section2}
\subsection{Robust problem formulation}
In this section, we are interested in studying a distributionally robust optimization problem, given by
\eq{\label{eq1}&\min_{\boldsymbol{x}\in\scF_{(\delta,\rho)}}\max_{P\in\scU_\delta(Q)}E_P[f(\boldsymbol{x},\boldsymbol{\theta})],}
where $f(\boldsymbol{x},\boldsymbol{\theta})$ is a function that depends both on the $n$-dimensional decision vector $\boldsymbol{x}$ and the $n$-dimensional random vector $\boldsymbol{\theta}$, whose distribution $P$ is supported on the ambiguity set $\scU_\delta(Q)$, $P$ is the underlying probability distribution, $Q$ is the empirical probability derived from historical data, $\scF_{(\delta,\rho) }:=\{\boldsymbol{x}:\min_{P\in\scU_\delta(Q)}E_P[g(\boldsymbol{x},\boldsymbol{\theta})]\geq\rho\}$ is the feasible region with the worst case constraint, $g$ is a measurable function, $\rho$ is the worst acceptable target value, $\delta$ is the discrepancy between two probability measures $P$ and $Q$.

We need to define the optimal transport cost and Wasserstein distance to define our ambiguity set. Firstly, we assume the cost function $c: \bbR^n\times\bbR^n\rightarrow [0,\infty]$ is lower semi-continuous and $c(u,w)=0$ if and only if $u=w$ for $u, w\in \bbR^n$. Given two probability measures $\mu$ and $\nu$, the optimal transport cost or discrepancy between $\mu$ and $\nu$, denoted by $C(\mu,\nu)$, can be defined as
\eqn{C(\mu,\nu)=\inf_{(U,W)\in\Pi(\mu,\nu)} \bbE[c(U,W)],
}
where \eqn{\Pi(\mu,\nu):=\crl{(U,W): \bbP\circ U^{-1} =\mu, \bbP\circ W^{-1}=\nu}.
}
Intuitively speaking, the optimal transport cost $C(\mu,\nu)$ is measuring the cheapest way of rearranging the probability distribution $\mu$ into $\nu$, where the cost function $c(u,w)$ can be interpreted as the cost of transporting unit mass from $u$ to $w$.

For $\kappa \geq 1$, the Wasserstein distance of order $\kappa$ between $\mu$ and $\nu$ is defined as
\eqn{W_\kappa(\mu,\nu)=\brak{\inf_{(U,W)\in\Pi(\mu,\nu)}\bbE[d(U,W)^\kappa]}^{1/\kappa},
}
where $d$ is the metric defined on the codomain of $U$ and $W$.

According to 
\cite{villani2009optimal}, if $c^{1/\kappa}(\cdot,\cdot)$ is a metric for any $\kappa\geq 1$, then $C^{1/\kappa}(\cdot,\cdot)$ is also a metric between two probability distributions. In that case,
\eqn{C^{1/\kappa}(\mu,\nu)=\brak{\inf_{(U,W)\in\Pi(\mu,\nu)} \bbE\left[c(U,W)\right]}^{1/\kappa}.
}
If we choose $c(u,w)=d^{\kappa}(u,w)$, then $C^{1/\kappa}(\mu,\nu)$ is the standard Wasserstein distance of order $\kappa$. Throughout this article, we will select $C^{1/\kappa}(\cdot,\cdot)$ as the discrepancy between two probability measures and choose the following cost function
\eqn{c(u,w)= \|u-w\|_2^\kappa,
}
where $\kappa\geq 1$, and $\|\cdot\|_2$ is the Euclidean norm on $\bbR^n$.

Finally, for a given $\delta>0$, we define an ambiguity set by $\delta$-neighborhood of $Q$, 
\eqn{\scU_\delta(Q):=\{P:C^{1/\kappa}(P,Q)\leq\delta\}.}

The distribution $P$ is not precisely known and is often partially observable through a finite set of independent samples of past realizations $(\theta_i)_{i=1, 2, ..., N}$ of the random vector $\boldsymbol{\theta}$. To generate data-driven solutions, we approximate the distribution $P$ with a discrete empirical
probability distribution $Q$,

\eqn{Q(\cdot)=\frac{1}{N}\sum_{i=1}^{N}\delta_{\theta_i}(\cdot),
}
where $\delta_{\theta_i}(\cdot)$ is the indicator function.

\begin{assumption}
\label{as1}
Assume that the cost function $c(\boldsymbol{\theta},\theta_i)=\|\boldsymbol{\theta}-\theta_i\|^{\kappa}_2$ for $\kappa\geq 1$.
\end{assumption}

\begin{proposition}
\label{proposition1}
The primal problem given in (\ref{eq1}) is equivalent to the following dual problem:
\eqn{& \min_{\gamma_1\geq 0,\boldsymbol{x}}\quad\gamma_1\delta+\frac{1}{N}\sum_{i=1}^N \brak{\sup_{\boldsymbol{\theta}}\crl{f(\boldsymbol{x},\boldsymbol{\theta})-\gamma_1\|\boldsymbol{\theta}-\theta_i\|_2^\kappa}},\\
&\text{s.t.}\quad -\inf_{\gamma_0\geq0}\left\{\gamma_0\delta+\frac{1}{N}\sum_{i=1}^N\brak{\sup_{\boldsymbol{\theta}}\{-g(\boldsymbol{x},\boldsymbol{\theta})-\gamma_0\|\boldsymbol{\theta}-\theta_i\|_2^\kappa\}}\right\}\geq\rho,}
in the sense that the two problems have the same optimal values.
\end{proposition}
\begin{proof}
See Section \ref{proofproposition1}.
\end{proof}

In addition, we assume the following assumptions throughout this article:
\begin{assumption}
\label{as2}
Given a target value $\rho$, the non-robust problem
\eq{\label{nonrobust}&\min_{\boldsymbol{x}}\;E_{P^*}[f(\boldsymbol{x},\boldsymbol{\theta})]\quad \text{subject to}\quad E_{P^*}[g(\boldsymbol{x},\boldsymbol{\theta})]=\rho}
has a unique solution $\boldsymbol{x}^*$. 
\end{assumption}
\begin{assumption}
\label{as3}
$f(\boldsymbol{x},\boldsymbol{\theta})$ and $g(\boldsymbol{x},\boldsymbol{\theta})$ are  differentiable with respect to $\boldsymbol{x}$.
\end{assumption}
According to Assumptions \ref{as2} and \ref{as3}, we have $\Pi_{P^*}=\crl{\boldsymbol{x}^*}$. 
Let us define
\[
h(\boldsymbol{x}^*,\boldsymbol{\theta})= \displaystyle\frac{\partial}{\partial x}f(\boldsymbol{x}^*,\boldsymbol{\theta})-\lambda_0^*\frac{\partial}{\partial x}g(\boldsymbol{x}^*,\boldsymbol{\theta})
\]
where $\lambda_0^*$ is the Lagrange multiplier appear in optimization problem of \eqref{nonrobust}.
In this case, the first order derivative optimal condition can be written as  
\eq{\label{non1.1}
&E_{P^*}[h(\boldsymbol{x}^*,\boldsymbol{\theta})]=\boldsymbol{0}.
}
 We assume the following conditions for $h(\boldsymbol{x}^*,\boldsymbol{\theta})$.
\begin{assumption}
	\label{as4}
$E_{P^*}\|h(\boldsymbol{x}^*,\boldsymbol{\theta})\|^2_2<\infty$ and $h(\boldsymbol{x}^*,\boldsymbol{\theta})$ is continuously differentiable with respect to $\boldsymbol{\theta}$, $D_{\boldsymbol{\theta}}h(\boldsymbol{x}^*,\boldsymbol{\theta})$.
\end{assumption}
\begin{assumption}
\label{as5}
For any $\xi_0\in \bbR^n\backslash\crl{\textbf{0}}$, the function $h(\boldsymbol{x}^*,\boldsymbol{\theta})$ in Assumption \ref{as4} satisfies
\eqn{P^*(\|\xi_0^\intercal D_{\boldsymbol{\theta}}h(\boldsymbol{x}^*,\boldsymbol{\theta})\|_2>0)>0.}
\end{assumption}
\subsection{Selection of the Wasserstein ambiguity set's size}
Assume the worst acceptable target value $\rho$ be chosen by the investor. We would like to choose the ambiguity set size as the minimum uncertainty level so that the optimal solution of non-robust model (\ref{nonrobust}) under $P^*$ is in the neighborhood of the estimate of $\boldsymbol{x}^*$ with confidence level $1-\delta_0$ which is assigned by the investor. In other words,
	\begin{align*}
{\rm arg}\min\crl{\delta:P^*\left(\boldsymbol{x}^*\in\bigcup_{P\in \scU_\delta(Q)}\Pi_P\right)\geq 1-\delta_0},
	\end{align*}
where $\Pi_P$ is the set of all solutions of the non-robust problem (\ref{nonrobust}) under $P$. 

In order to solve the problem asymptotically, we refer to 
the RWPI approach illustrated in \cite{blanchet2019robust}. For RWP function
\eq{\label{RN}\scR_N(\kappa)=\min\{C^{1/\kappa}(P,Q):E_P[h(\boldsymbol{x}^*,\boldsymbol{\theta})]=\boldsymbol{0}\},}
it is known that
\[
\crl{\scR_N(\kappa)\leq\delta}=\crl{\boldsymbol{x}^*\in\bigcup_{P\in \scU_\delta(Q)}\Pi_P}.
\]
Therefore, our choice for the ambiguity set size is
\eqn{
	\delta^* :=\min\crl{\delta:P^*\left(\scR_N(\kappa)\leq\delta\right)\geq 1-\delta_0},
}
where $1-\delta_0$ is a confidence level. 

Then, we use the following convergence result to obtain $\delta^*$.
\begin{theorem}[Theorem 3, \cite{blanchet2019robust}]
\label{theorem1}
Under Assumptions \ref{as1} to \ref{as5}, as $N\rightarrow \infty$, we have
\eqn{N^{\kappa/2}\scR_N(\kappa)\Rightarrow \widebar{\scR}(\kappa),}
where, for $\kappa >1$,
\eqn{\widebar{\scR}(\kappa)=\max_{\xi_0\in\bbR^n}\crl{\kappa\xi_0^\intercal H-(\kappa-1)E_{P^*}\|\xi_0^\intercal D_{\boldsymbol{\theta}}h(\boldsymbol{x}^*,\boldsymbol{\theta})\|_2^{\kappa/(\kappa-1)}}; }
if $\kappa =1$,
\eqn{\widebar{\scR}(1)=\max_{\xi_0:P^*(\|\xi_0^\intercal D_{\boldsymbol{\theta}}h(\boldsymbol{x}^*,\boldsymbol{\theta})\|_2>1)=0}\crl{\xi_0^\intercal H},}
where $H\sim N(\boldsymbol{0},Cov[h(\boldsymbol{x}^*,\boldsymbol{\theta})])$ and $Cov[h(\boldsymbol{x}^*,\boldsymbol{\theta})]=E_{P^*}[h(\boldsymbol{x}^*,\boldsymbol{\theta})h(\boldsymbol{x}^*,\boldsymbol{\theta})^\intercal]$.

In particular, if $\kappa=2$ and $E_{P^*}[D_{\boldsymbol{\theta}}h(\boldsymbol{x}^*,\boldsymbol{\theta})]$ is invertible,
\eqn{\widebar{\scR}(2)= \max_{\xi_0\in\bbR^n}\crl{2\xi_0^\intercal H-\xi_0^\intercal E_{P^*}[ D_{\boldsymbol{\theta}}h(\boldsymbol{x}^*,\boldsymbol{\theta})]\xi_0}=H^\intercal(E_{P^*}[ D_{\boldsymbol{\theta}}h(\boldsymbol{x}^*,\boldsymbol{\theta})])^{-1}H.
}
\end{theorem}

In sum, we first estimate $\boldsymbol{x}^*$ 
by calculating the optimal solution from non-robust problem (\ref{nonrobust}) under probability $Q$. Then we use the result of Theorem \ref{theorem1} to obtain the $(1-\delta_0)$-quantile of $\;\widebar{\scR}(\kappa)$, which is denoted by $\eta_{1-\delta_0}$. Then, since $N^{\kappa/2}\scR_N(\kappa)\Rightarrow \widebar{\scR}(\kappa)$ as $N\rightarrow \infty$, we choose the appropriate size of the Wasserstein ambiguity set
\eqn{\delta^* = \eta_{1-\delta_0}/N^{\kappa/2}.
}

\section{Robust mean-CVaR application}
\label{section3}
\subsection{Robust mean-CVaR problem formulation}
After formulating the RWO framework, we consider its application in the portfolio optimization problem. Assume a market with $n$ risky assets and one risk-free asset without an interest rate. Let $(\Omega,\scF,\bbF, P^*)$ be a filtered probability space with $n$ discrete stochastic processes $\{S^i_t: t=1,...,T\}, i=1,2,..,n$ and $\bbF$ is the filtration generated by $\{S^i:i=1,2,...,n\}$. Here, $S^i_t$ represents the $i$th risky asset price at time $t$. Let us denote $R^i$ the return of risky asset $i$ on the time interval $[1, T]$, that is $R^i=(S_{T}^i-S_1^i)/S_1^i $. Then, we can define a $n$-dimensional return vector $\boldsymbol{R}=(R^1,R^2,...,R^n)$. Throughout this paper, any vector is understood to be a column vector and the transpose of $\boldsymbol{R}$ is denoted by $\boldsymbol{R}^{\intercal}$ and ${\bf 1}$ denotes a vector of appropriate dimension where each element is $1$. 

We let $\pi^i$ be the amount of money invested in the $i$th risky asset on the time interval $[1, T]$ and 
$\boldsymbol{\pi}=(\pi^1,\pi^2,...,\pi^n).$ Define the portfolio's wealth return $w(\boldsymbol{\pi},\boldsymbol{R}):=\boldsymbol{\pi}^\intercal \boldsymbol{R}$
and its loss $L:=loss(\boldsymbol{R})=-\boldsymbol{\pi}^\intercal \boldsymbol{R}.$
Let $F_L$ be its cumulative distribution function $F_L(l) =\bbP\{L\leq l\}$. Let $F^{-1}_L(v)$ be its left continuous inverse $F^{-1}_L(v)=\min\{l:F_L(l)\geq v\}$.

For a given $\delta>0$, we define an ambiguity set by $\delta$-neighborhood of $Q$, 
\eqn{\scU_\delta(Q):=\{P:C^{1/\kappa}(P,Q)\leq\delta\},}
where $Q$ is the empirical probability we obtained from historical data, that is $Q(\cdot)=\frac{1}{N}\sum_{i=1}^{N}\delta_{R_i}(\cdot),$
where $(R_i)_{i=1, 2, ..., N}$ are realizations of $\boldsymbol{R}$ and $\delta_{R_i}(\cdot)$ is the indicator function.
The feasible region of portfolio investment is given by
\eqn{\scF_{(\delta,\rho)}=\crl{\boldsymbol{\pi}:\max_{P\in\scU_{\delta}(Q)}E_P(\boldsymbol{\pi}^\intercal\boldsymbol{R})\geq \rho,\;\boldsymbol{\pi}^\intercal \boldsymbol{1}=1}.
}

Now let us define the worst-case conditional value-at-risk (WCVaR$_\alpha$). For a given risk level $\alpha\in(0,1)$, we denote the value-at-risk VaR$_\alpha$ and the conditional value-at-risk CVaR$_\alpha$ as
\eqn{&VaR_{\alpha}(L):=\min\{l:F_L(l)\geq\alpha\}=F^{-1}_L(\alpha),\\
&CVaR_{\alpha}(L):=E_P[L|L>VaR_{\alpha}(L)].}
\cite{rockafellar2000optimization} and \cite{rockafellar2002conditional} demonstrated that the calculation of CVaR$_\alpha$ can be achieved by minimizing the following auxiliary function
\eqn{CVaR_{\alpha}(L)=\min_a\crl{a +\frac{1}{\alpha}E_P[L-a]^+},}
where $a\in\bbR$ is the threshold of the loss and $[L-a]^+=\max(L-a,0)$.

This dual representation led us to define the the worst-case CVaR$_\alpha$ as follow:
\eq{\label{eq3.7}WCVaR_{\alpha}(-\boldsymbol{\pi}^\intercal\boldsymbol{R})&:=\max_{P\in\scU_{\delta}(Q)}CVaR_{\alpha}(-\boldsymbol{\pi}^\intercal\boldsymbol{R})=\max_{P\in\scU_{\delta}(Q)}\min_{\boldsymbol{\pi}\in\scF_{(\delta,\rho)},a}\crl{a +\frac{1}{\alpha}E_P[-\boldsymbol{\pi}^\intercal\boldsymbol{R}-a]^+}.}
\begin{assumption}
\label{ass6.1}
Assume that the return of risky asset $i$, $R^i$, is a bounded random variable.
\end{assumption}
\begin{theorem}
\label{theorem3.3}
For given $\delta>0$,
\eq{\label{eq3.9}WCVaR_{\alpha}(-\boldsymbol{\pi}^\intercal\boldsymbol{R})
&=\min_{\boldsymbol{\pi}\in\scF_{(\delta,\rho)},a}\max_{P\in\scU_{\delta}(Q)}\crl{a +\frac{1}{\alpha}E_P[-\boldsymbol{\pi}^\intercal\boldsymbol{R}-a]^+}.}
\end{theorem}
\begin{proof}
Let us rewrite \eqref{eq3.7} as
\eqn{WCVaR_{\alpha}(-\boldsymbol{\pi}^\intercal\boldsymbol{R})&=\max_{P\in\scU_{\delta}(Q)}\min_{\boldsymbol{\pi}\in\scF_{(\delta,\rho)},a}F_{\alpha}(\boldsymbol{\pi},a),
}
where $F_{\alpha}(\boldsymbol{\pi},a,\boldsymbol{R})=a +\displaystyle\frac{1}{\alpha}E_P[-\boldsymbol{\pi}^\intercal\boldsymbol{R}-a]^+$. 

As a function of $a\in\bbR$ and $\boldsymbol{\pi}$, $F_{\alpha}(\boldsymbol{\pi},a,\cdot)$ is finite and convex. According to Assumption \ref{ass6.1}, as a function of $\boldsymbol{R}$, $F_{\alpha}(\cdot,\cdot,\boldsymbol{R})$ is bounded. 

In addition, \cite{pichler2018quantitative} (Proposition 3) proved that the Wasserstein ambiguity set $\scU_{\delta}(Q)$ is also tight for every $\delta>0$ and weakly compact. Therefore, by Theorem II.2 of \cite{hota2018}, the claim is proved.
\end{proof}

Then we can define the robust mean-CVaR optimization problem 
\eq{\label{eq4.10}&\min_{\boldsymbol{\pi}, a}\max_{P\in\scU_\delta(Q)}\crl{a +\frac{1}{\alpha}E_P[-\boldsymbol{\pi}^\intercal\boldsymbol{R}-a]^+},\\
&\text{s.t.}\;\min_{P\in\scU_\delta(Q)}\;E_{P}(\boldsymbol{\pi^\intercal \boldsymbol{R}})\geq \rho,\\
&\qquad\;\boldsymbol{\pi}^\intercal \boldsymbol{1}=1.}

In the above section, we get the dual formulation of the robust optimization. According to Proposition \ref{proposition1} in our RWO framework, we obtain the dual problem of the robust mean-CVaR (\ref{eq4.10}).
\begin{proposition}
\label{proposition4.1}
Consider the cost function $c(u,w)=\|u-w\|_2^{\kappa}$ with $\kappa=1$. The robust mean-CVaR problem given in (\ref{eq4.10}) is equivalent to the following dual problem:
\begin{equation}\label{RMC-1}\tag{RMC-1}\begin{aligned}
&\min_{\boldsymbol{\pi},a}\;\crl{a +\displaystyle\frac{1}{\alpha}\left(E_Q[-\boldsymbol{\pi}^\intercal \boldsymbol{R}-a]^+ + \delta\|\boldsymbol{\pi}\|_2\right)},\\
&\text{s.t.}\quad E_{Q}[\boldsymbol{\pi}^{\intercal}\boldsymbol{R}]-\delta\|\boldsymbol{\pi}\|_2\geq \rho,\\
&\qquad\;\boldsymbol{\pi}^\intercal \boldsymbol{1}=1.
\end{aligned}\end{equation}
If $\kappa=2$, the robust mean-CVaR (\ref{eq4.10}) admits the following dual problem:
\begin{equation}\label{RMC-2}\tag{RMC-2}\begin{aligned}
& \min_{\boldsymbol{\pi},a,\gamma,s_i}\gamma\delta+\frac{1}{N}\sum_{i=1}^N s_i,\\
&\text{s.t.}\; \frac{\|\boldsymbol{\pi}\|_2^2}{4\gamma\alpha^2}-\frac{1}{\alpha}\boldsymbol{\pi}^\intercal R_i+a\brak{1-\frac{1}{\alpha}}\leq s_i,\\
&\qquad a\leq s_i,\\
&\qquad E_{Q}[\boldsymbol{\pi}^{\intercal}\boldsymbol{R}]-\sqrt{\delta}\|\boldsymbol{\pi}\|_2\geq\rho,\\
&\qquad\boldsymbol{\pi}^\intercal 1= 1,\\
&\qquad\gamma\geq 0.
\end{aligned}\end{equation}
\end{proposition}
\begin{proof}
See Section \ref{proofA1}.
\end{proof}

\subsection{Selection of the Wasserstein ambiguity set's size for robust mean-CVaR problem}
Once obtaining the dual problems \ref{RMC-1} and \ref{RMC-2} for the robust mean-CVaR (\ref{eq4.10}), next we will choose an appropriate size of the Wasserstein ambiguity set for the robust mean-CVaR model. Because the value of $\delta$ too large or too small has a detrimental effect on the robust optimal solution. 

First let us state the non-robust mean-CVaR problem:
\eq{\label{nonsmoothing}\;&\min_{\boldsymbol{\pi},a}\;\left\{a +\frac{1}{\alpha}E_P\left[-\boldsymbol{\pi}^\intercal \boldsymbol{R}-a\right]^+\right\}\quad \text{ subject to }\quad E_P[\boldsymbol{\pi}^\intercal \boldsymbol{R}]= \rho\quad\text{and}\quad\boldsymbol{\pi}^\intercal 1= 1.}
As the max function $\edg{\cdots}^+$ does not satisfy Assumption \ref{as3}, we regularize the objective function using the method in \cite{peng1999non}:
\eqn{f(t,\boldsymbol{\pi},a,\boldsymbol{R})
&=t\ln\left(\exp\left(\frac{-\boldsymbol{\pi}^{\intercal}\boldsymbol{R}-a}{t}\right)+1\right) .}
Note that $f(t, \boldsymbol{\pi},a,\boldsymbol{R})\to [-\boldsymbol{\pi}^{\intercal}\boldsymbol{R}-a]^+$ as $t\to0$.
For $t<<1$, assume 
\begin{assumption}
\label{as6}
Given a target value $\rho$, define the smooth non-robust mean-CVaR problem 
\eq{\label{smoothing}&\min_{\boldsymbol{\pi},a}\; a+\frac{1}{\alpha}E_{P^*}\left[t\ln\left(\exp\Big(\frac{-\boldsymbol{\pi}^{\intercal}\boldsymbol{R}-a}{t}\Big)+1\right)\right],\\
&\text{s.t.} \;E_{P^*}[\boldsymbol{\pi}^{\intercal}\boldsymbol{R}]=\rho,\\
&\;\;\quad\boldsymbol{\pi}^\intercal\boldsymbol{1}=1.
}
It has unique solutions $\boldsymbol{\pi}^*$ and $a^*$. 
\end{assumption}
According to Assumption \ref{as6} we have the optimal solutions $\boldsymbol{\pi}^*$ and $a^*$. Therefore there exist Lagrange multipliers $\lambda_1^*$ and $\lambda_2^*$ such that
\begin{align}
\nonumber&\scL(\boldsymbol{\pi}^*,a^*,\lambda_1^*,\lambda_2^*)=a^*+\frac{1}{\alpha}E_{P^*}\left[t\ln\left(exp\Big(\frac{-(\boldsymbol{\pi}^*)^{\intercal}\boldsymbol{R}-a^*}{t}\Big)+1\right)\right]-\lambda_1^*(E_{P^*}[(\boldsymbol{\pi}^*)^\intercal\boldsymbol{R}]-\rho)\\
&\qquad\qquad\qquad\qquad-\lambda_2^*((\boldsymbol{\pi}^*)^\intercal\boldsymbol{1}-1)\nonumber,\\
&\label{lagrange_1} D_{\boldsymbol{\pi}^*}\scL(\boldsymbol{\pi}^*,a^*,\lambda_1^*,\lambda_2^*)=E_{P^*}\left[\frac{-\boldsymbol{R}}{\alpha[1+exp(((\boldsymbol{\pi}^*)^\intercal\boldsymbol{R}+a^*)/t)]}\right]-\lambda_1^*E_{P^*}[\boldsymbol{R}]-\lambda_2^*\boldsymbol{1}=\boldsymbol{0},\\
&\nonumber D_{a^*}\scL(\boldsymbol{\pi}^*,a^*,\lambda_1^*,\lambda_2^*)= E_{P^*}\left[1-\frac{1}{\alpha[1+exp(((\boldsymbol{\pi}^*)^\intercal\boldsymbol{R}+a^*)/t)]}\right]=0,\\
&\nonumber D_{\lambda_1^*}\scL(\boldsymbol{\pi}^*,a^*,\lambda_1^*,\lambda_2^*)= (\boldsymbol{\pi}^*)^\intercal E_{P^*}[\boldsymbol{R}]-\rho = 0,\\
&\nonumber D_{\lambda_2^*}\scL(\boldsymbol{\pi}^*,a^*,\lambda_1^*,\lambda_2^*)= (\boldsymbol{\pi}^*)^\intercal\boldsymbol{1}-1 = 0.
\end{align}
Multiplying (\ref{lagrange_1}) by $(\boldsymbol{\pi}^*)^\intercal$ and substituting $\rho$ for $(\boldsymbol{\pi}^*)^\intercal E_{P^*}[\boldsymbol{R}]$, we obtain
\eqn{\lambda_2^*= (\boldsymbol{\pi}^*)^\intercal E_{P^*}[g(\boldsymbol{\pi}^*,a^*,\boldsymbol{R})]-\lambda_1^*\rho,
}
where
\eqn{g(\boldsymbol{\pi}^*,a^*,\boldsymbol{R})=\frac{-\boldsymbol{R}}{\alpha[1+\exp(((\boldsymbol{\pi}^*)^\intercal\boldsymbol{R}+a^*)/t)]}.
}
Then we substitute $\lambda_2^*$ into (\ref{lagrange_1}). For each $i$, we get
\eqn{\lambda_1^*= \frac{E_{P^*}[g^i(\boldsymbol{\pi}^*,a^*,\boldsymbol{R})]-(\boldsymbol{\pi}^*)^\intercal E_{P^*}[ g(\boldsymbol{\pi}^*,a^*,\boldsymbol{R})]}{E_{P^*}[R^i]-\rho}.
}
We have
\eqn{E_{P^*}[h(\boldsymbol{\pi}^*,a^*,\boldsymbol{R})]=\boldsymbol{0},}
where 
\eqn{h(\boldsymbol{\pi}^*,a^*,\boldsymbol{R}) =g(\boldsymbol{\pi}^*,a^*,\boldsymbol{R})-\lambda_1^*\boldsymbol{R}-\lambda_2^*\boldsymbol{1}.
}
Since $Q\Rightarrow P^*$ as the data sample $N$ increases to infinity, it is reasonable to estimate $\boldsymbol{\pi}^*$  and $a^*$ by calculating the optimal solution from the smooth non-robust mean-CVaR problem (\ref{smoothing}) under probability $Q$. There exist Lagrange multipliers $\lambda_1$ and $\lambda_2$ such that
\eqn{E_P[h(\boldsymbol{\pi},a,\boldsymbol{R})]=E_P[g(\boldsymbol{\pi},a,\boldsymbol{R})-\lambda_1\boldsymbol{R}-\lambda_2\boldsymbol{1}]=\boldsymbol{0}.
}
Then we can define the RWP function
\eq{\label{R_N}\scR_N(\kappa)=\min\{C^{1/\kappa}(P,Q):E_P[h(\boldsymbol{\pi},a,\boldsymbol{R})]=\boldsymbol{0}\}.}
Our choice for $\delta$ is
\eqn{
	\delta^* :=\min\crl{\delta:P^*\left(\scR_N(\kappa)\leq\delta\right)\geq 1-\delta_0},
}
where the confidence level $1-\delta_0$ is chosen by the investor. 

Assumptions \ref{as4} and \ref{as5} also hold for $h(\boldsymbol{\pi}^*,a^*,\boldsymbol{R})$ (proof see Section \ref{prooftheorem2}). As all the conditions required in the selection of the Wasserstein ambiguity set's size (Theorem \ref{theorem1}), then we obtain the asymptotic upper bound of $\scR_N(\kappa)$ (\ref{R_N}) for the robust mean-CVaR models \ref{RMC-1} and \ref{RMC-2}.
\begin{theorem}
\label{theorem2}
Under Assumptions \ref{as6} and Theorem \ref{theorem1}, as $N\rightarrow \infty$, we have
\eqn{N^{\kappa/2}\scR_N(\kappa)\Rightarrow \widebar{\scR}(\kappa),}
where, for $\kappa >1$,
\eqn{\widebar{\scR}(\kappa)=\max_{\xi\in\bbR^n}\crl{\kappa\xi^\intercal Z-(\kappa-1)E_{P^*}\|\xi^\intercal D_{\boldsymbol{R}}h(\boldsymbol{\pi}^*,a^*,\boldsymbol{R})\|_2^{\kappa/(\kappa-1)}}}
and
\eqn{D_{\boldsymbol{R}}h(\boldsymbol{\pi}^*,a^*,\boldsymbol{R}) &= \frac{-I_n}{\alpha[1+exp(((\boldsymbol{\pi}^*)^\intercal\boldsymbol{R}+a^*)/t)]}\\
&+ \frac{\boldsymbol{R}(\boldsymbol{\pi}^*)^\intercal}{\alpha t[1+exp(((\boldsymbol{\pi}^*)^\intercal\boldsymbol{R}+a^*)/t)][1+exp((-(\boldsymbol{\pi}^*)^\intercal\boldsymbol{R}-a^*)/t)]}-\lambda_1^*I_n.
}
If $\kappa =1$,
\eqn{\widebar{\scR}(1):=\max_{\xi\in \Xi}\crl{\xi^\intercal Z},}
where 
\eqn{Z\sim \mathcal{N}(\boldsymbol{0},E_{P^*}[h(\boldsymbol{\pi}^*,a^*,\boldsymbol{R})h(\boldsymbol{\pi}^*,a^*,\boldsymbol{R})^\intercal]),
}
and 
\eqn{\Xi=\{\xi\in\bbR^n:\|\xi^\intercal D_{\boldsymbol{R}}h(\boldsymbol{\pi}^*,a^*,\boldsymbol{R})\|\leq 1 \}.
}
Moreover, $\widebar{\scR}(1)$ has the upper bound:
\eqn{\widebar{\scR}(1):=\sup_{\xi\in \Xi}\xi^\intercal Z \leq \sup_{\xi:\|\xi\|_2\leq 1}\xi^\intercal Z:=\|\Tilde{Z}\|_2,
}
where \eqn{\Tilde{Z}\sim\mathcal{N}\left(\boldsymbol{0},E_{P^*}\left[\left(\Big(\frac{1}{\alpha}+|\hat{\lambda}_1^*|\Big)|\boldsymbol{R}|+|\hat{\lambda}_2^*|\boldsymbol{1}\right)\left(\Big(\frac{1}{\alpha}+|\hat{\lambda}_1^*|\Big)|\boldsymbol{R}|+|\hat{\lambda}_2^*|\boldsymbol{1}\right)^\intercal\right]\right)
}
and $\hat{\lambda}_1^* =\lim_{t\rightarrow0^+}\lambda_1^*$, $\hat{\lambda}_2^*=\lim_{t\rightarrow0^+}\lambda_2^*$.

If $\kappa=2$ and $E_{P^*}[D_{\boldsymbol{R}}h(\boldsymbol{\pi}^*,a^*,\boldsymbol{R})]$ is invertible,
\eqn{\widebar{\scR}(2)= \max_{\xi\in\bbR^n}\crl{2\xi^\intercal Z-\xi^\intercal E_{P^*}[ D_{\boldsymbol{R}}h(\boldsymbol{\pi}^*,a^*,\boldsymbol{R})]\xi}=Z^\intercal(E_{P^*}[ D_{\boldsymbol{R}}h(\boldsymbol{\pi}^*,a^*,\boldsymbol{R})])^{-1}Z.}
Let $A=(\boldsymbol{\pi}^*)^\intercal\boldsymbol{R}+a^*$. As $t\rightarrow0^+$, we can get the upper bound:
\eqn{ \hat{\scR}_N(2)=\lim_{t\rightarrow0^+}\widebar{\scR}(2) \leq \left(E_{P^*}\left[-\hat{\lambda}_1^*\mathbf{1}_{\crl{A>0}}-\frac{1+\alpha\hat{\lambda}_1^*}{\alpha}\mathbf{1}_{\crl{A<0}}\right]\right)^{-1}\Tilde{Z}^\intercal\Tilde{Z}.
}
\end{theorem}
\begin{proof}
See Section \ref{prooftheorem2}.
\end{proof}

\section{Implementation of investment strategy}
\label{section4}
In practice, it is not plausible to rebalance the portfolio daily because of the transaction costs. Therefore, we implement the strategy in the following way.
At time $0$, we invest $\pi_1^*$ and
let it evolve until
\begin{align}\label{timetorebalance}
	\max_{i=1,...,n}\left|\displaystyle\frac{(\pi^i_1)^*(1+R^i_1)(1+R^i_2)\cdots(1+R^i_{t-1})-\pi^i_t}{(\pi^i_1)^*(1+R^i_1)(1+R^i_2)\cdots(1+R^i_{t-1})}\right|> 0.05.
\end{align}
When \eqref{timetorebalance} happens, we rebalance our portfolio to align with our target investment strategy $\pi$. Again, we wait until the $\pi^*$ deviates from $\pi$ more than 5\% and if it does, we rebalance our portfolio. We repeat the scheme until $N$. See the following algorithm.
\begin{itemize}
\item[]{\bf Algorithm}
	\begin{itemize}
		\item[(i)] We let $(\pi^*_t)_{t=1,2,...,N}$ be our investment strategy from each mean-CVaR model. At time 0, we invest $\pi_1^* $and $\tau_0=1$.
		\item[(ii)] For $k\in\bbN$, define stopping times
		\begin{align*}
			\tau_k:=\min\crl{\tau_{k-1}<t\leq T:\max_{i=1,...,n}\left|\displaystyle\frac{(\pi^i_{\tau_{k-1}})^*(1+R^i_1)(1+R^i_2)\cdots(1+R^i_{t-1})-\pi^i_t}{(\pi^i_1)^*(1+R^i_1)(1+R^i_2)\cdots(1+R^i_{t-1})}\right|> 0.05},
		\end{align*}
	and let
	\begin{align*}
		(\pi^i_t)^*:=
		\begin{cases}
			(\pi^i_{\tau_{k-1}})^*(1+R^i_{\tau_{k-1}})(1+R^i_{\tau_{k-1}+1})\cdots(1+R^i_{t-1}),& \text{if } t\in(\tau_{k-1},\tau_k)\\
			\pi_{\tau_k} &\text{if } t= \tau_k
		\end{cases}.
	\end{align*}
	\end{itemize}
\end{itemize}
\begin{Remark}
	Note that the transaction costs of the above algorithm for time $[1,T]$ is not optimised. However, since $N\gg T$, the investments $\pi^*_t$ for $t<T$ have virtually no effect on the performance of investment overall. 
\end{Remark}

When we assume there exist transaction costs, we use the linear transaction costs of 0.2\%. More precisely, our transaction costs is given by
\begin{align*}
	TC=0.002\sum_{k\in\bbN,\tau_k\leq T}\sum^n_{i=1}\left|(\pi^i_{\tau_k-1})^*(1+R_{\tau_k-1})-(\pi^i_{\tau_k})^*\right|.
\end{align*}

\section{Empirical studies}
\label{section5}
In this section, we consider the $L^2$-norm and our robust mean-CVaR models \ref{RMC-1} and \ref{RMC-2}. Firstly, we analyze the convergence of the Wasserstein ambiguity set's size as the in-sample size increases. Then we compare our robust mean-CVaR model's performance with those of well-known strategies.

The data provided by Bloomberg Terminal is used for the empirical study. We select the 100 largest S\&P 500 companies by market cap at three different time points, Feb 1, 2002, June 1, 2004, June 1, 2006, Aug 1, 2008, and June 1, 2009. For each time point, we use the daily adjusted closing price data for the previous 2 years to estimate our parameters and simulate the strategies for the following 8 years. We assume the net investment of each period is a fixed constant $\$1$. We compare the portfolio allocation, the Sharpe ratio, the mean/CVaR, and the cumulative wealth of our robust mean-CVaR models and other conventional methods with and without transaction costs.
\subsection{Influence of the smooth parameter}
To select an appropriate size of the Wasserstein ambiguity set, we transform the non-robust mean-CVaR (\ref{nonsmoothing}, NMC) into a smooth problem (\ref{smoothing}, Smooth). We test the influence of the smooth parameter $t$ updating from $10^{-5}$ to $10^{-1}$ in the smooth problem compared with the non-robust problem.
\FigInsert{fig1}
Figure \ref{fig1} showed as $t$ goes to small (less than $10^{-4}$), smooth mean-CVaR and NMC problems almost have the same optimal values. Thus we can choose smooth mean-CVaR instead of NMC and forward this smooth problem to the calculation of the Wasserstein ambiguity set's size. 

\subsection{Convergence of the Wasserstein ambiguity set's size } 

Next, we will analyze the convergence of the Wasserstein ambiguity set's size as the in-sample size increases. The simulation sample of $\scR_N(\kappa)$ is set to 10000. We estimate $0.95$ quantile of $\scR_N(\kappa)$ and denote the estimate of the quantile by $\eta_{0.95}$. Then our Wasserstein ambiguity set's size is $\delta=\eta_{0.95}/N^{\kappa/2}$.
\FigInsert{fig2}
\FigInsert{fig3}

It is clear in Figure \ref{fig2} and \ref{fig3} that $\delta$ decreases as the in-sample size increases. If the in-sample size reaches infinity, the uncertainty size $\delta$ will be close to 0. We know the model will have too much ambiguity and the available data will be useless when choosing a large $\delta$, and the utility of robustness will be disappeared if we choose a very small $\delta$. Therefore, for later simulations, we choose the past two years as in-sample size to get corresponding parameters.

\subsection{Comparison of our model with conventional methods} 
In this section, we will compare our robust mean-CVaR model with other conventional methods like classical sample base mean-CVaR (\cite{rockafellar2000optimization} and \cite{rockafellar2002conditional}), the robust mean-CVaR under box uncertainty (\cite{zhu2009worst}), and the robust mean-CVaR under distribution ambiguity (\cite{kang2019data}).

The classical mean-CVaR portfolio selection problem is 
\eqn{\min_{\boldsymbol{\pi},a}\crl{a +\displaystyle\frac{1}{\alpha}E_P[-\boldsymbol{\pi}^\intercal \boldsymbol{R}-a]^+}\quad s.t\quad E_{P}(\boldsymbol{\pi}^\intercal \boldsymbol{R})\geq \rho.}
In reality, the investor does not know the underlying distribution $P$. If we use the empirical probability $Q$ instead of $P$ and only a collection of $N$ historical observations $\crl{R_1},...,\crl{R_N}$ of $\boldsymbol{R}$ is available, then the classical mean-CVaR problem can be recast as the the following sample based non-robust mean-CVaR optimization problem
\begin{equation}\tag{NMC}\begin{aligned}
\min_{\pi,a}\;\left\{a +\frac{1}{\alpha N}\sum_{i=1}^N\left[-\boldsymbol{\pi}^\intercal R_i-a\right]^+\right\}\quad s.t\quad \frac{1}{N}\sum_{i=1}^N(\boldsymbol{\pi}^\intercal \boldsymbol{R}_i)\geq \rho.
\end{aligned}\end{equation}
\cite{zhu2009worst} assumed random return vector $\boldsymbol{R}$ follows a discrete distribution and the sample space of $\boldsymbol{R}$ is given by $\crl{R_1,...,R_N}$. Let $p_k$ denote the probability according to $k$-$th$ sample and $\sum_{k=1}^N p_k = 1$, $p_k\geq 0$, $k=1,...,N$. Denote $\boldsymbol{p} = (p_1,...,p_N)^\intercal$.
Then $E_P(\cdot) = \boldsymbol{p}^\intercal(\cdot)$.

Suppose that $\boldsymbol{p}$ belongs to a box, i.e.,
\eqn{\boldsymbol{p}\in\scU_B = \crl{\boldsymbol{p}: \boldsymbol{p} = \boldsymbol{p}^0+\boldsymbol{\eta}, \boldsymbol{1}^\intercal\boldsymbol{\eta}=0,  \underline{\boldsymbol{\eta}}\leq \boldsymbol{\eta}\leq\overline{\boldsymbol{\eta}}},
}
where $\boldsymbol{p}^0$ is a nominal distribution that represents the most
likely distribution, $\underline{\boldsymbol{\eta}}$ and $\overline{\boldsymbol{\eta}}$ are given constant vectors. They consider the following robust mean-CVaR model under box uncertainty
\eq{\label{eq6.16}\min_{\boldsymbol{\pi}}\max_{\boldsymbol{\theta}\in\scU_B}\crl{a +\displaystyle\frac{1}{\alpha}E_P[-\boldsymbol{\pi}^\intercal \boldsymbol{R}-a]^+}
\quad s.t\quad\min_{\boldsymbol{\theta}\in\scU_B}\;E_{P}(\boldsymbol{\pi^\intercal \boldsymbol{R}})\geq \rho.}
Then they get the following dual minimization problem of (\ref{eq6.16}) with variables $(\boldsymbol{\pi},\boldsymbol{u},a,\zeta,z,\boldsymbol{\xi},\boldsymbol{\omega},\delta,$\\$\boldsymbol{\tau},\boldsymbol{\nu})\in \bbR^n\times\bbR^N\times\bbR\times\bbR\times\bbR\times\bbR^N\times\bbR^N\times\bbR\times\bbR^N\times\bbR^N$,
\begin{equation}\tag{BMC}\begin{aligned}&\min\;\zeta\\
&\text{s.t.}\;a+\frac{1}{\alpha}(\boldsymbol{\theta}^0)^\intercal \boldsymbol{u} + \frac{1}{\alpha}(\overline{\boldsymbol{\eta}}^\intercal\boldsymbol{\xi}+\underline{\boldsymbol{\eta}}^\intercal\boldsymbol{\omega})\leq \zeta\;\nonumber\\
& \qquad \boldsymbol{1}z +\boldsymbol{\xi}+\boldsymbol{\omega}=\boldsymbol{u} \nonumber\\
& \qquad \boldsymbol{\xi}\geq 0,\;\boldsymbol{\omega}\leq 0 \nonumber\\
& \qquad u_i \geq -\boldsymbol{\pi}^\intercal\boldsymbol{R}_i -a, i=1,...,N\nonumber\\
&\qquad  u_i \geq 0, i=1,...,N,\nonumber\\\nonumber
&\qquad \boldsymbol{\theta}\in\scU_B,
\end{aligned}\end{equation}
where
\eqn{\scU_B = \{(\boldsymbol{\pi},\delta,\boldsymbol{\tau},\boldsymbol{\nu}):&\boldsymbol{\pi}^\intercal\boldsymbol{1}=1, 0\leq\pi_i\leq 1, \boldsymbol{1}\delta+\boldsymbol{\tau}+\boldsymbol{\nu}=\boldsymbol{\pi}^\intercal\boldsymbol{R},\\
&\boldsymbol{\tau}\leq 0,\boldsymbol{\nu}\geq 0, (\boldsymbol{\theta}^0)^\intercal(\boldsymbol{\pi}^\intercal\boldsymbol{R})+\overline{\boldsymbol{\eta}}^\intercal\boldsymbol{\tau}+\underline{\boldsymbol{\eta}}^\intercal\boldsymbol{\nu}\geq \rho \}.}

\cite{kang2019data} let $E_P(\boldsymbol{R})=\mu$, $E_Q(\boldsymbol{R})=\hat{\mu}$, $E_P(\boldsymbol{RR}^\intercal )=\Sigma$, and $E_Q(\boldsymbol{RR}^\intercal )=\hat{\Sigma}$. They consider the ambiguity sets defined as 
\begin{align}
\scU(\gamma_1,\gamma_2)=
\left\{P\in\scM_+:\begin{array}{l}
       P(\boldsymbol{R}\in\Omega)=1,\\
(\mu-\hat{\mu})^\intercal\hat{\Sigma}^{-1}(\mu-\hat{\mu})\leq \gamma_1,\\
\|\Sigma-\hat{\Sigma}\|_2\leq \gamma_2,\;\Sigma \succ 0
\end{array}
\right\}.
\end{align}
They consider the following robust mean-CVaR model under distribution ambiguity,
\eq{\label{eq6.19}\min_{\boldsymbol{\pi}}\max_{P\in\scU(\gamma_1,\gamma_2)}\crl{a +\displaystyle\frac{1}{\alpha}E_P[-\boldsymbol{\pi}^\intercal \boldsymbol{R}-a]^+}\quad s.t\quad\min_{P\in\scU(\gamma_1,\gamma_2)}\;E_{P}(\boldsymbol{\pi^\intercal \boldsymbol{R}})\geq \rho.}
By introducing auxiliary variables $j$, $k\in\bbR$, they rewrite
the optimization problem (\ref{eq6.19}) as 
\begin{equation}\tag{KMC}\begin{aligned}
&\min_{\pi,j,k}\;-\hat{\mu}^\intercal \boldsymbol{\pi}+\sqrt{\gamma_1}k+\hat{\alpha} j,\\
&\text{s.t.}\;\sqrt{\gamma_1}\sqrt{\boldsymbol{\pi}^\intercal\hat{\Sigma}\boldsymbol{\pi}}\leq \hat{\mu}^\intercal\boldsymbol{\pi}-\rho,\\
&\quad\sqrt{\boldsymbol{\pi}^\intercal(\hat{\Sigma}+\gamma_2\boldsymbol{I}_n)\boldsymbol{\pi}}\leq j,\\
&\quad\sqrt{\boldsymbol{\pi}^\intercal\hat{\Sigma}\boldsymbol{\pi}}\leq k,
\end{aligned}\end{equation}
where $\displaystyle\hat{\alpha}=\sqrt{\frac{\alpha}{1-\alpha}}$.

They use the bootstrap to construct samples for computing the optimal parameters $\gamma_1$ and $\gamma_2$ (See \cite{kang2019data}, Section 3).

\subsubsection*{Simulation Results without transaction costs}
Table \ref{table1} shows the result of the optimal investment without transaction costs for the five methods.
\TabInsert{table1}
The following plots are the compositions of portfolios for each strategy without transaction costs, starting from different dates.
\FigInsert{fig4}
\FigInsert{fig5}
\FigInsert{fig6}
\FigInsert{fig7}
\FigInsert{fig8}
We assume the initial investment amount is \$1 and each strategy is self-financing. We get the cumulative wealth as follows. 
\FigInsert{fig9}
\FigInsert{fig10}
\FigInsert{fig11}
\FigInsert{fig12}
\FigInsert{fig13}
Cumulative wealth is not clear to reflect its variance. We will consider the rolling one-year Sharpe ratio, which is a good way to analyze the historical performance of an investment. This makes it easy to observe the change and robustness throughout the time interval. It tells investors whether the strategy tends to underperform from time to time or whether it it will always perform well. The following figures are the rolling 1-year Sharpe ratios of each strategy without transaction costs, starting from different dates.
\FigInsert{fig14}
\FigInsert{fig15}
\FigInsert{fig16}
\FigInsert{fig17}
\FigInsert{fig18}

\subsubsection*{Simulation Results with transaction costs}
In this subsection, we consider the transaction costs for each strategy. Table \ref{table2} is the comparison between the investment strategies with the transaction costs.
\TabInsert{table2}
The following plots are the compositions of portfolios for each strategy with transaction costs, starting from different dates.
\FigInsert{fig19}
\FigInsert{fig20}
\FigInsert{fig21}
\FigInsert{fig22}
\FigInsert{fig23}

The following plots are the cumulative wealth and maximum drawdown of each strategy with transaction costs. 
\FigInsert{fig24}
\FigInsert{fig25}
\FigInsert{fig26}
\FigInsert{fig27}
\FigInsert{fig28}

The following plots are the rolling 1-year Sharpe ratios of each strategy with transaction costs, starting from different dates.
\FigInsert{fig29}
\FigInsert{fig30}
\FigInsert{fig31}
\FigInsert{fig32}
\FigInsert{fig33}

\subsection{Simulation conclusion}
In our simulations, RMC-1 and RMC-2 models performed competitively compared to NMC, BMC, and KMC models. The simulation results are similar whether there are transaction costs or not.

For the two simulations starting from 2002.02.01 and 2009.06.01, our RMC-1 and RMC-2 models showed the best Sharpe ratio, Mean/CVaR, cumulative wealth, and rolling one-year Sharpe ratio among all strategies. For the simulation starting from 2006.06.01, the RMC-1 model showed the best performance, but the RMC-2 model showed moderate Mean/CVaR and consistently outperformed the KMC method. For the simulation starting from 2004.06.01, the RMC-1 model showed the best Sharpe ratio and the second best Mean/CVaR, but the RMC-2 model showed the worst performance. For the simulation starting from 2008.08.01, RMC-1 and RMC-2 models showed the worst Sharpe ratios and moderate Mean/CVaR outperformed the KMC strategy. Depending on the in-sample data we selected, it is possible that our robust models perform worse than other methods. However, it is worth noting that both RMC-1 and RMC-2 models showed the most well-diversified portfolio allocation among all strategies.

While there is a limitation that we simulate these strategies only for five time frames, our empirical result provides a promising performance of RMC-1 and RMC-2 we suggested.

\section{Conclusion}
\label{section6}
In this article, we formulate a framework to solve distributionally robust optimization problem. We allow the true probability measure to be inside a Wasserstein ball that is specified by the empirical data and the given confidence level. We transform the robust optimization into a non-robust minimization with a penalty term. We provide a method to select an appropriate size of the Wasserstein ambiguity set. Then we apply our RWO framework to the robust mean-CVaR optimization problem. Moreover, its numerical simulations of the US stock market provide a promising result compared to other popular strategies.
\section*{Acknowledgements}
The authors wish to thank the reviewers and the editor for
providing feedback which can help improve this paper.
\section*{Data availability statement}
The data that support the findings of this study are available from Bloomberg.
Restrictions apply to the availability of these data, which were used under license for this study. Data are available from the authors with the permission of Bloomberg.

\bibliographystyle{abbrvnat}
\bibliography{main.bib}

\begin{thebibliography}{20}
\providecommand{\natexlab}[1]{#1}
\providecommand{\url}[1]{\texttt{#1}}
\expandafter\ifx\csname urlstyle\endcsname\relax
  \providecommand{\doi}[1]{doi: #1}\else
  \providecommand{\doi}{doi: \begingroup \urlstyle{rm}\Url}\fi

\bibitem[Artzner et~al.(1999)Artzner, Delbaen, Eber, and
  Heath]{artzner1999coherent}
P.~Artzner, F.~Delbaen, J.-M. Eber, and D.~Heath.
\newblock Coherent measures of risk.
\newblock \emph{Mathematical finance}, 9\penalty0 (3):\penalty0 203--228, 1999.

\bibitem[Blanchet et~al.(2019{\natexlab{a}})Blanchet, Kang, and
  Murthy]{blanchet2019robust}
J.~Blanchet, Y.~Kang, and K.~Murthy.
\newblock Robust wasserstein profile inference and applications to machine
  learning.
\newblock \emph{Journal of Applied Probability}, 56\penalty0 (3):\penalty0
  830--857, 2019{\natexlab{a}}.

\bibitem[Blanchet et~al.(2019{\natexlab{b}})Blanchet, Kang, Murthy, and
  Zhang]{blanchet2019data}
J.~Blanchet, Y.~Kang, K.~Murthy, and F.~Zhang.
\newblock Data-driven optimal transport cost selection for distributionally
  robust optimization.
\newblock In \emph{2019 winter simulation conference (WSC)}, pages 3740--3751.
  IEEE, 2019{\natexlab{b}}.

\bibitem[Delage and Ye(2010)]{delage2010distributionally}
E.~Delage and Y.~Ye.
\newblock Distributionally robust optimization under moment uncertainty with
  application to data-driven problems.
\newblock \emph{Operations research}, 58\penalty0 (3):\penalty0 595--612, 2010.

\bibitem[Esfahani and Kuhn(2018)]{esfahani2018data}
P.~M. Esfahani and D.~Kuhn.
\newblock Data-driven distributionally robust optimization using the
  wasserstein metric: Performance guarantees and tractable reformulations.
\newblock \emph{Mathematical Programming}, 171\penalty0 (1):\penalty0 115--166,
  2018.

\bibitem[Fabozzi et~al.(2007)Fabozzi, Kolm, Pachamanova, and
  Focardi]{fabozzi2007robust}
F.~J. Fabozzi, P.~N. Kolm, D.~A. Pachamanova, and S.~M. Focardi.
\newblock \emph{Robust portfolio optimization and management}.
\newblock John Wiley \& Sons, 2007.

\bibitem[Fournier and Guillin(2015)]{fournier2015rate}
N.~Fournier and A.~Guillin.
\newblock On the rate of convergence in wasserstein distance of the empirical
  measure.
\newblock \emph{Probability Theory and Related Fields}, 162\penalty0
  (3):\penalty0 707--738, 2015.

\bibitem[Hansen and Sargent(2001)]{hansen2001robust}
L.~Hansen and T.~J. Sargent.
\newblock Robust control and model uncertainty.
\newblock \emph{American Economic Review}, 91\penalty0 (2):\penalty0 60--66,
  2001.

\bibitem[Hota et~al.(2018)Hota, Cherukuri, and Lygeros]{hota2018}
A.~R. Hota, A.~Cherukuri, and J.~Lygeros.
\newblock Data-driven chance constrained optimization under wasserstein
  ambiguity sets.
\newblock \emph{arXiv}, 2018.
\newblock URL \url{https://arxiv.org/abs/1805.06729}.

\bibitem[Hota et~al.(2019)Hota, Cherukuri, and Lygeros]{hota2019data}
A.~R. Hota, A.~Cherukuri, and J.~Lygeros.
\newblock Data-driven chance constrained optimization under wasserstein
  ambiguity sets.
\newblock In \emph{2019 American Control Conference (ACC)}, pages 1501--1506.
  IEEE, 2019.

\bibitem[Kang et~al.(2019)Kang, Li, Li, and Zhu]{kang2019data}
Z.~Kang, X.~Li, Z.~Li, and S.~Zhu.
\newblock Data-driven robust mean-cvar portfolio selection under distribution
  ambiguity.
\newblock \emph{Quantitative Finance}, 19\penalty0 (1):\penalty0 105--121,
  2019.

\bibitem[Lobo and Boyd(2000)]{lobo2000worst}
M.~S. Lobo and S.~Boyd.
\newblock The worst-case risk of a portfolio.
\newblock \emph{Unpublished manuscript. Available from
  \url{https://web.stanford.edu/~boyd/papers/pdf/risk_bnd.pdf}}, 2000.

\bibitem[Peng and Lin(1999)]{peng1999non}
J.-M. Peng and Z.~Lin.
\newblock A non-interior continuation method for generalized linear
  complementarity problems.
\newblock \emph{Mathematical Programming}, 86\penalty0 (3):\penalty0 533--563,
  1999.

\bibitem[Pichler and Xu(2018)]{pichler2018quantitative}
A.~Pichler and H.~Xu.
\newblock Quantitative stability analysis for minimax distributionally robust
  risk optimization.
\newblock \emph{Mathematical Programming}, pages 1--31, 2018.

\bibitem[Rockafellar and Uryasev(2002)]{rockafellar2002conditional}
R.~T. Rockafellar and S.~Uryasev.
\newblock Conditional value-at-risk for general loss distributions.
\newblock \emph{Journal of banking \& finance}, 26\penalty0 (7):\penalty0
  1443--1471, 2002.

\bibitem[Rockafellar et~al.(2000)Rockafellar, Uryasev,
  et~al.]{rockafellar2000optimization}
R.~T. Rockafellar, S.~Uryasev, et~al.
\newblock Optimization of conditional value-at-risk.
\newblock \emph{Journal of risk}, 2:\penalty0 21--42, 2000.

\bibitem[Tong et~al.(2008)Tong, Wu, and Qi]{tong2008available}
X.~Tong, F.~F. Wu, and L.~Qi.
\newblock Available transfer capability calculation using a smoothing pointwise
  maximum function.
\newblock \emph{IEEE Transactions on Circuits and Systems I: Regular Papers},
  55\penalty0 (1):\penalty0 462--474, 2008.

\bibitem[T{\"u}t{\"u}nc{\"u} and Koenig(2004)]{tutuncu2004robust}
R.~H. T{\"u}t{\"u}nc{\"u} and M.~Koenig.
\newblock Robust asset allocation.
\newblock \emph{Annals of Operations Research}, 132\penalty0 (1):\penalty0
  157--187, 2004.

\bibitem[Villani(2009)]{villani2009optimal}
C.~Villani.
\newblock \emph{Optimal transport: old and new}, volume 338.
\newblock Springer, 2009.

\bibitem[Zhu and Fukushima(2009)]{zhu2009worst}
S.~Zhu and M.~Fukushima.
\newblock Worst-case conditional value-at-risk with application to robust
  portfolio management.
\newblock \emph{Operations research}, 57\penalty0 (5):\penalty0 1155--1168,
  2009.

\end{thebibliography}

\begin{appendices}
\section{}
\subsection{}
\begin{proof}[Proof of Proposition \ref{proposition1}]
\label{proofproposition1}
We consider the following constraint problem
\eq{\label{eqA1.1}\min_{P\in\scU_\delta(Q)}E_P[g(\boldsymbol{x},\boldsymbol{\theta})]}
equivalently,
\eq{\label{eqA1.2}-\max_{P\in \scU_\delta(Q)}E_P[-g(\boldsymbol{x},\boldsymbol{\theta})].}
According to Proposition 1 in \cite{blanchet2019robust}, we obtain the dual problem
\eq{\label{eq1.3}\max_{P\in \scU_\delta(Q)}E_P[-g(\boldsymbol{x},\boldsymbol{\theta})]=\inf_{\gamma_0\geq0}\left[\gamma_0\delta+\frac{1}{N}\sum_{i=1}^N\Phi_{\gamma_0}(\theta_i)\right],}
where \eq{\label{eq1.4}\Phi_{\gamma_0}(\theta_i)&=\sup_{\boldsymbol{\theta}}\{-g(\boldsymbol{x},\boldsymbol{\theta})-\gamma_0 c(\boldsymbol{\theta},\theta_i)\}\nonumber\\
&=\sup_{\boldsymbol{\theta}}\{-g(\boldsymbol{x},\boldsymbol{\theta})-\gamma_0\|\boldsymbol{\theta}-\theta_i\|_2^\kappa\}.}
Therefore, (\ref{eq1.3}) becomes
\eq{\label{eq1.5}\max_{P\in \scU_\delta(Q)}E_P[-g(\boldsymbol{x},\boldsymbol{\theta})]=\inf_{\gamma_0\geq0}\left\{\gamma_0\delta+\frac{1}{N}\sum_{i=1}^N\brak{\sup_{\boldsymbol{\theta}}\crl{-g(\boldsymbol{x},\boldsymbol{\theta})-\gamma_0\|\boldsymbol{\theta}-\theta_i\|_2^\kappa}}\right\}
}
or
\eq{\label{eq1.6}\min_{P\in \scU_\delta(Q)}E_P[g(\boldsymbol{x},\boldsymbol{\theta})]=-\inf_{\gamma_0\geq0}\left\{\gamma_0\delta+\frac{1}{N}\sum_{i=1}^N\brak{\sup_{\boldsymbol{\theta}}\crl{-g(\boldsymbol{x},\boldsymbol{\theta})-\gamma_0\|\boldsymbol{\theta}-\theta_i\|_2^\kappa}}\right\}.
}
Therefore, the feasible region can be rewritten as
\eqn{\scF_{(\delta,\rho) }:=\crl{\boldsymbol{x}:-\inf_{\gamma_0\geq0}\left\{\gamma_0\delta+\frac{1}{N}\sum_{i=1}^N\brak{\sup_{\boldsymbol{\theta}}\crl{-g(\boldsymbol{x},\boldsymbol{\theta})-\gamma_0\|\boldsymbol{\theta}-\theta_i\|_2^\kappa}}\right\}\geq\rho}.
}
For the primal problem (\ref{eq1}), the inner optimization problem is
\eq{\label{eqA1}\max_{P\in\scU_\delta(Q)}E_P\Big[f(\boldsymbol{x},\boldsymbol{\theta})\Big].}
Based on Proposition 1 in \cite{blanchet2019robust}, the optimal value of (\ref{eqA1}) is
\eqn{\min_{\gamma_1\geq0}\left[\gamma_1\delta+\frac{1}{N}\sum_{i=1}^N\Phi_{\gamma_1}(\theta_i)\right],}
where
\eqn{\Phi_{\gamma_1}(\theta_i)=\sup_{\boldsymbol{\theta}}\Big[f(\boldsymbol{x},\boldsymbol{\theta})-\gamma_1 c(\boldsymbol{\theta},\theta_i)\Big].}
Then the initial optimization problem (\ref{eq1}) becomes
\eq{\label{eqA2}\min_{\boldsymbol{x}\in\scF_{(\delta,\rho)}}\max_{P\in\scU_\delta(Q)}E_P[f(\boldsymbol{x},\boldsymbol{\theta})]&=\min_{\boldsymbol{x}\in\scF_{(\delta,\rho)}}\Big\{\inf_{\gamma_1\geq 0}\Big[\gamma_1\delta+\frac{1}{N}\sum_{i=1}^N\Phi_{\gamma_1}(\theta_i)\Big]\Big\}\\
&=\min_{\boldsymbol{x}\in\scF_{(\delta,\rho)}}\Big\{\inf_{\gamma_1\geq 0}\Big[\gamma_1\delta+\frac{1}{N}\sum_{i=1}^N\sup_{\boldsymbol{\theta}}\crl{f(\boldsymbol{x},\boldsymbol{\theta})-\gamma_1 c(\boldsymbol{\theta},\theta_i)}\Big]\Big\}\\
&=\min_{\boldsymbol{x}\in\scF_{(\delta,\rho)},\gamma_1\geq 0}\crl{ \gamma_1\delta+\frac{1}{N}\sum_{i=1}^N\sup_{\boldsymbol{\theta}}\crl{f(\boldsymbol{x},\boldsymbol{\theta})-\gamma_1\|\boldsymbol{\theta}-\theta_i\|_2^\kappa}}.
}
Therefore, the initial optimization problem (\ref{eq1}) has the following dual problem:
\eqn{& \min_{\gamma_1\geq 0,\boldsymbol{x}}\quad\gamma_1\delta+\frac{1}{N}\sum_{i=1}^N \brak{\sup_{\boldsymbol{\theta}}\crl{f(\boldsymbol{x},\boldsymbol{\theta})-\gamma_1\|\boldsymbol{\theta}-\theta_i\|_2^\kappa}},\\
&\text{s.t.}\quad -\inf_{\gamma_0\geq0}\left\{\gamma_0\delta+\frac{1}{N}\sum_{i=1}^N\brak{\sup_{\boldsymbol{\theta}}\{-g(\boldsymbol{x},\boldsymbol{\theta})-\gamma_0\|\boldsymbol{\theta}-\theta_i\|_2^\kappa\}}\right\}\geq\rho.}
This completes the proof.
\end{proof}

\subsection{}
\begin{proof}[Proof of Proposition \ref{proposition4.1}]
\label{proofA1}
Consider $\kappa=1$. Because the problem
\eqn{\min_{P\in \scU_\delta(Q)}[E_P(\boldsymbol{\pi}^\intercal\boldsymbol{R})]}
is equivalent to
\eqn{-\max_{P\in \scU_\delta(Q)}[E_P((-\boldsymbol{\pi})^\intercal\boldsymbol{R})].}
We  obtain the dual formulation from Proposition \ref{proposition1} that
\eq{\label{A.36}\max_{P\in \scU_\delta(Q)}[E_P((-\boldsymbol{\pi})^\intercal\boldsymbol{R})]=\inf_{\gamma_0^*\geq0}\crl{\gamma_0^*\delta+\frac{1}{N}\sum_{i=1}^N\Phi_{\gamma_0^*}(R_i)},}
where \eqn{\Phi_{\gamma_0^*}(R_i)&=\sup_{\boldsymbol{R}}\{(-\boldsymbol{\pi})^\intercal\boldsymbol{R}-\gamma_0^* c(\boldsymbol{R},R_i)\}.}
For each $i$, we have
\eqn{\Phi_{\gamma_0^*}&=\sup_{\boldsymbol{R}}\{(-\boldsymbol{\pi})^\intercal\boldsymbol{R}-\gamma_0^*\|\boldsymbol{R}-R_i\|_2\}\nonumber\\
&=\sup_{\Delta_i}\{(-\boldsymbol{\pi})^{\intercal}(\Delta_i+R_i)-\gamma_0^*\|\Delta_i\|_2\}\nonumber\\
&=\sup_{\Delta_i}\{(-\boldsymbol{\pi})^{\intercal}\Delta_i-\gamma_0^*\|\Delta_i\|_2\}-\boldsymbol{\pi}^{\intercal}R_i\nonumber\\
&=\sup_{\Delta_i}\{\|\boldsymbol{\pi}\|_2\|\Delta_i\|_2-\gamma_0^*\|\Delta_i\|_2\}-\boldsymbol{\pi}^{\intercal}R_i\nonumber\\
&=\begin{cases} 
-\boldsymbol{\pi}^{\intercal}R_i & \mbox{if}\; \gamma_0^*\geq\|\boldsymbol{\pi}\|_2 
\\ +\infty & \mbox{if}\; \gamma_0^*\leq\|\boldsymbol{\pi}\|_2
\end{cases}}
For the outer minimization (\ref{A.36}), it is sufficient to restrict to $\gamma_0^*\geq\|\boldsymbol{\pi}\|_2$. Then we can get
\eq{\max_{P\in \scU_\delta(Q)}[E_P((-\boldsymbol{\pi})^{\intercal}\boldsymbol{R})]&=\inf_{\gamma_0^*\geq \|\boldsymbol{\pi}\|_2}\crl{\gamma_0^*\delta+\frac{1}{N}\sum_{i=1}^N(-\boldsymbol{\pi}^{\intercal}R_i)}\nonumber\\
&=\delta\|\boldsymbol{\pi}\|_2-E_{Q}[\boldsymbol{\pi}^{\intercal}\boldsymbol{R}]}
or\eq{\min_{P\in \scU_\delta(Q)}[E_P(\boldsymbol{\pi}^{\intercal}\boldsymbol{R})]=E_{Q}[\boldsymbol{\pi}^{\intercal}\boldsymbol{R}]-\delta\|\boldsymbol{\pi}\|_2.}
Therefore, the feasible region can be rewritten as
\eqn{\scF_{(\delta,\rho) }:=\crl{\boldsymbol{\pi}:E_{Q}[\boldsymbol{\pi}^{\intercal}\boldsymbol{R}]-\delta\|\boldsymbol{\pi}\|_2\geq\rho,\;\boldsymbol{\pi}^\intercal\boldsymbol{1}=1}.
}
The optimization problem (\ref{eq4.10}) is 
\eq{\label{A.31}
&\min_{\boldsymbol{\pi}\in\scF_{(\delta,\rho)},a}\max_{P\in\scU_\delta(Q)}\crl{a +\displaystyle\frac{1}{\alpha}E_P[-\boldsymbol{\pi},\boldsymbol{R}-a]^+}\nonumber\\
&=\min_{\boldsymbol{\pi}\in\scF_{(\delta,\rho)},a}\crl{a+\frac{1}{\alpha}\left(\max_{P\in\scU_\delta(Q)}E_P[-\boldsymbol{\pi}^\intercal\boldsymbol{R}-a]^+\right)}.
}
Refer to Proposition \ref{proposition1}, we can get 
\eq{\label{A.32}\max_{P\in\scU_\delta(Q)}E_P[-\boldsymbol{\pi}^\intercal\boldsymbol{R}-a]^+=\inf_{\gamma_1^*\geq0}\crl{\gamma_1^*\delta+\frac{1}{N}\sum_{i=1}^N\Phi_{\gamma_1^*}(R_i)},}
where \eqn{\Phi_{\gamma_1^*}(R_i)&=\sup_{\boldsymbol{R}}\{(-\boldsymbol{\pi}^\intercal\boldsymbol{R}-a)^+-\gamma_1^*(\boldsymbol{R},R_i)\}.}
For each $i$, let us consider the maximization problem and for simplicity we denote $\Delta_i = \boldsymbol{R}-R_i$. Then
\eq{\label{A.33}\Phi_{\gamma_1^*}(R_i)&=\sup_{\boldsymbol{R}}\crl{(-\boldsymbol{\pi}^\intercal\boldsymbol{R}-a)^+ -\gamma_1^*\|\boldsymbol{R}-R_i\|_2}\nonumber\\
&=\sup_{\Delta_i}\crl{(-\boldsymbol{\pi}^\intercal(\Delta_i+R_i)-a)^+ -\gamma_1^*\|\Delta_i\|_2}\nonumber\\
&=\sup_{\Delta_i}\sup_{0\leq \beta_i \leq1}\crl{\beta_i(-\boldsymbol{\pi}^\intercal(\Delta_i+R_i)-a)-\gamma_1^*\|\Delta_i\|_2}\nonumber\\
&=\sup_{0\leq \beta_i \leq1}\sup_{\Delta_i}\crl{-\beta_i\boldsymbol{\pi}^\intercal\Delta_i-\gamma_1^*\|\Delta_i\|_2+\beta_i(-\boldsymbol{\pi}^\intercal R_i-a)}\nonumber\\
&=\sup_{0\leq \beta_i \leq1}\sup_{\Delta_i}\crl{\beta_i\|\boldsymbol{\pi}\|_2\|\Delta_i\|_2-\gamma_1^*\|\Delta_i\|_2+\beta_i(-\boldsymbol{\pi}^\intercal R_i-a)}\nonumber\\
&=\begin{cases} 
(-\boldsymbol{\pi}^\intercal R_i-a)^+ & \mbox{if}\; \gamma_1^*\geq\|\boldsymbol{\pi}\|_2 
\\ +\infty & \mbox{if}\; \gamma_1^*\leq\|\boldsymbol{\pi}\|_2.
\end{cases}}
The first equality follows from the observation that $(r)^+ = \sup_{0\leq \beta \leq 1}\beta r$. Note that the function now is concave in $\Delta_i$ and linear in $\beta$, and $\beta$ is in a compact set. Then we can get the second equality by applying minimax theorem to switch the order of maxima. For the third equality, we use the H$\ddot{o}$lder inequality $|\boldsymbol{\pi}^\intercal\Delta_i|\leq\|\boldsymbol{\pi}\|_2\|\Delta_i\|_2$.
Finally we choose the the restriction $\gamma_1^*\geq\|\boldsymbol{\pi}\|_2$ for the outer minimization.

As a result, the outer minimization problem (\ref{A.32}) becomes 
\eq{&\inf_{\gamma_1^*\geq \|\boldsymbol{\pi}\|_2}\crl{\gamma_1^*\delta+\frac{1}{N}\sum_{i=1}^N(-\boldsymbol{\pi}^\intercal R_i-a)^+}=\frac{1}{N}\sum_{i=1}^N(-\boldsymbol{\pi}^\intercal R_i-a)^+ + \delta\|\boldsymbol{\pi}\|_2.
}
Therefore, the dual problem of (\ref{eq4.10}) is 
\eqn{&\min_{\boldsymbol{\pi},a}\;\crl{a +\displaystyle\frac{1}{\alpha}\left(E_Q[-\boldsymbol{\pi}^\intercal R_i-a]^+ + \delta\|\boldsymbol{\pi}\|_2\right)},\\
&\text{s.t.}\quad E_{Q}[\boldsymbol{\pi}^{\intercal}\boldsymbol{R}]-\delta\|\boldsymbol{\pi}\|_2\geq \rho,\\
&\qquad\;\boldsymbol{\pi}^\intercal \boldsymbol{1}=1.
}
This completes the proof for \ref{RMC-1}.

If $\kappa=2$, then we obtain the dual formulation from Proposition \ref{proposition1} that
\eq{\label{A.39}\max_{P\in \scU_\delta(Q)}[E_P((-\boldsymbol{\pi})^\intercal\boldsymbol{R})]=\inf_{\gamma'\geq0}\crl{\gamma'\delta+\frac{1}{N}\sum_{i=1}^N\Phi_{\gamma'}(R_i)},}
where \eqn{\Phi_{\gamma'}(R_i)&=\sup_{\boldsymbol{R}}\{(-\boldsymbol{\pi})^\intercal\boldsymbol{R}-\gamma' c(\boldsymbol{R},R_i)\}\nonumber\\
&=\sup_{\boldsymbol{R}}\{(-\boldsymbol{\pi})^\intercal\boldsymbol{R}-\gamma'\|\boldsymbol{R}-R_i\|_2^2\}\nonumber\\
&=\sup_\Delta\{(-\boldsymbol{\pi})^{\intercal}(\Delta+R_i)-\gamma'\|\Delta\|_2^2\}\nonumber\\
&=\sup_\Delta\{(-\boldsymbol{\pi})^{\intercal}\Delta-\gamma'\|\Delta\|_2^2\}-\boldsymbol{\pi}^{\intercal}R_i\nonumber\\
&=\sup_\Delta\{\|\boldsymbol{\pi}\|_2\|\Delta\|_2-\gamma'\|\Delta\|_2^2\}-\boldsymbol{\pi}^{\intercal}R_i\nonumber\\
&=\frac{\|\boldsymbol{\pi}\|_2^2}{4 \gamma'}-\boldsymbol{\pi}^\intercal R_i.}
Therefore, (\ref{A.39}) becomes
\eq{\max_{P\in \scU_\delta(Q)}[E_P((-\boldsymbol{\pi})^{\intercal}\boldsymbol{R})]&=\inf_{\gamma'\geq0}\crl{\gamma'\delta+\frac{1}{N}\sum_{i=1}^N\left[\frac{\|\boldsymbol{\pi}\|_2^2}{4\gamma'}-\boldsymbol{\pi}^{\intercal}R_i\right]}\nonumber\\
&=\inf_{\gamma'\geq0}\crl{\gamma'\delta+\frac{\|\boldsymbol{\pi}\|_2^2}{4\gamma'}-E_{Q}[\boldsymbol{\pi}^{\intercal}\boldsymbol{R}]}\nonumber\\
&=\sqrt{\delta}\|\boldsymbol{\pi}\|_2-E_{Q}\left[\boldsymbol{\pi}^{\intercal}\boldsymbol{R}\right]}
or\eq{\min_{P\in \scU_\delta(Q)}[E_P(\boldsymbol{\pi}^{\intercal}\boldsymbol{R})]=E_{Q}[\boldsymbol{\pi}^{\intercal}\boldsymbol{R}]-\sqrt{\delta}\|\boldsymbol{\pi}\|_2.}
For $\kappa=2$, the feasible region can be rewritten as
\eqn{\scF_{(\delta,\rho) }:=\crl{\boldsymbol{\pi}:E_{Q}[\boldsymbol{\pi}^{\intercal}\boldsymbol{R}]-\sqrt{\delta}\|\boldsymbol{\pi}\|_2\geq\rho,\;\boldsymbol{\pi}^\intercal\boldsymbol{1}=1}.
}
Let $f(\boldsymbol{\pi},a,\boldsymbol{R})=\max\left\{\displaystyle-\frac{1}{\alpha}\boldsymbol{\pi}^\intercal \boldsymbol{R}+a\brak{1-\frac{1}{\alpha}},a\right\}$. Consider our inner optimization problem
\eq{\label{A.42}\max_{P\in\scU_\delta(Q)}E_P\Big[f(\boldsymbol{\pi},\boldsymbol{R})\Big].}
According to Proposition \ref{proposition1}, the optimal value of (\ref{A.42}) is
\eqn{\min_{\gamma\geq0}\Big[\gamma\delta+\frac{1}{N}\sum_{i=1}^N\Phi_{\gamma}(R_i)\Big],}
where
\eqn{\Phi_{\gamma}(R_i)=\sup_{\boldsymbol{R}}\Big[f(\boldsymbol{\pi},\boldsymbol{R})-\gamma c(\boldsymbol{R},R_i)\Big].}
Then the inner optimization problem (\ref{A.42}) becomes
\eq{\label{A.43}\max_{P\in\scU_\delta(Q)}E_P[f(\boldsymbol{\pi},\boldsymbol{R})]&=\inf_{\gamma\geq 0}\Big[\gamma\delta+\frac{1}{N}\sum_{i=1}^N\Phi_{\gamma}(R_i)\Big]\\
&=\inf_{\gamma\geq 0}\Big[\gamma\delta+\frac{1}{N}\sum_{i=1}^N\sup_{\boldsymbol{R}}\crl{f(\boldsymbol{\pi},\boldsymbol{R})-\gamma c(\boldsymbol{R},R_i)}\Big]\\
&=\inf_{\gamma\geq 0}\Big[\gamma\delta+\frac{1}{N}\sum_{i=1}^N\sup_{\boldsymbol{R}}\crl{f(\boldsymbol{\pi},\boldsymbol{R})-\gamma\|\boldsymbol{R}-R_i\|_2^2}\Big].
}
Introducing the the auxiliary variables $s_i$, $i\leq N$, allows us to reformulate (\ref{A.43}) as
\eqn{& \inf_{\gamma\geq 0,s_i}\gamma\delta+\frac{1}{N}\sum_{i=1}^N s_i,\nonumber\\
&\text{s.t.}\; \sup_{\boldsymbol{R}}\crl{f(\boldsymbol{\pi},\boldsymbol{R})-\gamma\|\boldsymbol{R}-R_i\|^2_2}\leq s_i.
}
If $f(\boldsymbol{\pi},\boldsymbol{R})=\displaystyle-\frac{1}{\alpha}\boldsymbol{\pi}^\intercal \boldsymbol{R}+a\brak{1-\frac{1}{\alpha}}$, we have
\eqn{&\sup_{\boldsymbol{R}}\crl{f(\boldsymbol{\pi},\boldsymbol{R})- \gamma\|\boldsymbol{R}-R_i\|_2^2}\nonumber\\
&=\sup_{\boldsymbol{R}}\crl{-\frac{1}{\alpha}\boldsymbol{\pi}^\intercal \boldsymbol{R}+a\brak{1-\frac{1}{\alpha}}- \gamma\|\boldsymbol{R}-R_i\|_2^2}\nonumber\\
&=\sup_{\Delta}\crl{-\frac{1}{\alpha}\boldsymbol{\pi}^\intercal (\Delta+R_i)+a\brak{1-\frac{1}{\alpha}}- \gamma\|\Delta\|_2^2}\nonumber\\
&=\sup_{\Delta}\crl{-\gamma\|\Delta\|_2^2-\frac{1}{\alpha}\boldsymbol{\pi}^\intercal\Delta}-\frac{1}{\alpha}\boldsymbol{\pi}^\intercal R_i+a\brak{1-\frac{1}{\alpha}}\nonumber\\
&=\frac{\|\boldsymbol{\pi}\|_2^2}{4\gamma\alpha^2}-\frac{1}{\alpha}\boldsymbol{\pi}^\intercal R_i+a\brak{1-\frac{1}{\alpha}}\leq s_i.
}
If $f(\boldsymbol{\pi},\boldsymbol{R})=a$, we obtain
\eqn{a\leq s_i.}
Therefore, the inner maximum problem (\ref{A.42}) has the dual problem:
\eq{\label{A5.6}& \inf_{\gamma\geq 0,s_i}\gamma\delta+\frac{1}{N}\sum_{i=1}^N s_i,\nonumber\\
&\text{s.t.}\; \frac{\|\boldsymbol{\pi}\|_2^2}{4\gamma\alpha^2}-\frac{1}{\alpha}\boldsymbol{\pi}^\intercal R_i+a\brak{1-\frac{1}{\alpha}}\leq s_i,\nonumber\\
&\qquad a\leq s_i.
}
Finally, the initial optimization problem (\ref{eq4.10}) is equivalent to 
\eqn{& \min_{\boldsymbol{\pi},a,\gamma,s_i}\gamma\delta+\frac{1}{N}\sum_{i=1}^N s_i,\\
&\text{s.t.}\quad \frac{\|\boldsymbol{\pi}\|_2^2}{4\gamma\alpha^2}-\frac{1}{\alpha}\boldsymbol{\pi}^\intercal R_i+a\brak{1-\frac{1}{\alpha}}\leq s_i,\\
&\qquad a\leq s_i,\\
&\qquad E_{Q}[\boldsymbol{\pi}^{\intercal}\boldsymbol{R}]-\sqrt{\delta}\|\boldsymbol{\pi}\|_2\geq\rho,\\
&\qquad\boldsymbol{\pi}^\intercal 1= 1,\\
&\qquad\gamma\geq 0,\\
&\qquad a\in\bbR,\;\;s_i\in\bbR.
}
This completes the proof for \ref{RMC-2}.
\end{proof}

\subsection{}
\begin{proof}[Proof of Theorem \ref{theorem2}]
\label{prooftheorem2}

Assume $\boldsymbol{\pi}^*$ and $a^*$ satisfy the optimal condition \eqn{E_{P^*}[h(\boldsymbol{\pi}^*,a^*,\boldsymbol{R})]=\boldsymbol{0},}
where 
\eqn{h(\boldsymbol{\pi}^*,a^*,\boldsymbol{R}) =g(\boldsymbol{\pi}^*,a^*,\boldsymbol{R})-\lambda_1^*\boldsymbol{R}-\lambda_2^*\boldsymbol{1},
}
and
\eqn{g(\boldsymbol{\pi}^*,a^*,\boldsymbol{R})=\frac{-\boldsymbol{R}}{\alpha[1+exp(((\boldsymbol{\pi}^*)^\intercal\boldsymbol{R}+a^*)/t)]}.
}
We know
\eqn{\|h(\boldsymbol{\pi}^*,a^*,\boldsymbol{R})\|_2^2\leq \||g(\boldsymbol{\pi}^*,a^*,\boldsymbol{R})|+|\lambda_1^*\boldsymbol{R}|+|\lambda_2^*\boldsymbol{1}|\|_2^2.
}
Thus $E_{P^*}\|h(\boldsymbol{\pi}^*,a^*,\boldsymbol{R})\|_2^2$ is finite. Denote $I_n$ $n\times n$ identity matrix.
\eqn{D_{\boldsymbol{R}}h(\boldsymbol{\pi}^*,a^*,\boldsymbol{R}) &= \frac{-I_n}{\alpha[1+exp(((\boldsymbol{\pi}^*)^\intercal\boldsymbol{R}+a^*)/t)]}\\
&+ \frac{\boldsymbol{R}(\boldsymbol{\pi}^*)^\intercal}{\alpha t[1+exp(((\boldsymbol{\pi}^*)^\intercal\boldsymbol{R}+a^*)/t)][1+exp((-(\boldsymbol{\pi}^*)^\intercal\boldsymbol{R}-a^*)/t)]}-\lambda_1^*I_n.
}
It is continuously differentiable. For any $\xi \neq \boldsymbol{0}$, we know
\eqn{&P^{*}\Big(\|\xi^\intercal D_{\boldsymbol{R}}h(\boldsymbol{\pi}^*,a^*,\boldsymbol{R})\|_2=0\Big)\\
&=P^{*}\Big(t[1+exp((-(\boldsymbol{\pi}^*)^\intercal\boldsymbol{R}-a^*)/t)]\xi=(\xi^\intercal\boldsymbol{R})\boldsymbol{\pi}^*\\
&\quad -(\lambda_1^*\alpha t[1+exp(((\boldsymbol{\pi}^*)^\intercal\boldsymbol{R}+a^*)/t)][1+exp((-(\boldsymbol{\pi}^*)^\intercal\boldsymbol{R}-a^*)/t)])\xi\Big)=0
}
thus satisfying $P^*(\|\xi^\intercal D_{\boldsymbol{R}}h(\boldsymbol{\pi}^*,a^*,\boldsymbol{R})\|_2>0)>0$. This completes the proof that Assumptions \ref{as4} and \ref{as5} can also hold for $h(\boldsymbol{\pi}^*,a^*,\boldsymbol{R})$.

As all the conditions required in Theorem \ref{theorem1}, we obtain
\eqn{\sqrt{N}\scR_N(1) \Rightarrow \widebar{\scR}(1) := \sup_{\xi\in \Xi}\xi^\intercal Z,
}
where
\eqn{Z\sim \mathcal{N}(\boldsymbol{0},E_{P^*}[h(\boldsymbol{\pi}^*,a^*,\boldsymbol{R})h(\boldsymbol{\pi}^*,a^*,\boldsymbol{R})^\intercal])
}
and the set 
\eqn{\Xi=\{\xi\in\bbR^n:\|\xi^\intercal D_{\boldsymbol{R}}g(\boldsymbol{\pi}^*,a^*,\boldsymbol{R})\|\leq 1 \}.}
Proof of the upper bound for $\scR_N(1)$.
First, we claim that $\Xi$ is a subset of the norm ball $\{\xi\in\bbR^n:\|\xi\|_2\leq 1\}$. Due to H$\ddot{o}$lder's inequality $|\xi^\intercal\boldsymbol{R}|\leq \|\xi\|_2\|\boldsymbol{R}\|_2$, we obtain 
\eq{\label{eqA5.1}\|\xi^\intercal D_{\boldsymbol{R}}h(\boldsymbol{\pi}^*,a^*,\boldsymbol{R})\|_2 &=\Big\|\frac{-\xi^\intercal}{\alpha[1+exp(((\boldsymbol{\pi}^*)^\intercal\boldsymbol{R}+a^*)/t)]}\\
&\quad + \frac{(\xi^\intercal\boldsymbol{R})(\boldsymbol{\pi}^*)^\intercal}{\alpha t[1+exp(((\boldsymbol{\pi}^*)^\intercal\boldsymbol{R}+a^*)/t)][1+exp((-(\boldsymbol{\pi}^*)^\intercal\boldsymbol{R}-a^*)/t)]}-\lambda_1^*\xi^\intercal\Big\|_2
\\&\geq 
\Big\|\frac{-\xi}{\alpha[1+exp(((\boldsymbol{\pi}^*)^\intercal\boldsymbol{R}+a^*)/t)]}\\
&\quad + \frac{(\xi^\intercal\boldsymbol{R})\boldsymbol{\pi}^*}{\alpha t[1+exp(((\boldsymbol{\pi}^*)^\intercal\boldsymbol{R}+a^*)/t)][1+exp((-(\boldsymbol{\pi}^*)^\intercal\boldsymbol{R}-a^*)/t)]}\Big\|_2-|\lambda_1^*|\|\xi\|_2\\
&\geq 
\Big\|\frac{\xi}{\alpha[1+exp(((\boldsymbol{\pi}^*)^\intercal\boldsymbol{R}+a^*)/t)]}\Big\|_2\\
&\quad -\Big\|\frac{(\xi^\intercal\boldsymbol{R})\boldsymbol{\pi}^*}{\alpha t[1+exp(((\boldsymbol{\pi}^*)^\intercal\boldsymbol{R}+a^*)/t)][1+exp((-(\boldsymbol{\pi}^*)^\intercal\boldsymbol{R}-a^*)/t)]}\Big\|_2-|\lambda_1^*|\|\xi\|_2\\
&\geq \Big( \frac{1}{\alpha[1+exp(((\boldsymbol{\pi}^*)^\intercal\boldsymbol{R}+a^*)/t)]}\\
&\quad -\frac{\|\boldsymbol{R}\|_2\|\boldsymbol{\pi}^*\|_2}{\alpha t[1+exp(((\boldsymbol{\pi}^*)^\intercal\boldsymbol{R}+a^*)/t)][1+exp((-(\boldsymbol{\pi}^*)^\intercal\boldsymbol{R}-a^*)/t)]}-|\lambda_1^*|\Big)\|\xi\|_2.
}
 Following (\ref{eqA5.1}), if $\xi\in\bbR^n$ is such that $\|\xi\|_2=(1-\epsilon)^{-2}>1$ for a given $\epsilon>0$, then 
 \eqn{\|\xi^\intercal D_{\boldsymbol{R}}h(\boldsymbol{\pi}^*,a^*,\boldsymbol{R})\|_2>1,} 
 whenever
\eqn{(\boldsymbol{\pi}^*,a^*,\boldsymbol{R},\lambda_1^*)\in\Omega_{\epsilon}:=\Big\{(\boldsymbol{\pi}^*,a^*,\boldsymbol{R},\lambda_1^*):
\frac{1}{\alpha[1+exp(((\boldsymbol{\pi}^*)^\intercal\boldsymbol{R}+a^*)/t)]}\geq 1-\frac{\epsilon}{2},\\
\frac{\|\boldsymbol{R}\|_2\|\boldsymbol{\pi}\|_2}{\alpha t[1+exp(((\boldsymbol{\pi}^*)^\intercal\boldsymbol{R}+a^*)/t)][1+exp((-(\boldsymbol{\pi}^*)^\intercal\boldsymbol{R}-a^*)/t)]}+|\lambda_1^*|\leq \frac{\epsilon}{2}.
\Big\}
}
Since $\boldsymbol{R}$ has positive probability density almost everywhere, the set $\Omega_{\epsilon}$ has positive probability for every $\epsilon>0$. Thus, if $\|\xi\|_2>1$, $\|\xi^\intercal D_{\boldsymbol{R}}h(\boldsymbol{\pi}^*,a^*,\boldsymbol{R})\|_2 >1$ with positive probability. Therefore, $\Xi$ is a subset of $\crl{\xi:\|\xi\|_2\leq 1}$. Consequently,
\eqn{\sqrt{N}\scR_N(1)\Rightarrow \widebar{\scR}(1):=\sup_{\xi\in \Xi}\xi^\intercal Z \leq \sup_{\xi:\|\xi\|_2\leq 1}\xi^\intercal Z=\|\Tilde{Z}\|_2.
}
We can estimate $\Tilde{Z}$. When $t\rightarrow 0^+$, we have 
\eqn{\hat{g}(\boldsymbol{\pi}^*,a^*,\boldsymbol{R})=\lim_{t\rightarrow 0^+}g(\boldsymbol{\pi}^*,a^*,\boldsymbol{R})=\begin{cases} 
\boldsymbol{0} & \mbox{if}\; (\boldsymbol{\pi}^*)^\intercal\boldsymbol{R}+a^*>0\\
-\boldsymbol{R}/2\alpha  & \mbox{if}\; (\boldsymbol{\pi}^*)^\intercal\boldsymbol{R}+a^*=0
\\ -\boldsymbol{R}/\alpha & \mbox{if}\; (\boldsymbol{\pi}^*)^\intercal\boldsymbol{R}+a^*<0.
\end{cases}
}
Let $A=(\boldsymbol{\pi}^*)^\intercal\boldsymbol{R}+a^*$ and we have
\eqn{E_{P^*}(\hat{g}^i(\boldsymbol{\pi}^*,a^*,\boldsymbol{R}))= E_{P^*}\left(-\frac{R^i}{2\alpha}\mathbf{1}_{\crl{A=0}}-\frac{R^i}{\alpha}\mathbf{1}_{\crl{A<0}}\right)=-\frac{1}{\alpha} E_{P^*}(R^i\mathbf{1}_{\crl{A<0}}).}
Then we can obtain $\hat{\lambda}_1^*$ and $\hat{\lambda}_2^*$ as $t\rightarrow 0^+$,
\eqn{\hat{\lambda}_1^*=\lim_{t\rightarrow 0^+}\lambda_1^*&=\frac{E_{P^*}\left(-\frac{R^i}{2\alpha}\mathbf{1}_{\crl{A=0}}-\frac{R^i}{\alpha}\mathbf{1}_{\crl{A<0}}\right)-\sum_{i=1}^n(\pi^*)^i E_{P^*}\left(-\frac{R^i}{2\alpha}\mathbf{1}_{\crl{A=0}}-\frac{R^i}{\alpha}\mathbf{1}_{\crl{A<0}}\right)}{E_{P^*}[R^i]-\rho}\\
&=\frac{-\frac{1}{\alpha}E_{P^*}\left(R^i\mathbf{1}_{\crl{A<0}}\right)+\frac{1}{\alpha}\sum_{i=1}^n(\pi^*)^i E_{P^*}\left(R^i\mathbf{1}_{\crl{A<0}}\right)}{E_{P^*}[R^i]-\rho}
}
and 
\eqn{\hat{\lambda}_2^*=\lim_{t\rightarrow 0^+}\lambda_2^*=-\frac{1}{\alpha} \sum_{i=1}^n(\pi^*)^i E_{P^*}\left(R^i\mathbf{1}_{\crl{A<0}}\right)-\hat{\lambda}_1^*\rho.
}
Thus we have 
\eqn{\hat{h}(\boldsymbol{\pi}^*,a^*,\boldsymbol{R})=\lim_{t\rightarrow0^+}h(\boldsymbol{\pi}^*,a^*,\boldsymbol{R})=\hat{g}(\boldsymbol{\pi}^*,a^*,\boldsymbol{R})-\hat{\lambda}_1^*\boldsymbol{R}-\hat{\lambda}_2^*\boldsymbol{1}
}
and
\eqn{\|\hat{h}(\boldsymbol{\pi}^*,a^*,\boldsymbol{R})\|_2^2&\leq\||\hat{g}(\boldsymbol{\pi}^*,a^*,\boldsymbol{R})|+|\hat{\lambda}_1^*\boldsymbol{R}|+|\hat{\lambda}_2^*\boldsymbol{1}|\|_2^2\\
&\leq \left\|\frac{|\boldsymbol{R}|}{\alpha}+|\hat{\lambda}_1^*\boldsymbol{R}|+|\hat{\lambda}_2^*\boldsymbol{1}|\right\|_2^2\\
&= \left\|\left(\frac{1}{\alpha}+|\hat{\lambda}_1^*|\right)|\boldsymbol{R}|+|\hat{\lambda}_2^*|\boldsymbol{1}\right\|_2^2.
}
If we let 
\eqn{\Tilde{Z}\sim\mathcal{N}\left(\boldsymbol{0},E_{P^*}\left[\left(\Big(\frac{1}{\alpha}+|\hat{\lambda}_1^*|\Big)|\boldsymbol{R}|+|\hat{\lambda}_2^*|\boldsymbol{1}\right)\left(\Big(\frac{1}{\alpha}+|\hat{\lambda}_1^*|\Big)|\boldsymbol{R}|+|\hat{\lambda}_2^*|\boldsymbol{1}\right)^\intercal\right]\right),
}
then $Cov(\Tilde{Z})-Cov(Z)$ is positive definite. Thus, $\sqrt{N}\scR_N(1)$ is
stochastically dominated by $\|\Tilde{Z}\|_2$.

Proof of the upper bound for $\scR_N(2)$. If $p=q=\kappa=2$ and $E_{P^*}[D_{\boldsymbol{R}}h(\boldsymbol{\pi}^*,a^*,\boldsymbol{R})]$ is invertible,
\eqn{\widebar{\scR}(2)= \max_{\xi\in\bbR^n}\crl{2\xi^\intercal Z-\xi^\intercal E_{P^*}[ D_{\boldsymbol{R}}h(\boldsymbol{\pi}^*,a^*,\boldsymbol{R})]\xi}=Z^\intercal (E_{P^*}[ D_{\boldsymbol{R}}h(\boldsymbol{\pi}^*,a^*,\boldsymbol{R})])^{-1}Z.
}
As $t\rightarrow 0^+$, we have 
\eqn{\hat{D}_{\boldsymbol{R}}h(\boldsymbol{\pi}^*,a^*,\boldsymbol{R})
&=\lim_{t\rightarrow 0^+}D_{\boldsymbol{R}}h(\boldsymbol{\pi}^*,a^*,\boldsymbol{R})\\
&=\begin{cases} 
-\hat{\lambda}_1^*I_n & \mbox{if}\; (\boldsymbol{\pi}^*)^\intercal\boldsymbol{R}+a^*>0\\
\infty & \mbox{if}\; (\boldsymbol{\pi}^*)^\intercal\boldsymbol{R}+a^*=0
\\ -(\frac{1}{\alpha}+\hat{\lambda}_1^*)I_n & \mbox{if}\; (\boldsymbol{\pi}^*)^\intercal\boldsymbol{R}+a^*<0.
\end{cases}
}
Its expectation becomes 
\eqn{E_{P^*}[\hat{D}_{\boldsymbol{R}}h(\boldsymbol{\pi}^*,a^*,\boldsymbol{R})]=E_{P^*}\left[-\hat{\lambda}_1^*\mathbf{1}_{\crl{A>0}}-\frac{1+\alpha\hat{\lambda}_1^*}{\alpha}\mathbf{1}_{\crl{A<0}}\right]I_n
}
and the inverse is 
\eqn{(E_{P^*}[\hat{D}_{\boldsymbol{R}}h(\boldsymbol{\pi}^*,a^*,\boldsymbol{R})])^{-1}=\left(E_{P^*}\left[-\hat{\lambda}_1^*\mathbf{1}_{\crl{A>0}}-\frac{1+\alpha\hat{\lambda}_1^*}{\alpha}\mathbf{1}_{\crl{A<0}}\right]\right)^{-1}I_n.
}
Then we have
\eqn{\hat{\scR}_N(2)=\lim_{t\rightarrow0^+}\widebar{\scR}(2)&=\hat{Z}^\intercal(E_{P^*}[\hat{D}_{\boldsymbol{R}}h(\boldsymbol{\pi}^*,a^*,\boldsymbol{R})])^{-1}\hat{Z}\\
&=\left(E_{P^*}\left[-\hat{\lambda}_1^*\mathbf{1}_{\crl{A>0}}-\frac{1+\alpha\hat{\lambda}_1^*}{\alpha}\mathbf{1}_{\crl{A<0}}\right]\right)^{-1}\hat{Z}^\intercal \hat{Z},
}
where 
\eqn{\hat{Z}\sim \mathcal{N}(\boldsymbol{0},E_{P^*}[\hat{h}(\boldsymbol{\pi}^*,a^*,\boldsymbol{R})\hat{h}(\boldsymbol{\pi}^*,a^*,\boldsymbol{R})^\intercal]).}
As $t\rightarrow0^+$, we know $\|\hat{h}(\boldsymbol{\pi}^*,a^*,\boldsymbol{R})\|_2^2\leq\left\|\left(\displaystyle\frac{1}{\alpha}+|\hat{\lambda}_1^*|\right)|\boldsymbol{R}|+|\hat{\lambda}_2^*|\boldsymbol{1}\right\|_2^2$. Therefore we obtain the upper bound 
\eqn{N\scR_N(2) \Rightarrow \hat{\scR}_N(2) \leq \left(E_{P^*}\left[-\hat{\lambda}_1^*\mathbf{1}_{\crl{A>0}}-\frac{1+\alpha\hat{\lambda}_1^*}{\alpha}\mathbf{1}_{\crl{A<0}}\right]\right)^{-1}\Tilde{Z}^\intercal\Tilde{Z},
}
where
\eqn{\Tilde{Z}\sim\mathcal{N}\left(\boldsymbol{0},E_{P^*}\left[\left(\Big(\frac{1}{\alpha}+|\hat{\lambda}_1^*|\Big)|\boldsymbol{R}|+|\hat{\lambda}_2^*|\boldsymbol{1}\right)\left(\Big(\frac{1}{\alpha}+|\hat{\lambda}_1^*|\Big)|\boldsymbol{R}|+|\hat{\lambda}_2^*|\boldsymbol{1}\right)^\intercal\right]\right).
}
\end{proof}
\end{appendices}

\newpage
\section*{Figures}
\begin{figure}[H]
\centering
\includegraphics[width=15cm]{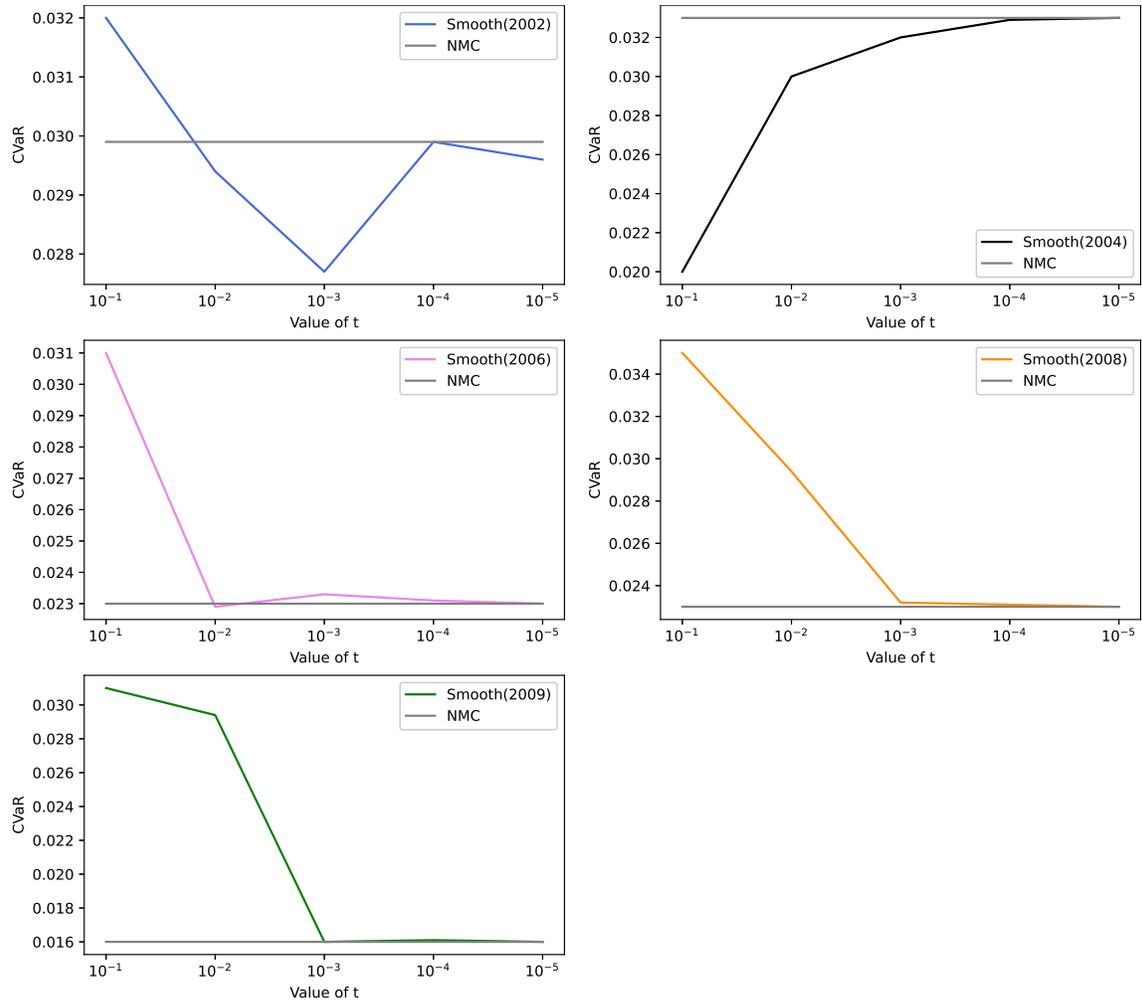}
\caption{Influence of the smooth parameter $t$}
\label{fig1}
\end{figure}

\begin{figure}[H]
\centering
\includegraphics[width=14cm]{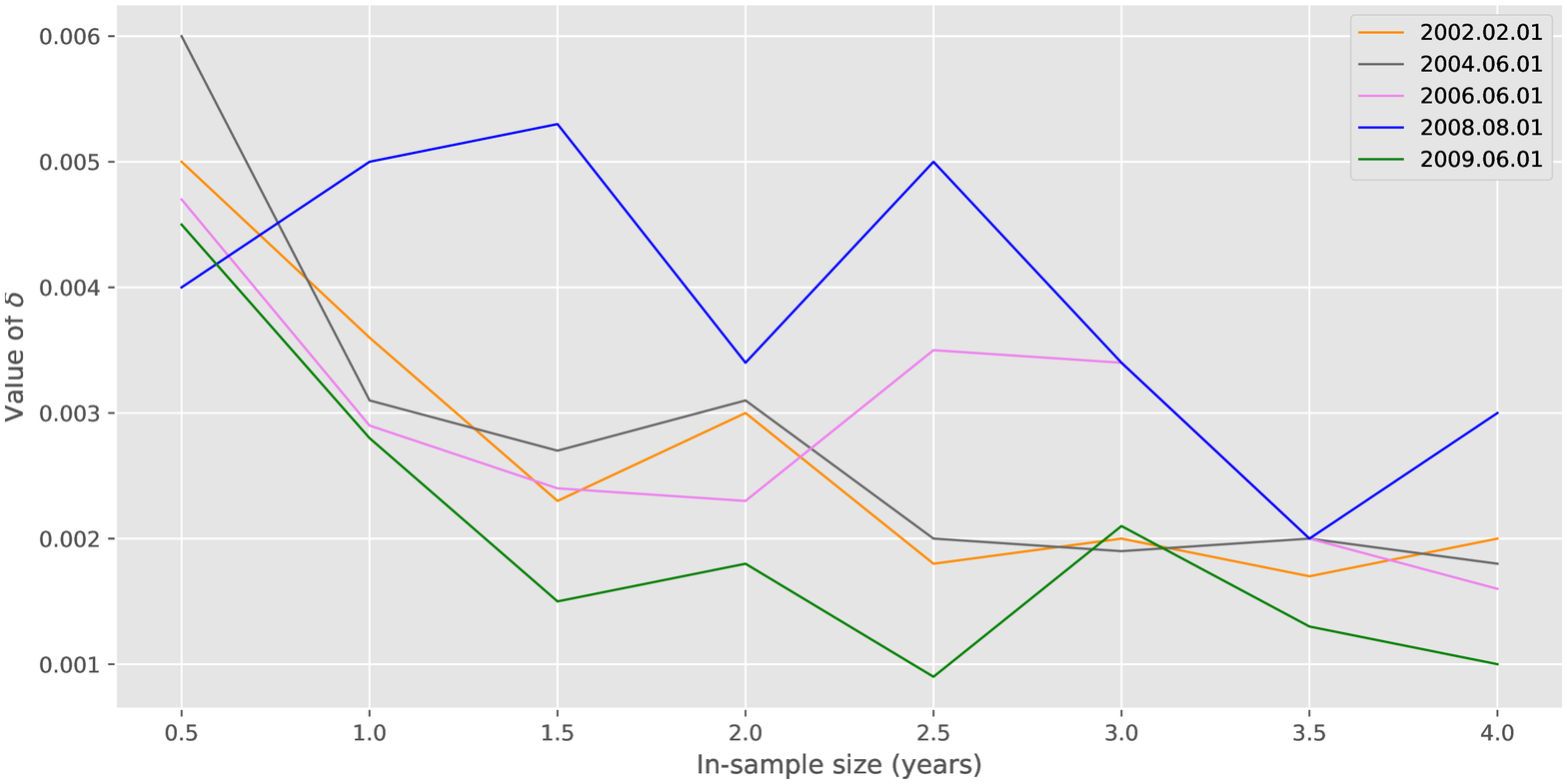}
\caption{$\delta$ for different in-sample size in RMC-1}
\label{fig2}
\end{figure}

\begin{figure}[H]
\centering
\includegraphics[width=14cm]{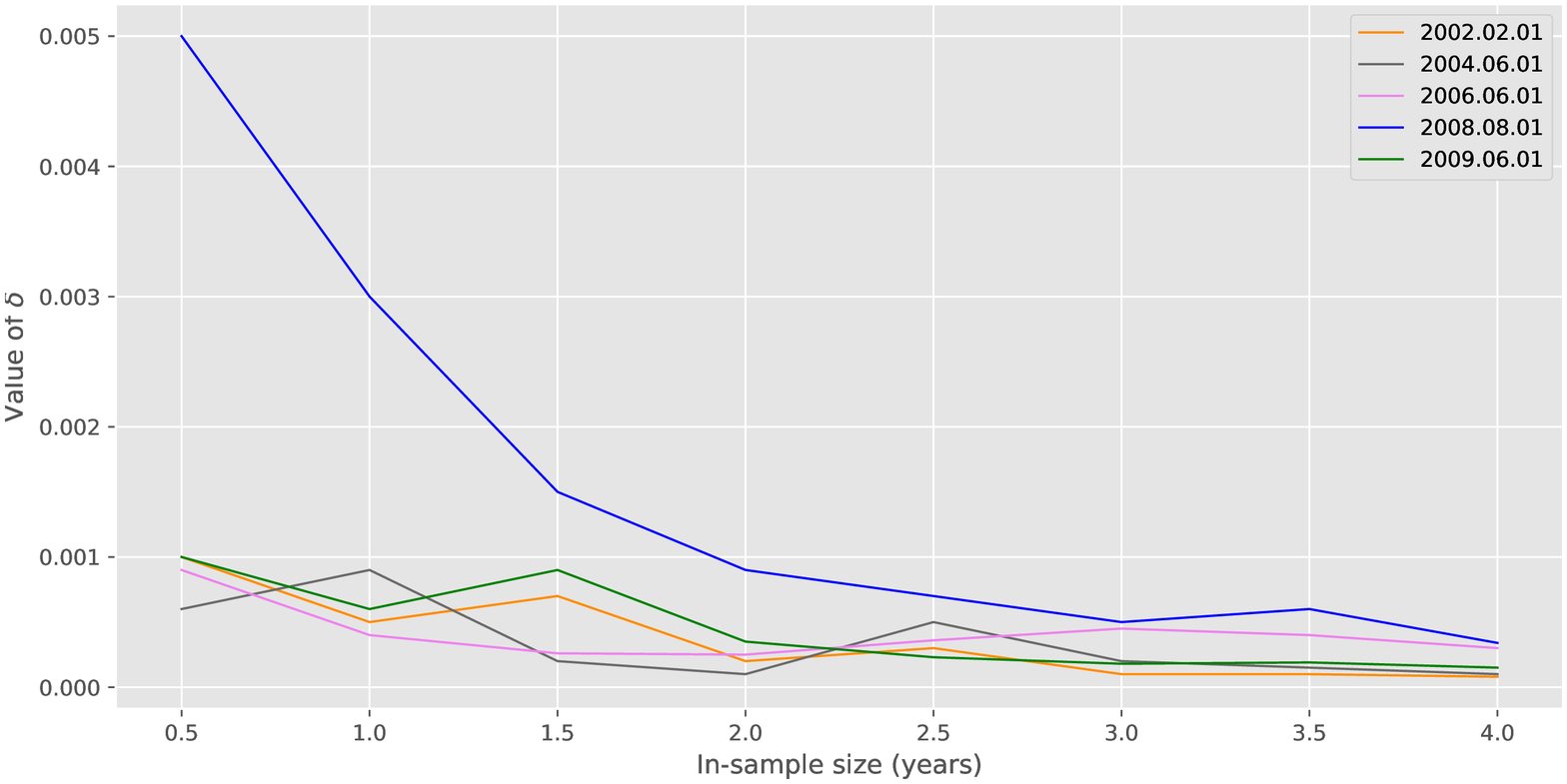}
\caption{$\delta$ for different in-sample size in RMC-2}
\label{fig3}
\end{figure}

\begin{figure}[H]
\centering
\includegraphics[width=12cm]{Fig/opt_weight2002.eps}
\caption{Composition of portfolios without transaction costs starting, 2002.02.01}
\label{fig4}
\end{figure}
\begin{figure}[H]
\centering
\includegraphics[width=12cm]{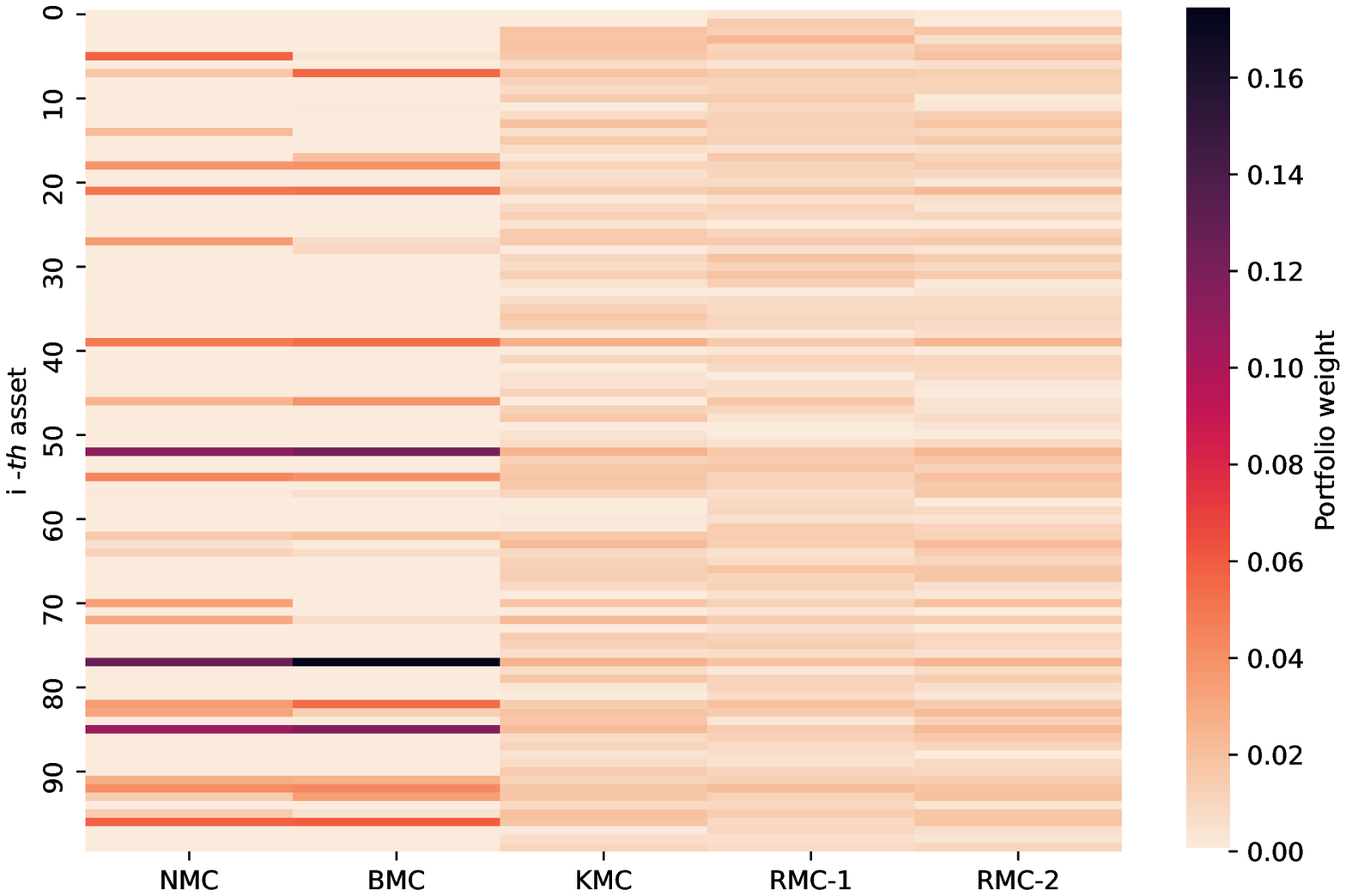}
\caption{Composition of portfolios without transaction costs starting, 2004.06.01}
\label{fig5}
\end{figure}
\begin{figure}[H]
\centering
\includegraphics[width=12cm]{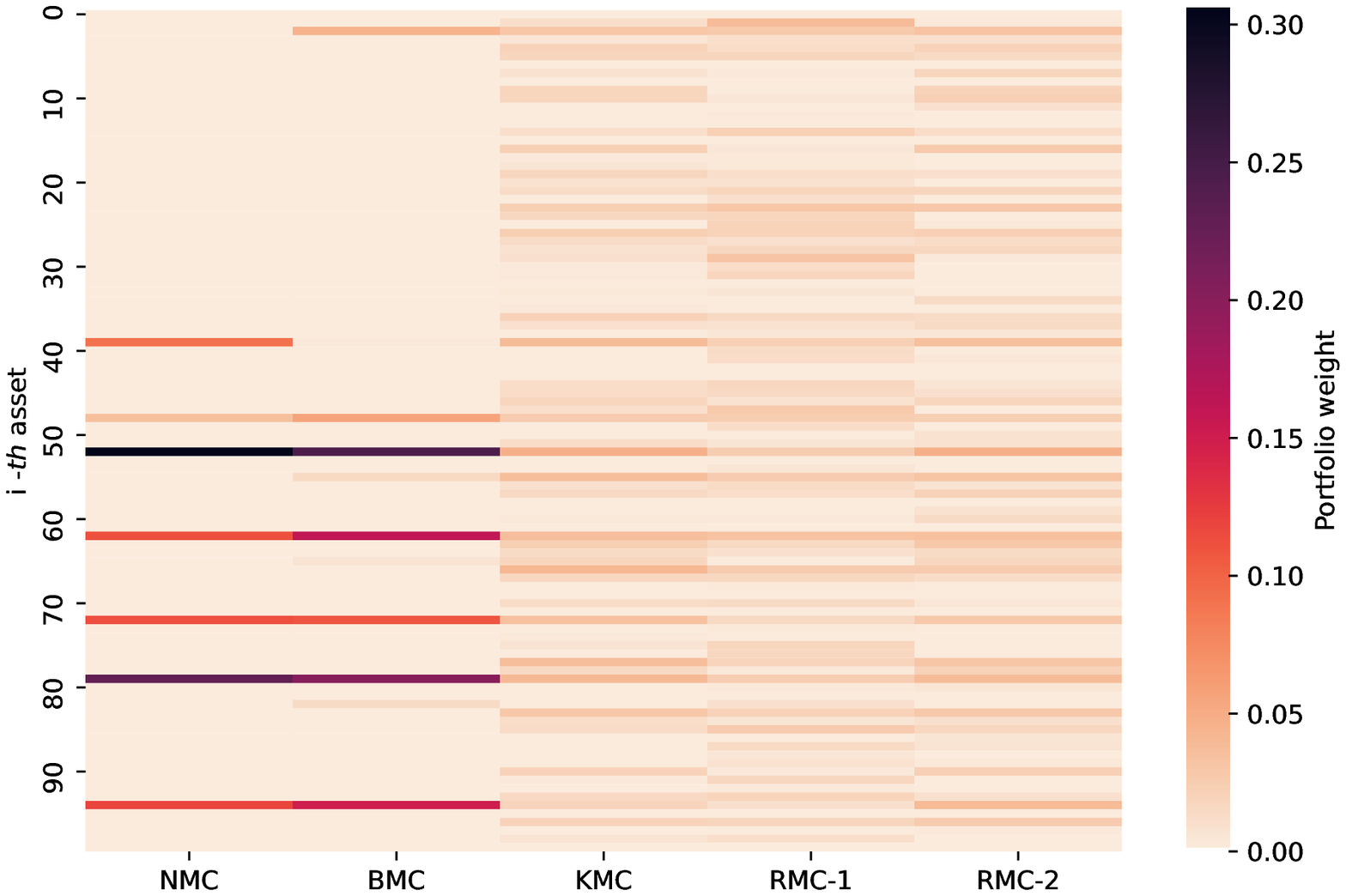}
\caption{Composition of portfolios without transaction costs starting, 2006.06.01}
\label{fig6}
\end{figure}
\begin{figure}[H]
\centering
\includegraphics[width=12cm]{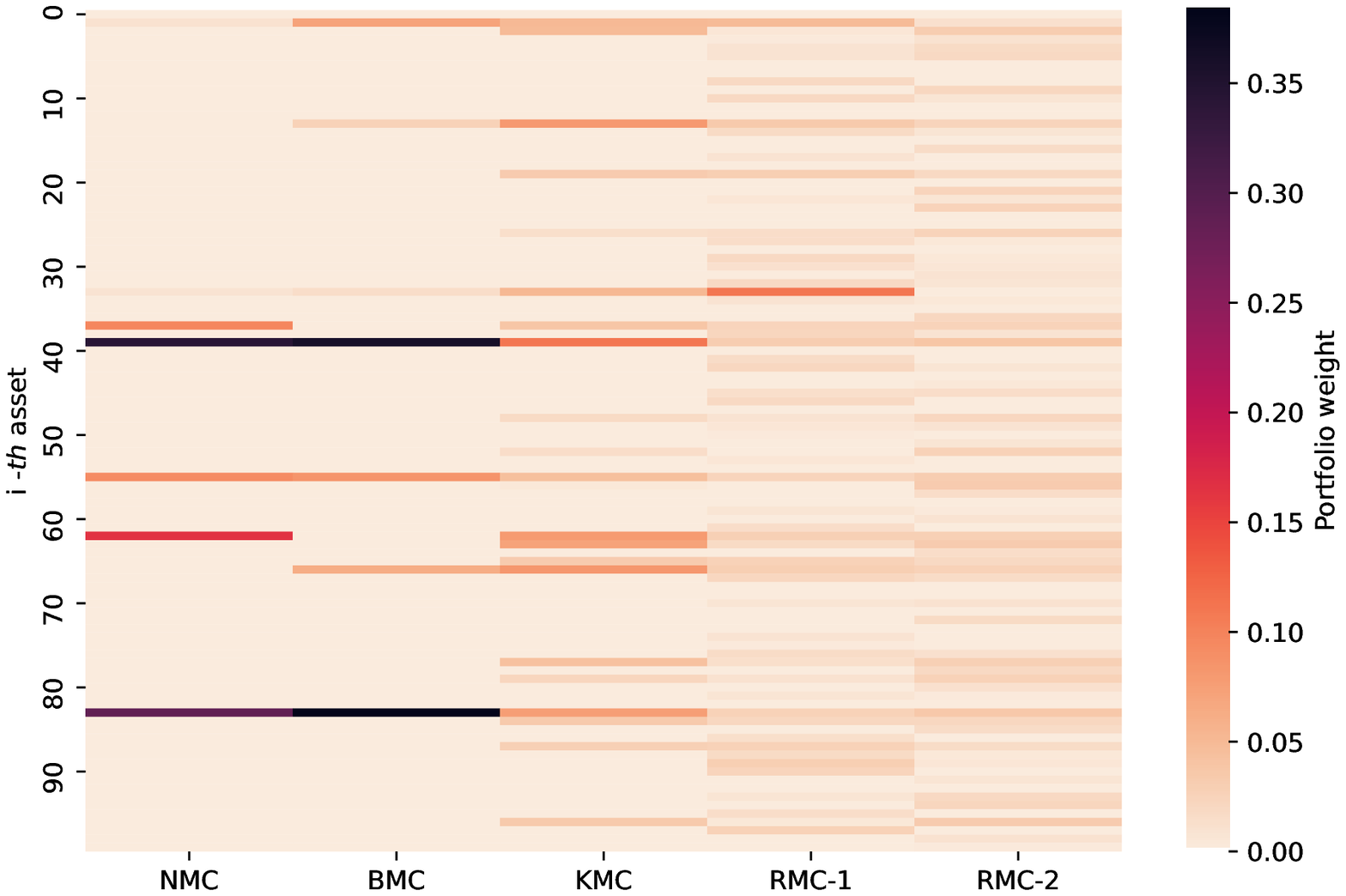}
\caption{Composition of portfolios without transaction costs starting, 2008.08.01}
\label{fig7}
\end{figure}
\begin{figure}[H]
\centering
\includegraphics[width=12cm]{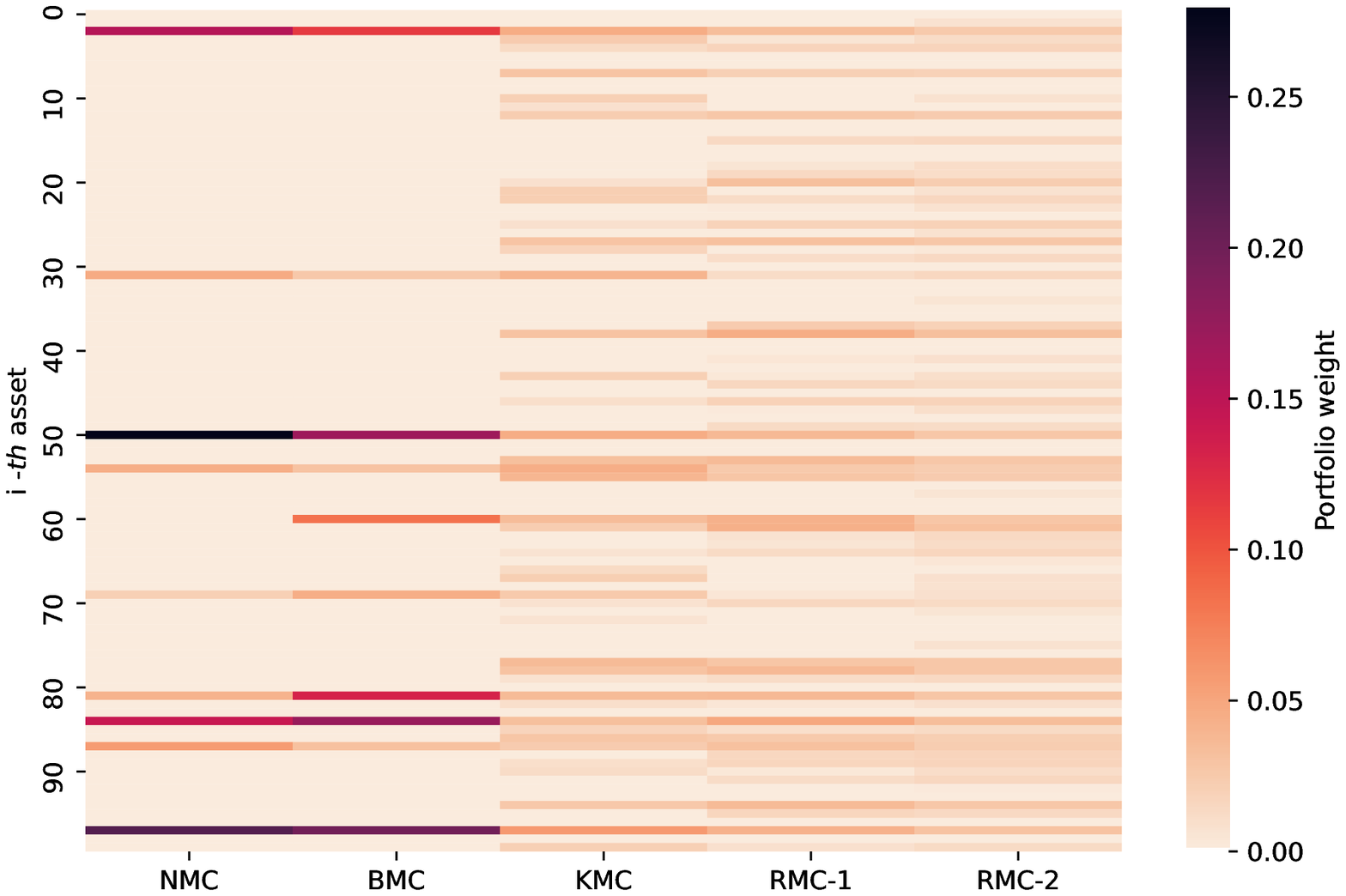}
\caption{Composition of portfolios without transaction costs starting, 2009.06.01}
\label{fig8}
\end{figure}

\begin{figure}[H]
\centering
\includegraphics[width=14cm]{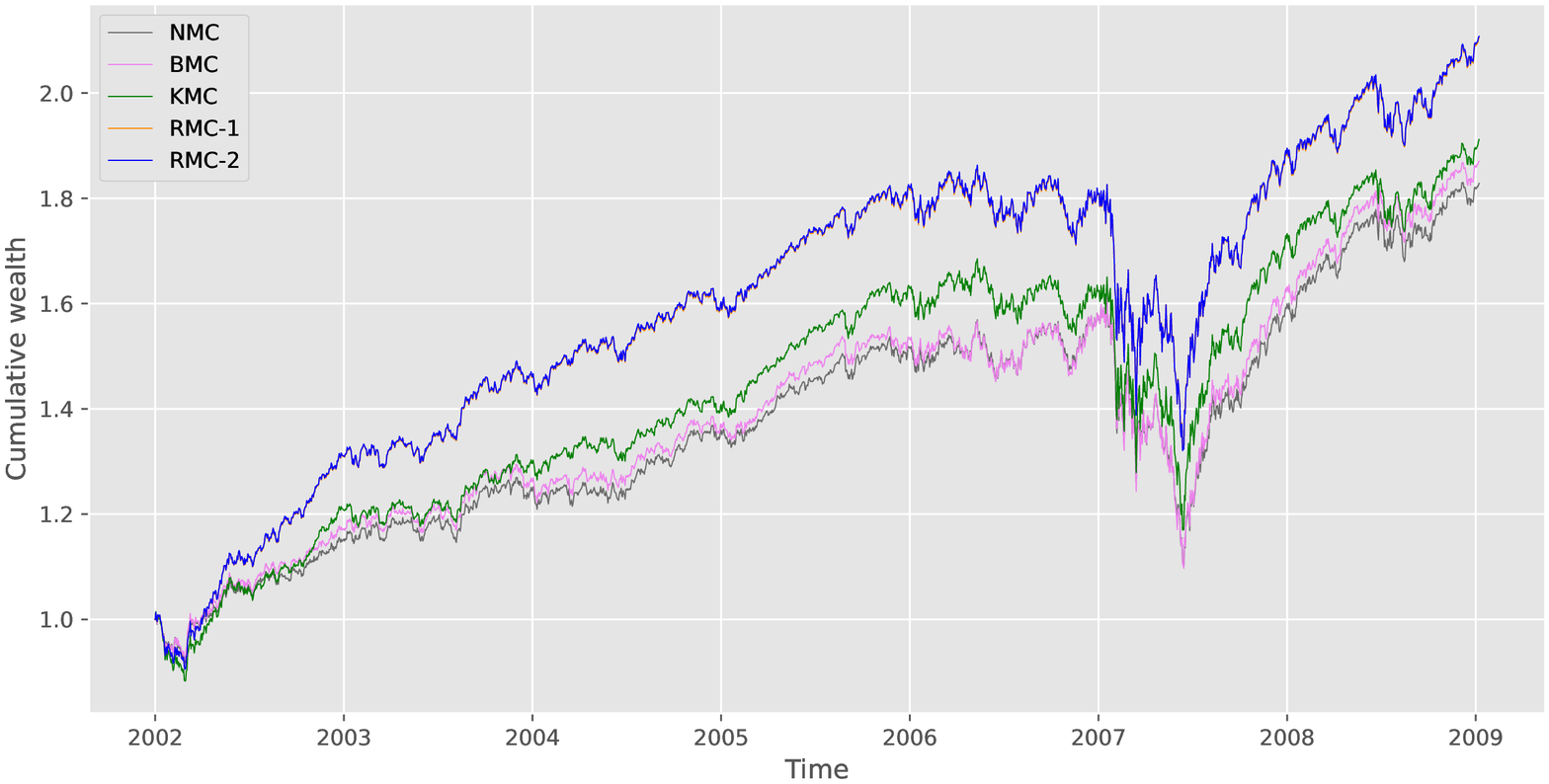}
\caption{Cumulative wealth without transaction costs, starting 2002.02.01}
\label{fig9}
\end{figure}
\begin{figure}[H]
\centering
\includegraphics[width=14cm]{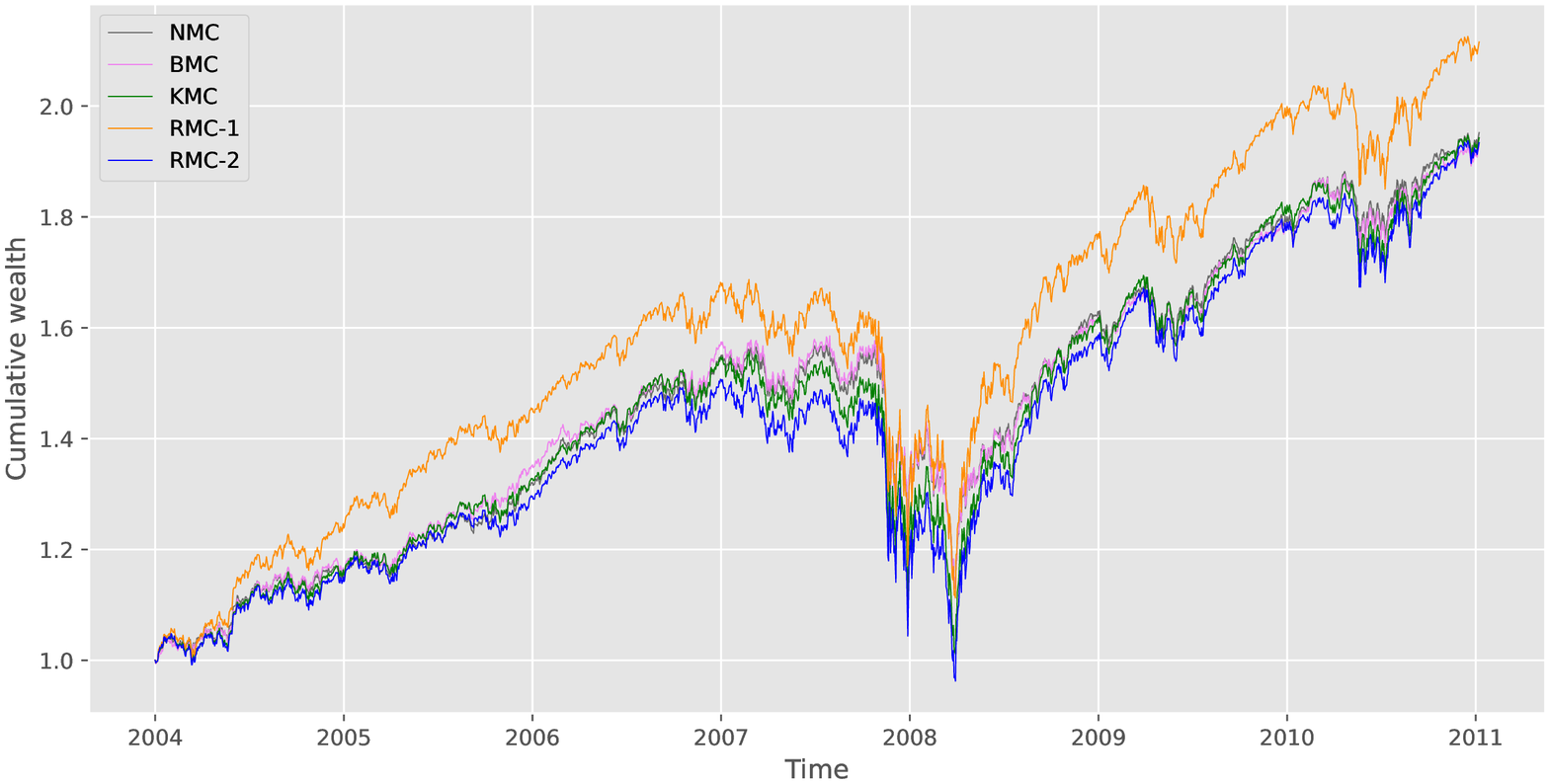}
\caption{Cumulative wealth without transaction costs, starting 2004.06.01}
\label{fig10}
\end{figure}\begin{figure}[H]
\centering
\includegraphics[width=14cm]{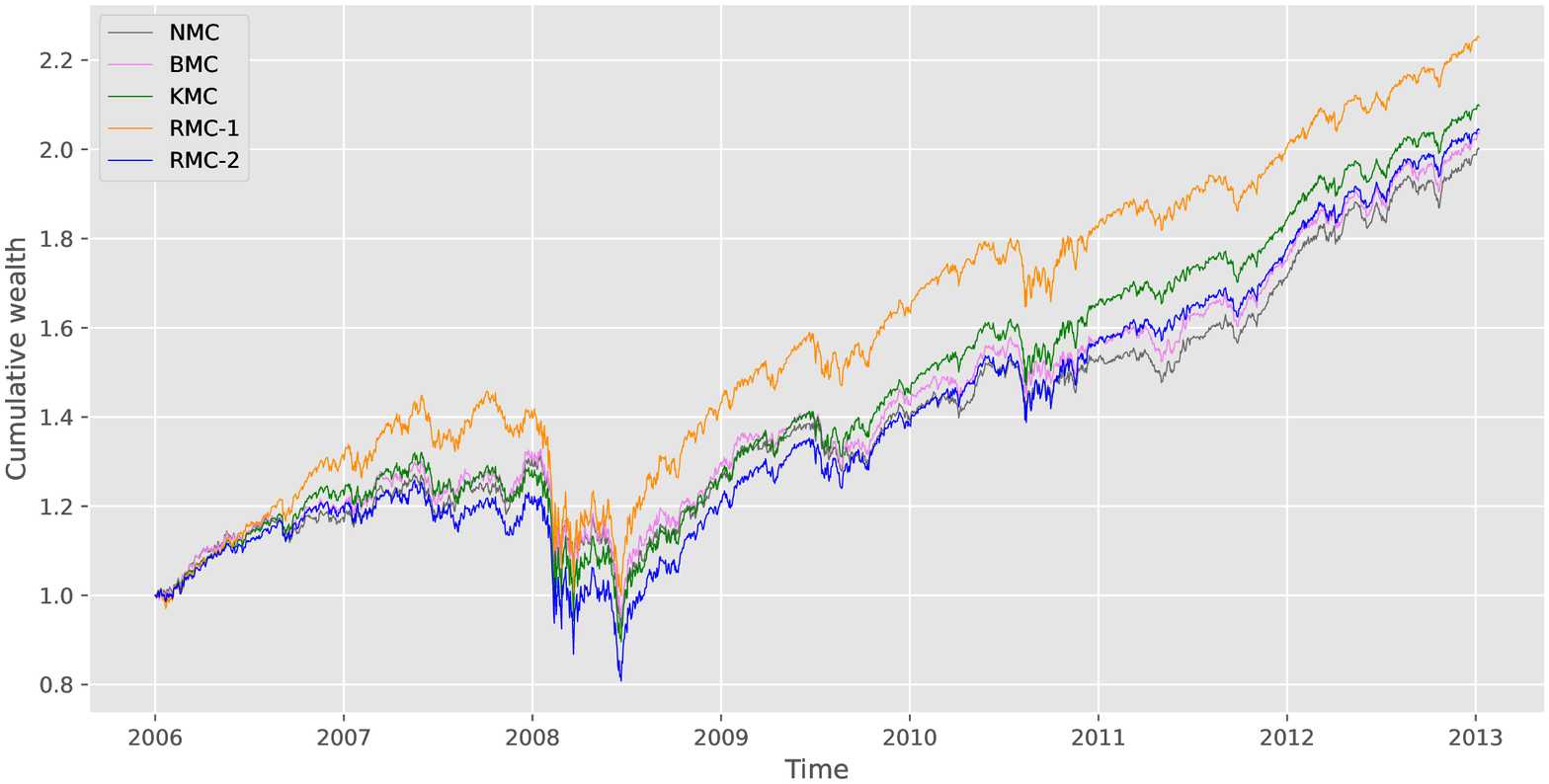}
\caption{Cumulative wealth without transaction costs, starting 2006.06.01}
\label{fig11}
\end{figure}
\begin{figure}[H]
\centering
\includegraphics[width=14cm]{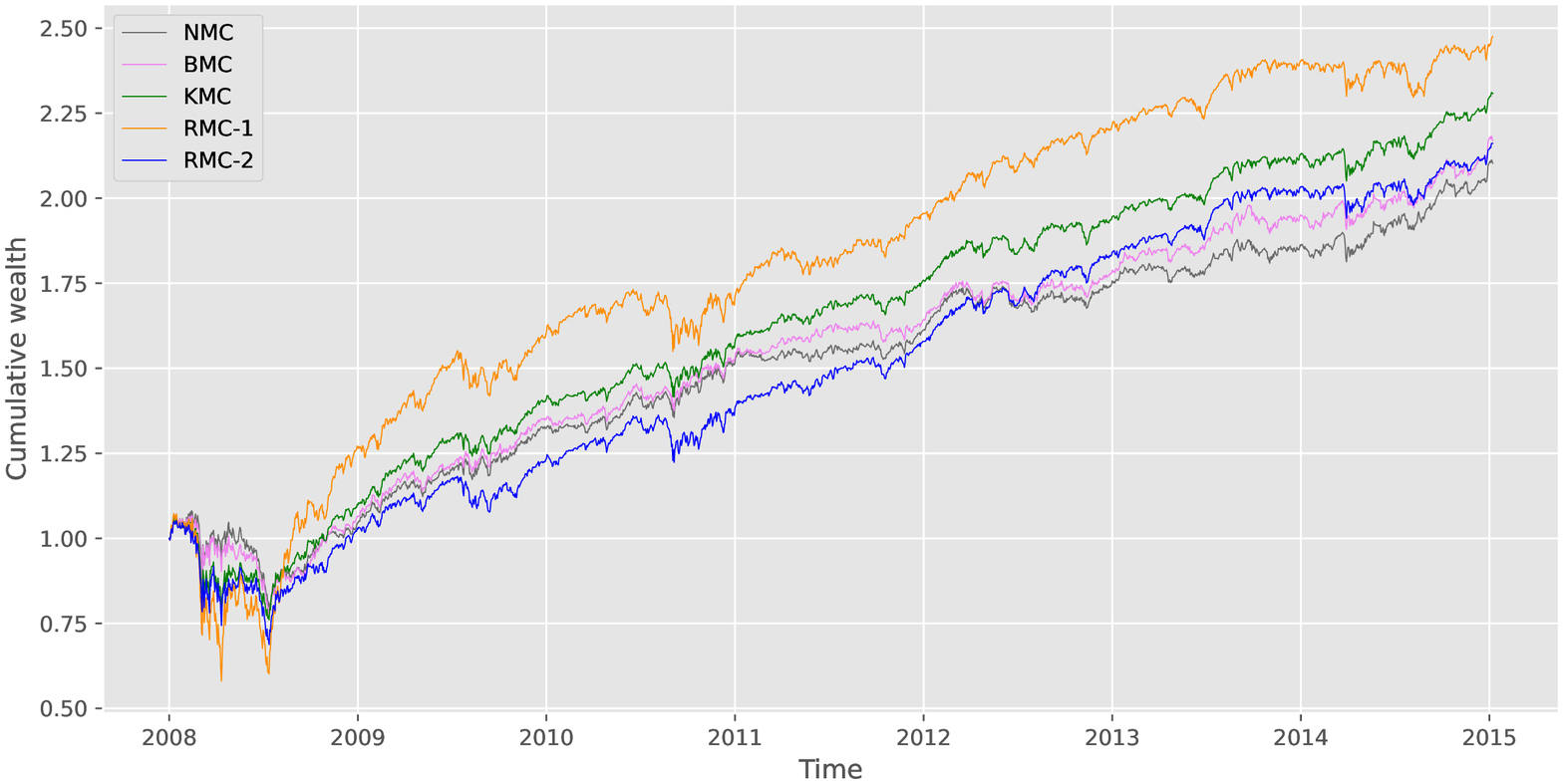}
\caption{Cumulative wealth without transaction costs, starting 2008.08.01}
\label{fig12}
\end{figure}
\begin{figure}[H]
\centering
\includegraphics[width=14cm]{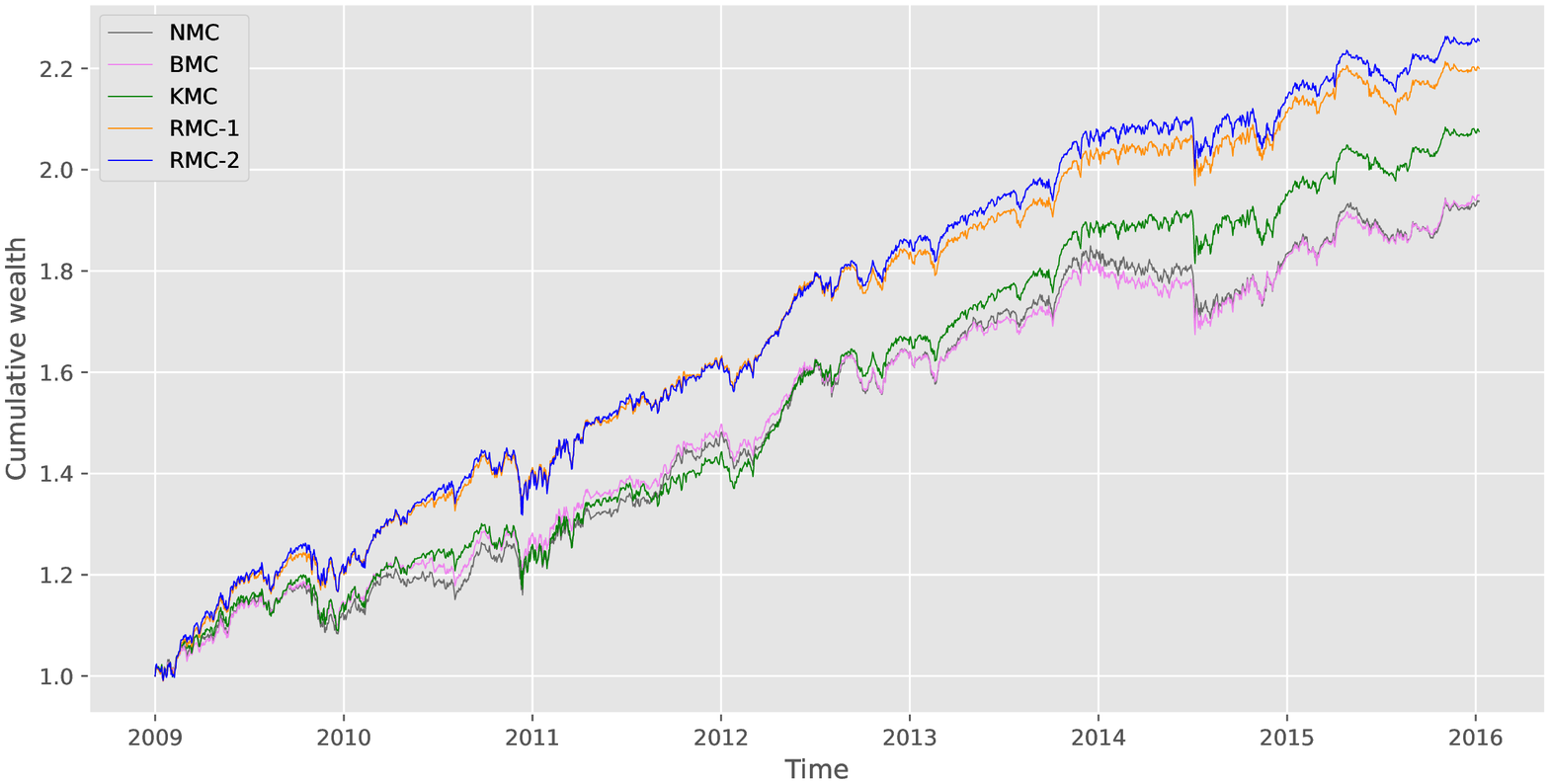}
\caption{Cumulative wealth without transaction costs, starting 2009.06.01}
\label{fig13}
\end{figure}

\begin{figure}[H]
\centering
\includegraphics[width=14cm]{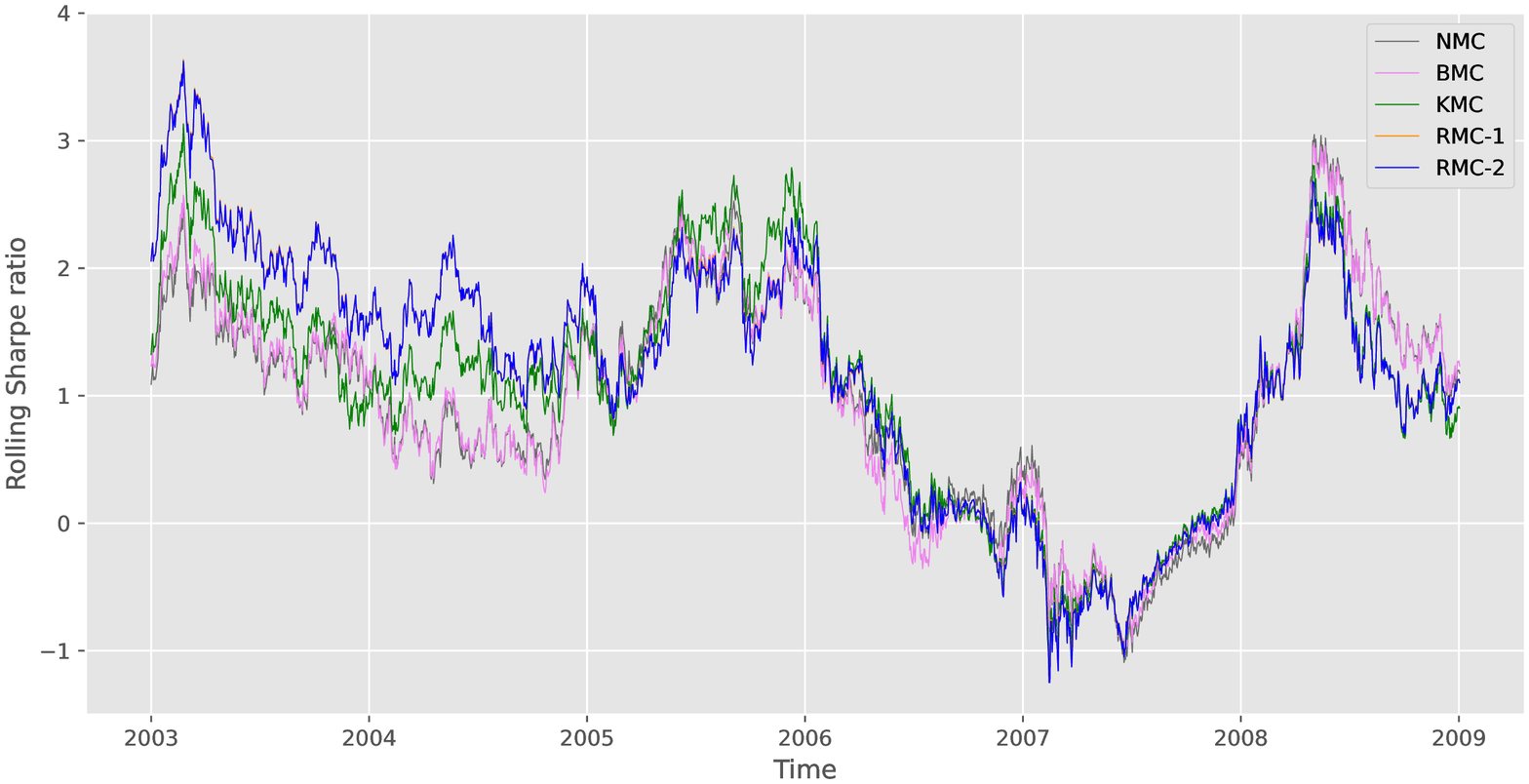}
\caption{Rolling one year Sharpe ratio without transaction costs, starting 2002.02.01}
\label{fig14}
\end{figure}
\begin{figure}[H]
\centering
\includegraphics[width=14cm]{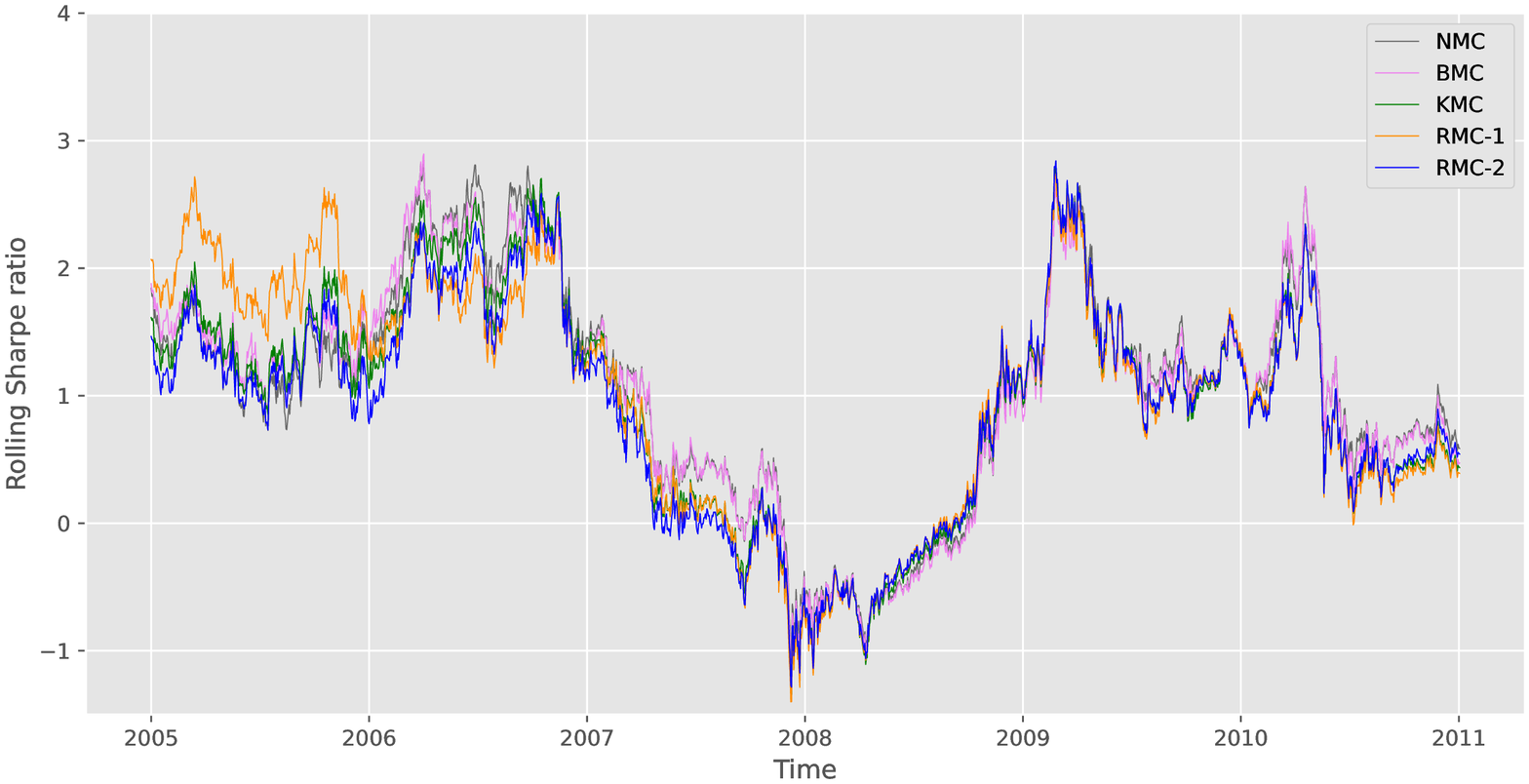}
\caption{Rolling one year Sharpe ratio without transaction costs, starting 2004.06.01}
\label{fig15}
\end{figure}\begin{figure}[H]
\centering
\includegraphics[width=14cm]{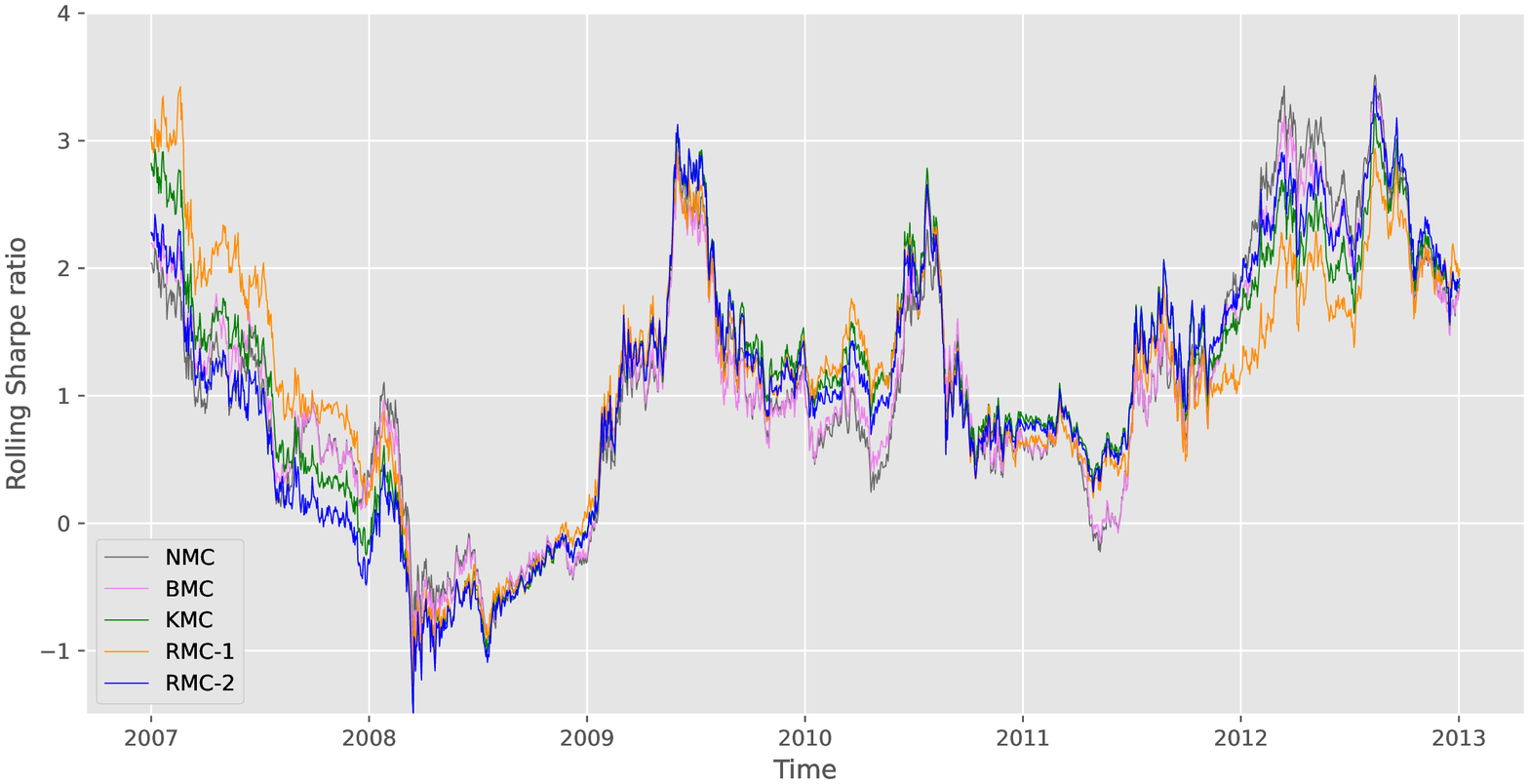}
\caption{Rolling one year Sharpe ratio without transaction costs, starting 2006.06.01}
\label{fig16}
\end{figure}
\begin{figure}[H]
\centering
\includegraphics[width=14cm]{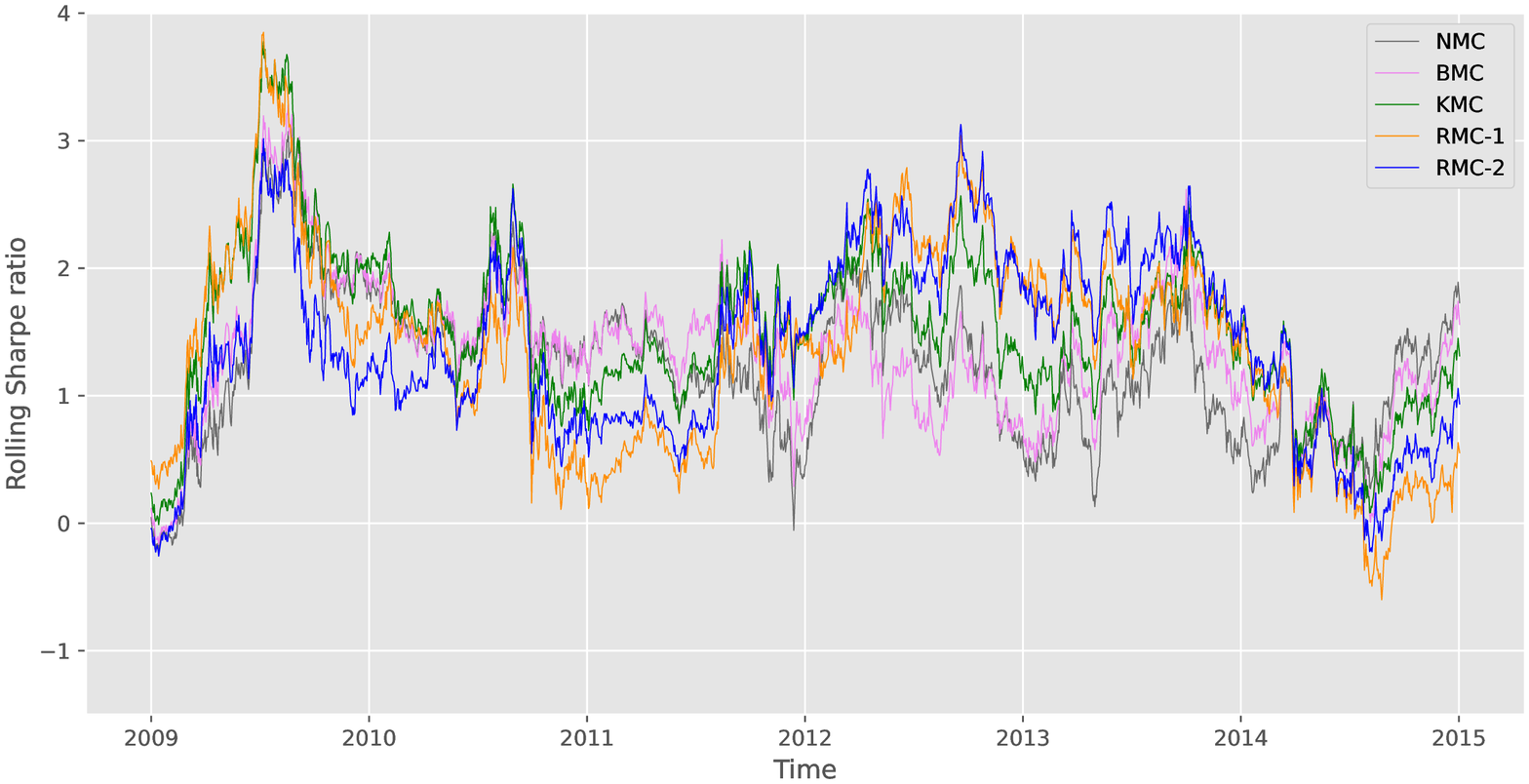}
\caption{Rolling one year Sharpe ratio without transaction costs, starting 2008.08.01}
\label{fig17}
\end{figure}
\begin{figure}[H]
\centering
\includegraphics[width=14cm]{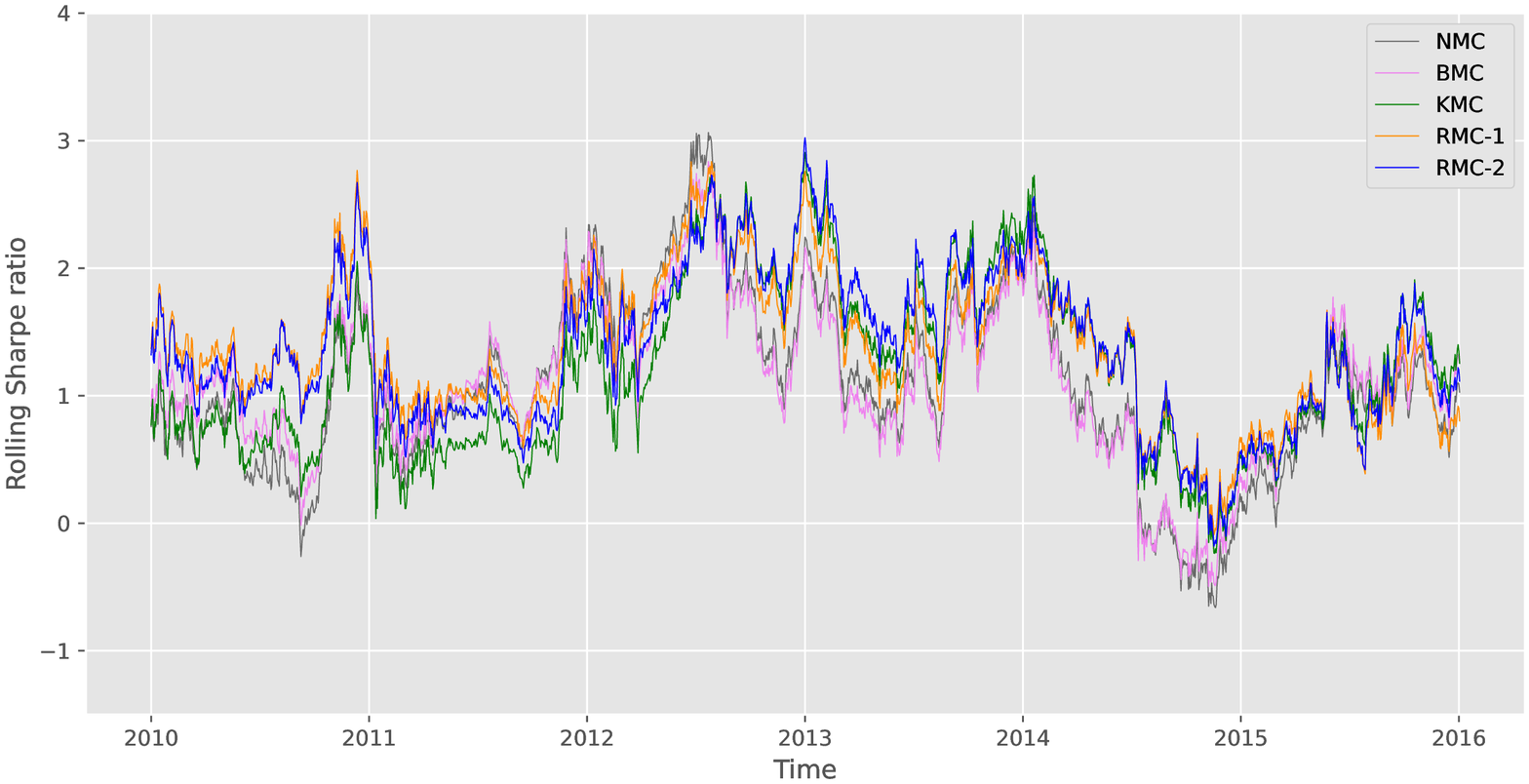}
\caption{Rolling one year Sharpe ratio without transaction costs, starting 2009.06.01}
\label{fig18}
\end{figure}

\begin{figure}[H]
\centering
\includegraphics[width=12cm]{Fig/act_weight2002.eps}
\caption{Composition of portfolios with transaction costs, starting 2002.02.01}
\label{fig19}
\end{figure}
\begin{figure}[H]
\centering
\includegraphics[width=12cm]{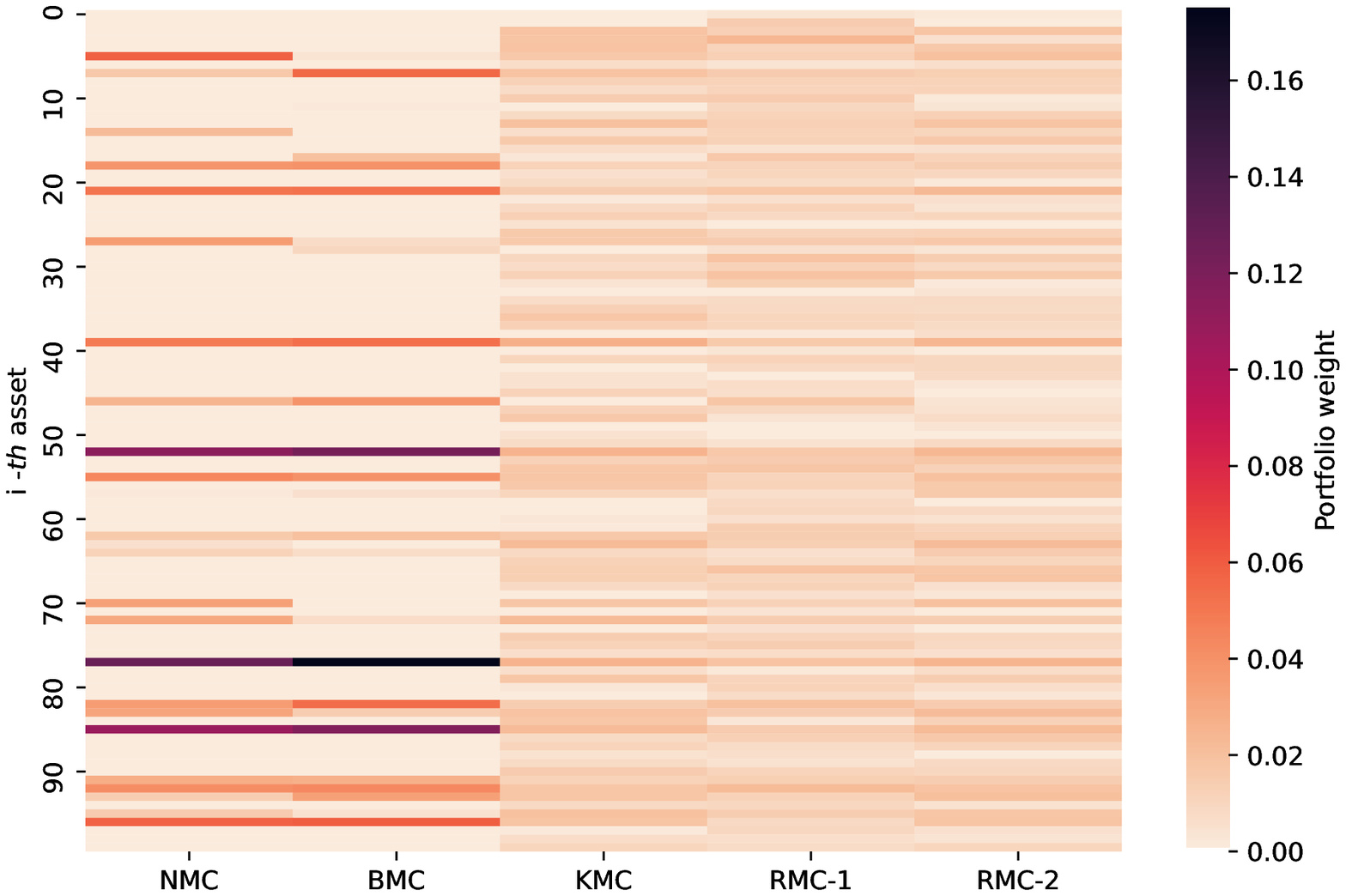}
\caption{Composition of portfolios with transaction costs, starting 2004.06.01}
\label{fig20}
\end{figure}
\begin{figure}[H]
\centering
\includegraphics[width=12cm]{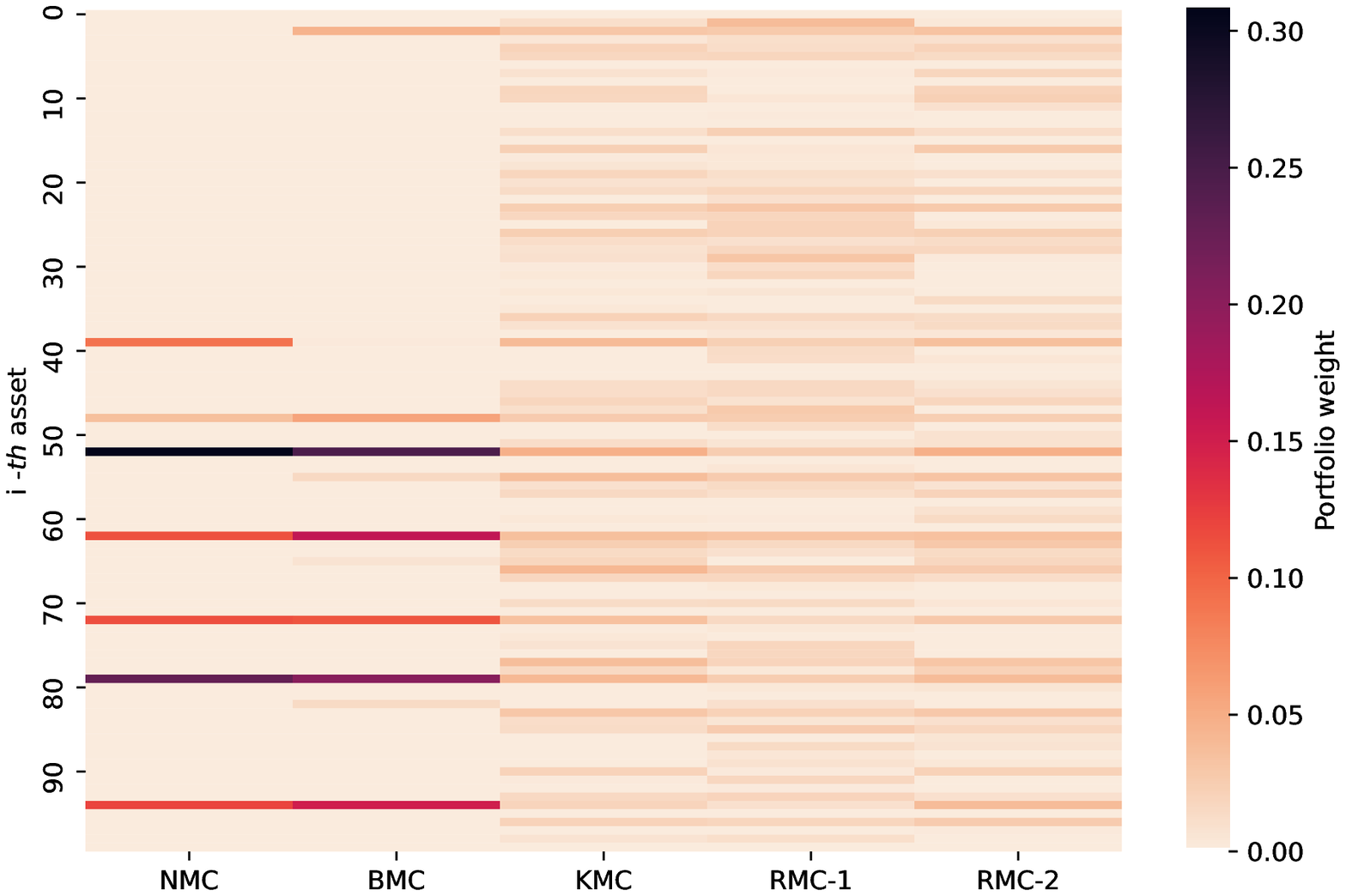}
\caption{Composition of portfolios with transaction costs, starting 2006.06.01}
\label{fig21}
\end{figure}
\begin{figure}[H]
\centering
\includegraphics[width=12cm]{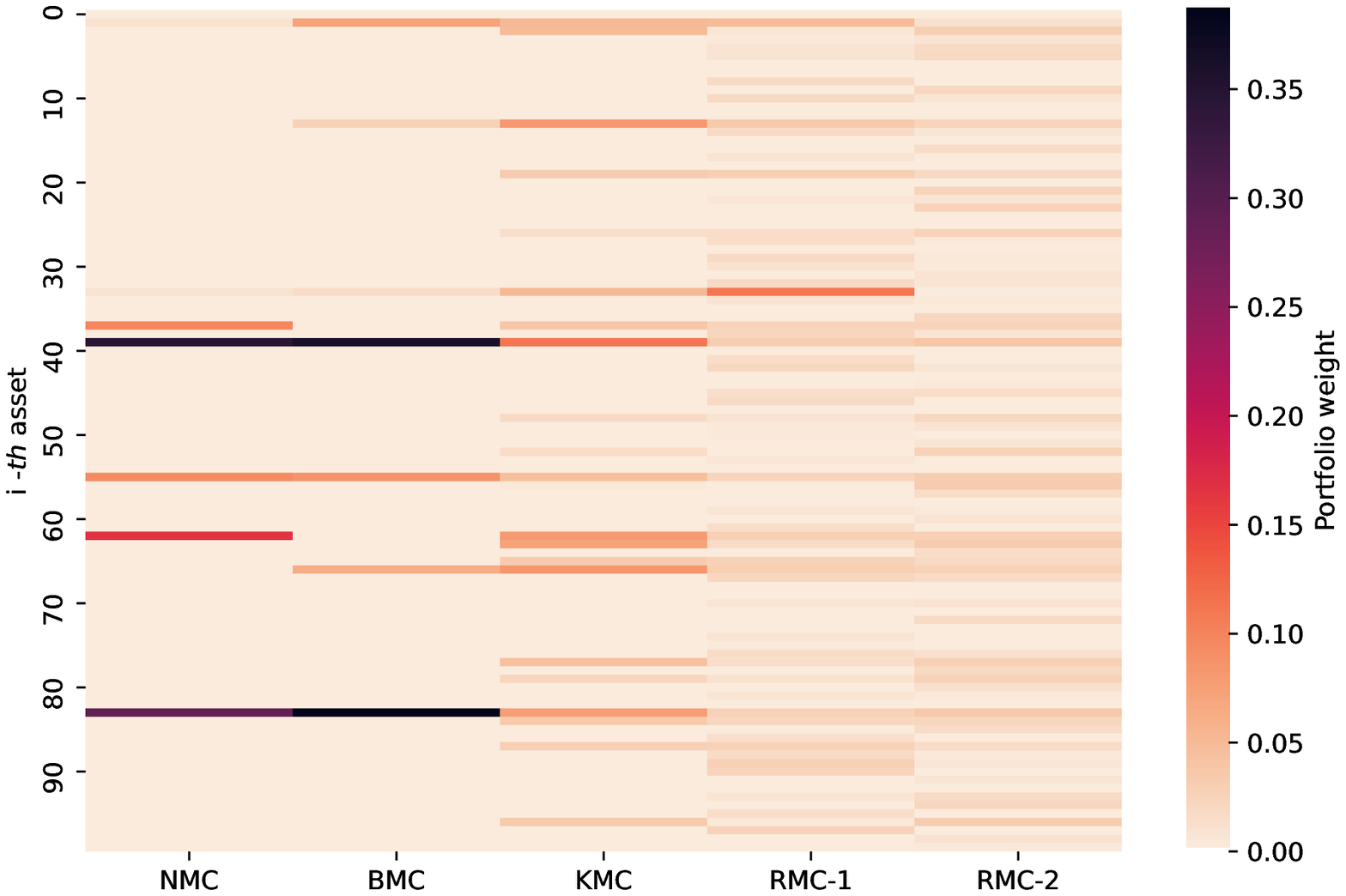}
\caption{Composition of portfolios with transaction costs, starting 2008.08.01}
\label{fig22}
\end{figure}
\begin{figure}[H]
\centering
\includegraphics[width=12cm]{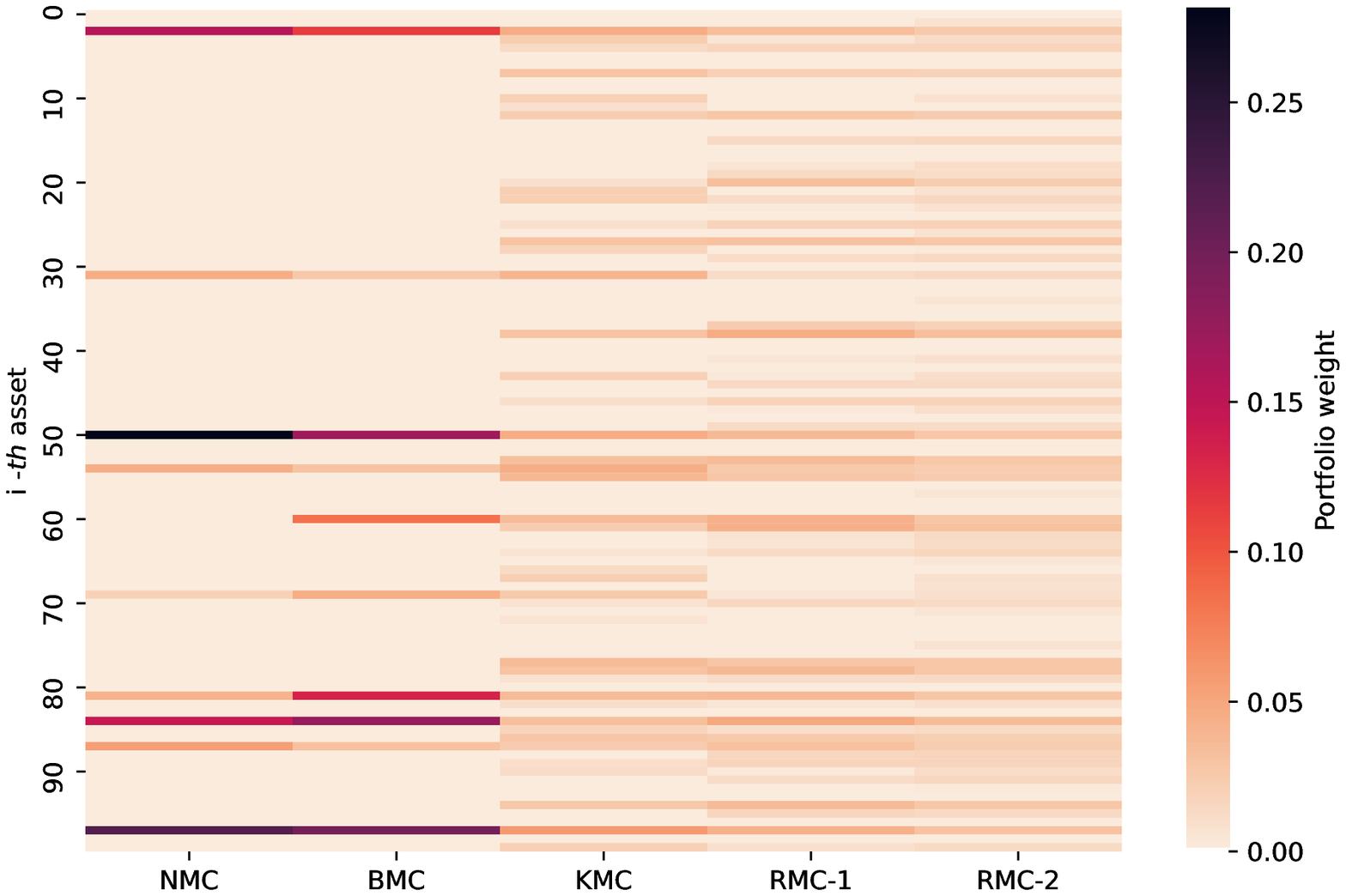}
\caption{Composition of portfolios with transaction costs, starting 2009.06.01}
\label{fig23}
\end{figure}

\begin{figure}[H]
\centering
\includegraphics[width=14cm]{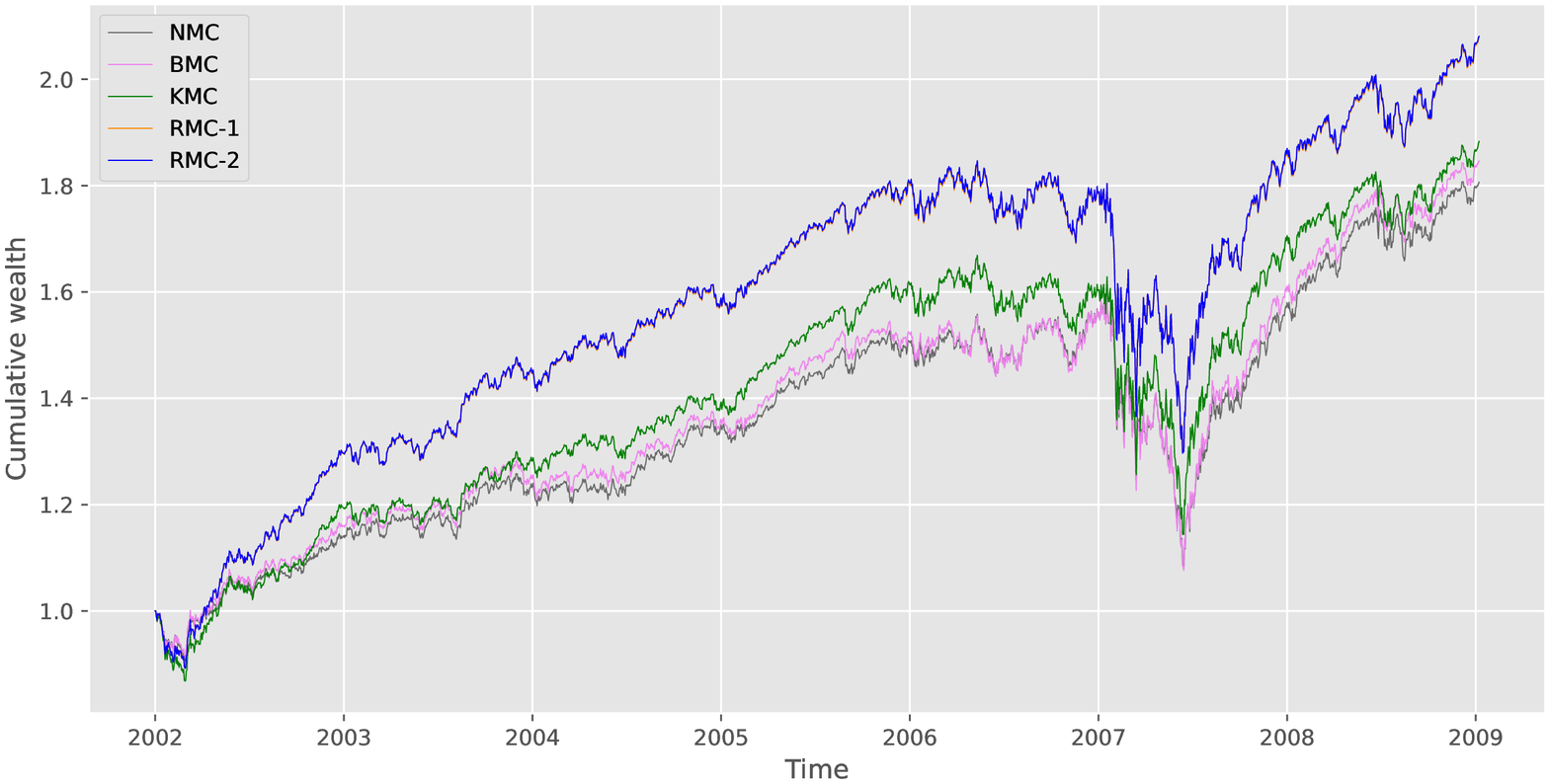}
\caption{Cumulative wealth with transaction costs, starting 2002.02.01}
\label{fig24}
\end{figure}
\begin{figure}[H]
\centering
\includegraphics[width=14cm]{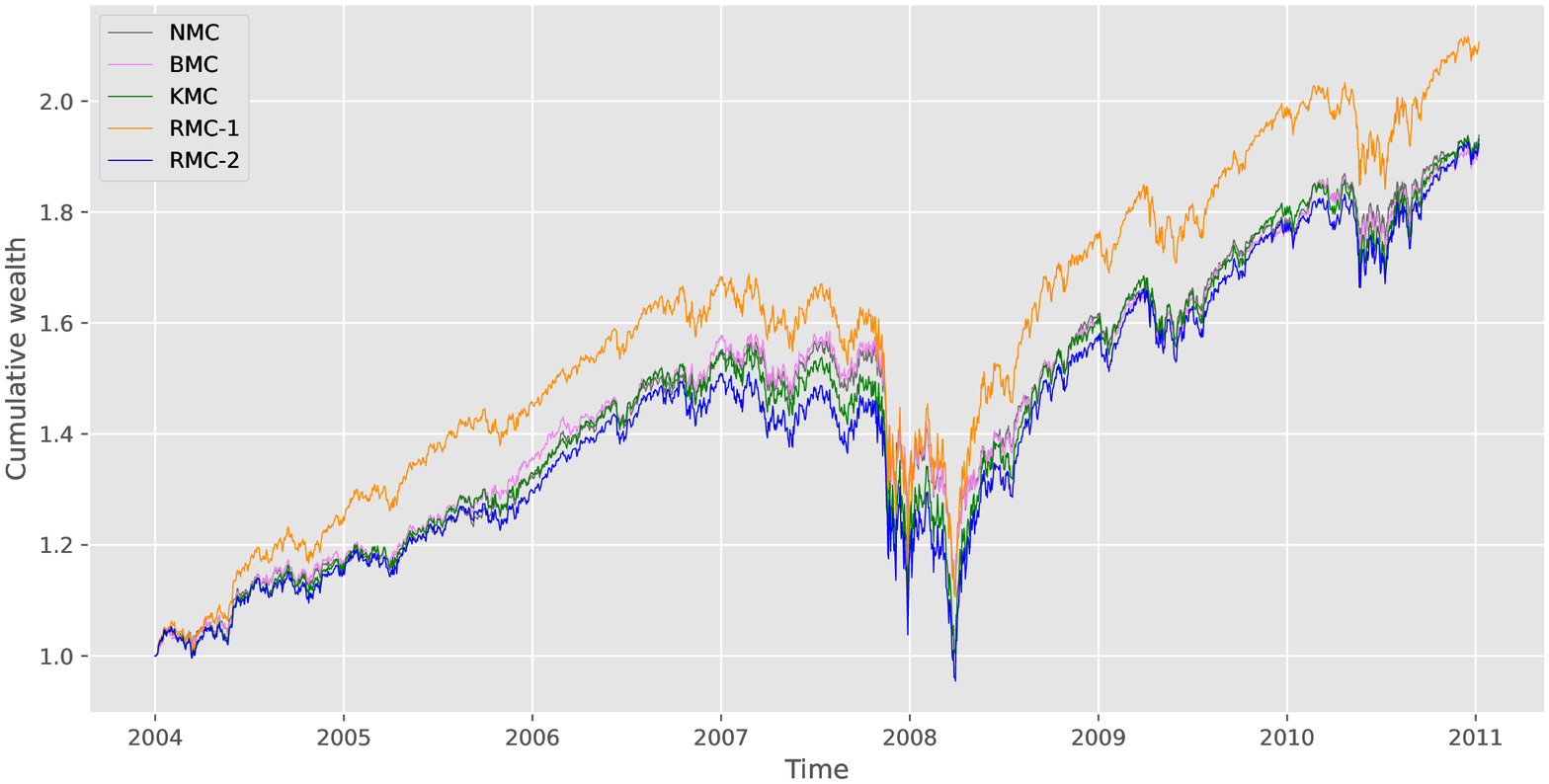}
\caption{Cumulative wealth with transaction costs, starting 2004.06.01}
\label{fig25}
\end{figure}
\begin{figure}[H]
\centering
\includegraphics[width=14cm]{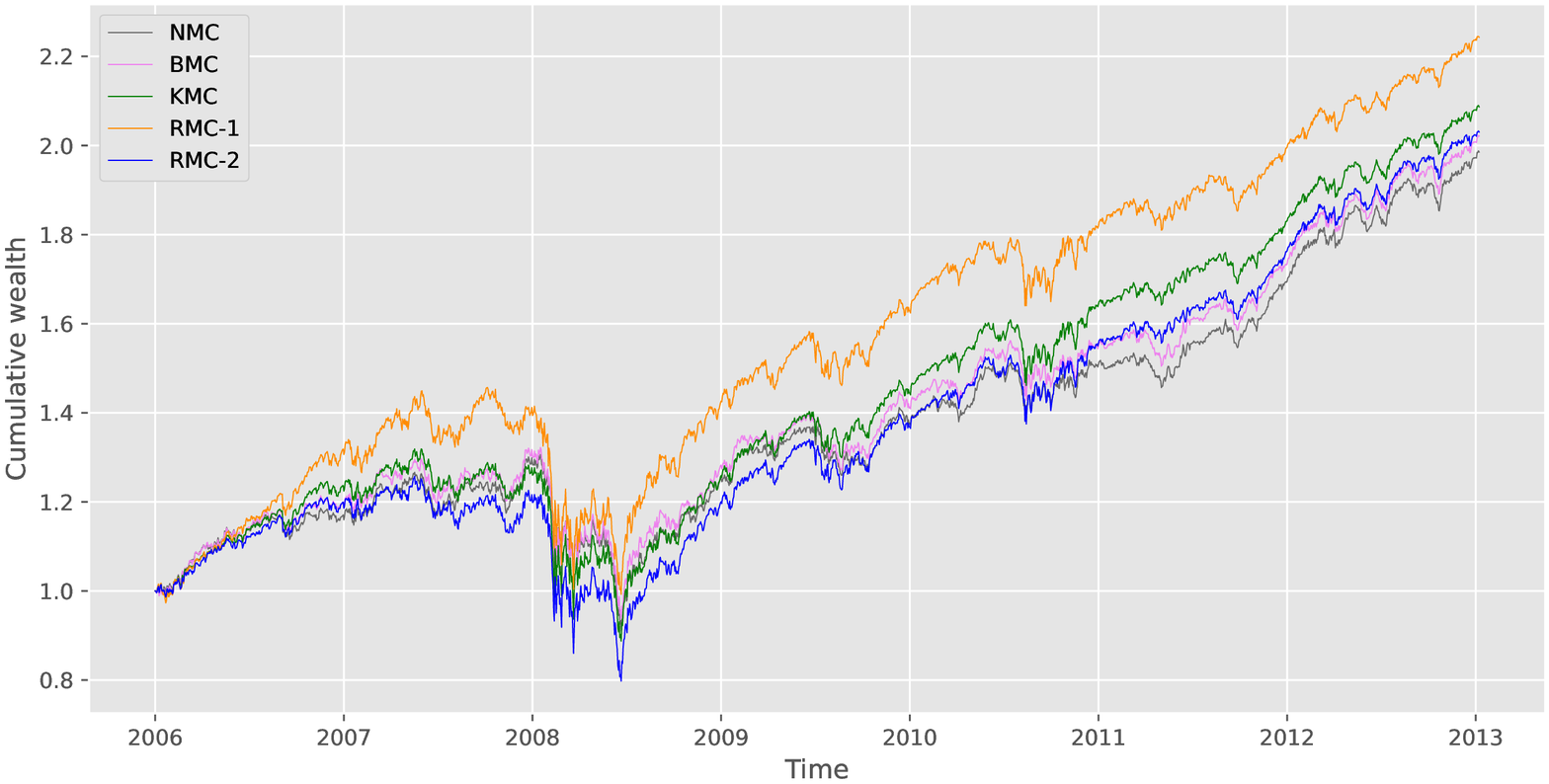}
\caption{Cumulative wealth with transaction costs, starting 2006.08.01}
\label{fig26}
\end{figure}
\begin{figure}[H]
\centering
\includegraphics[width=14cm]{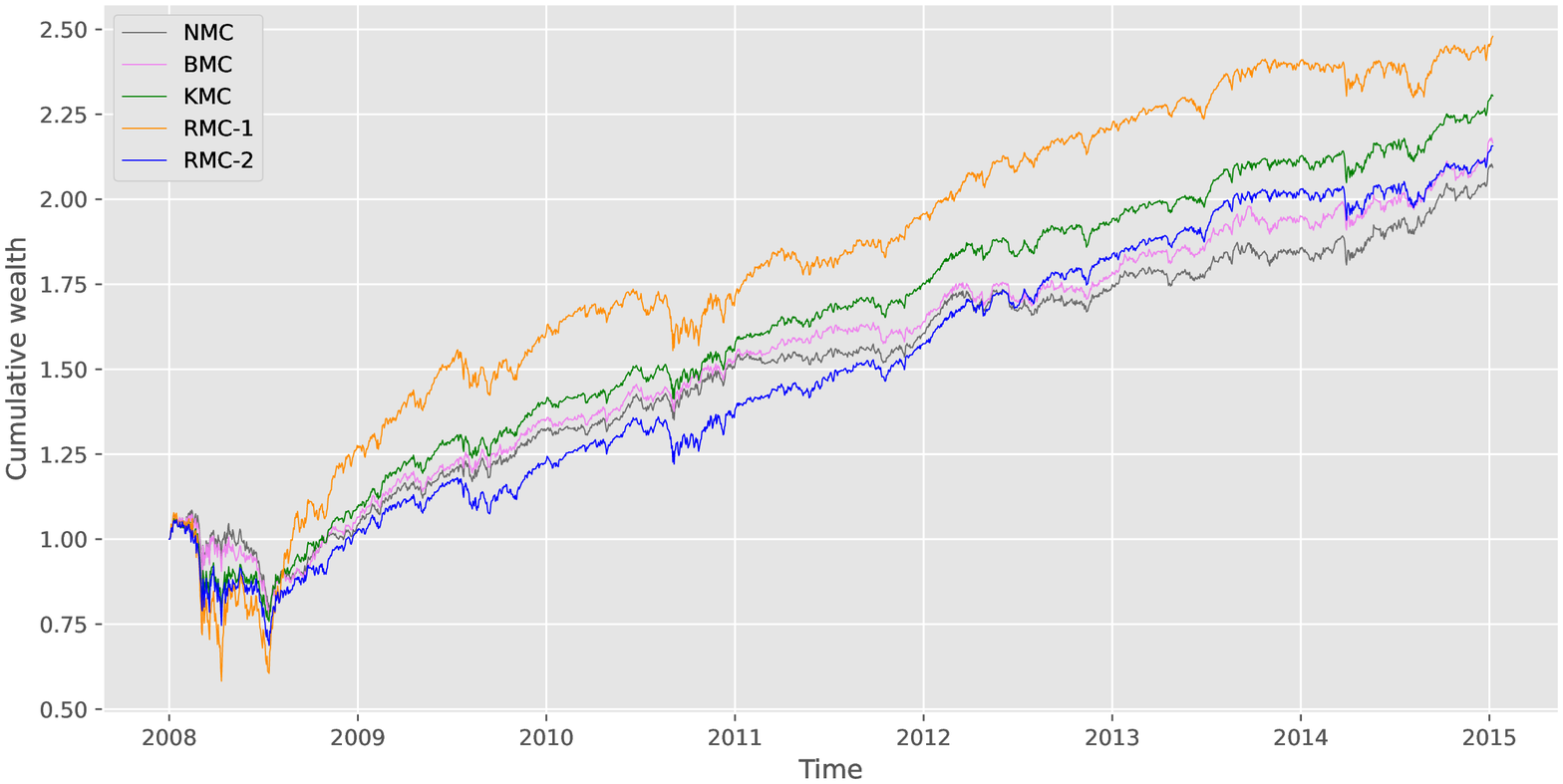}
\caption{Cumulative wealth with transaction costs, starting 2008.08.01}
\label{fig27}
\end{figure}
\begin{figure}[H]
\centering
\includegraphics[width=14cm]{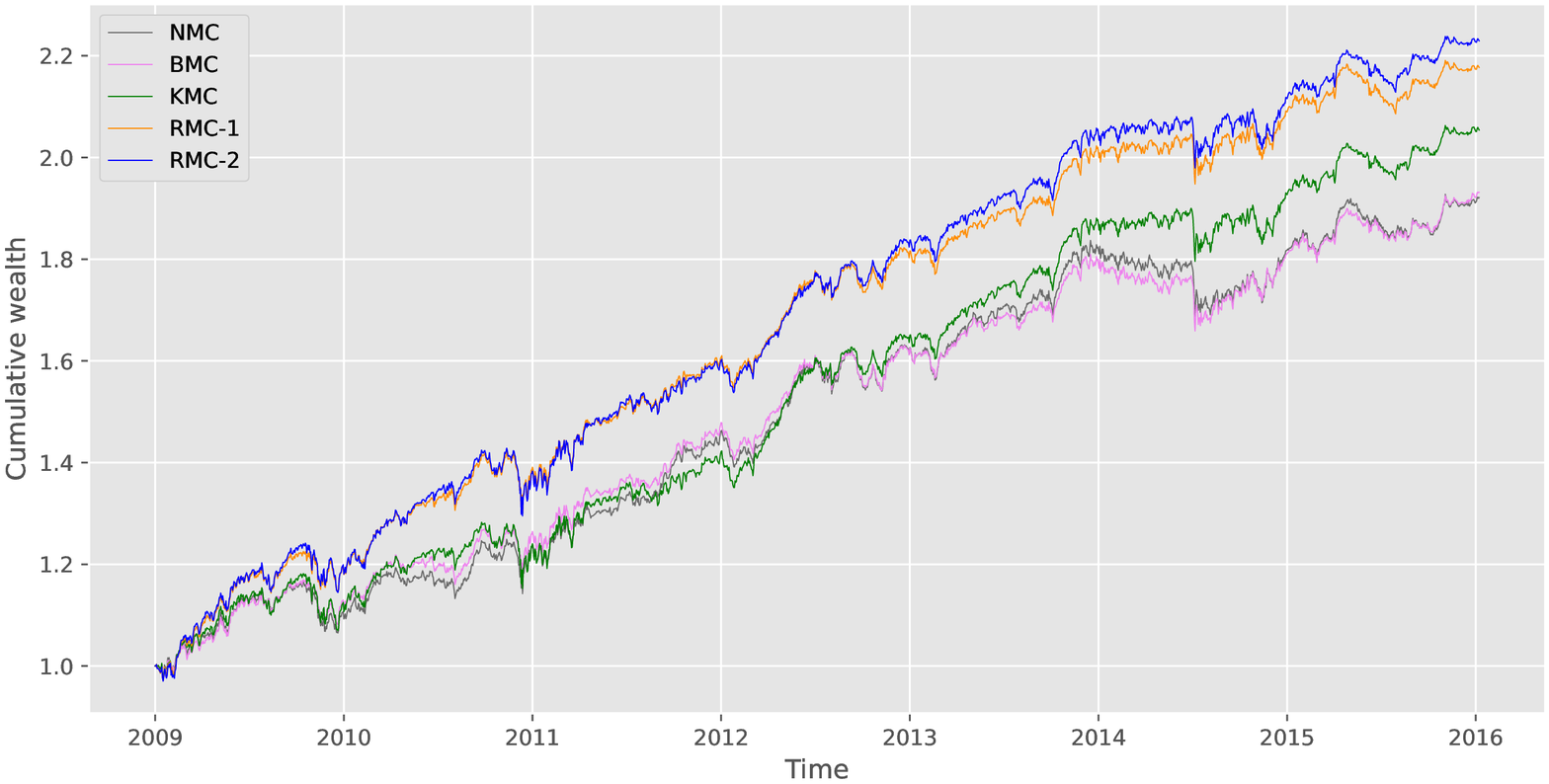}
\caption{Cumulative wealth with transaction costs, starting 2009.06.01}
\label{fig28}
\end{figure}

\begin{figure}[H]
\centering
\includegraphics[width=14cm]{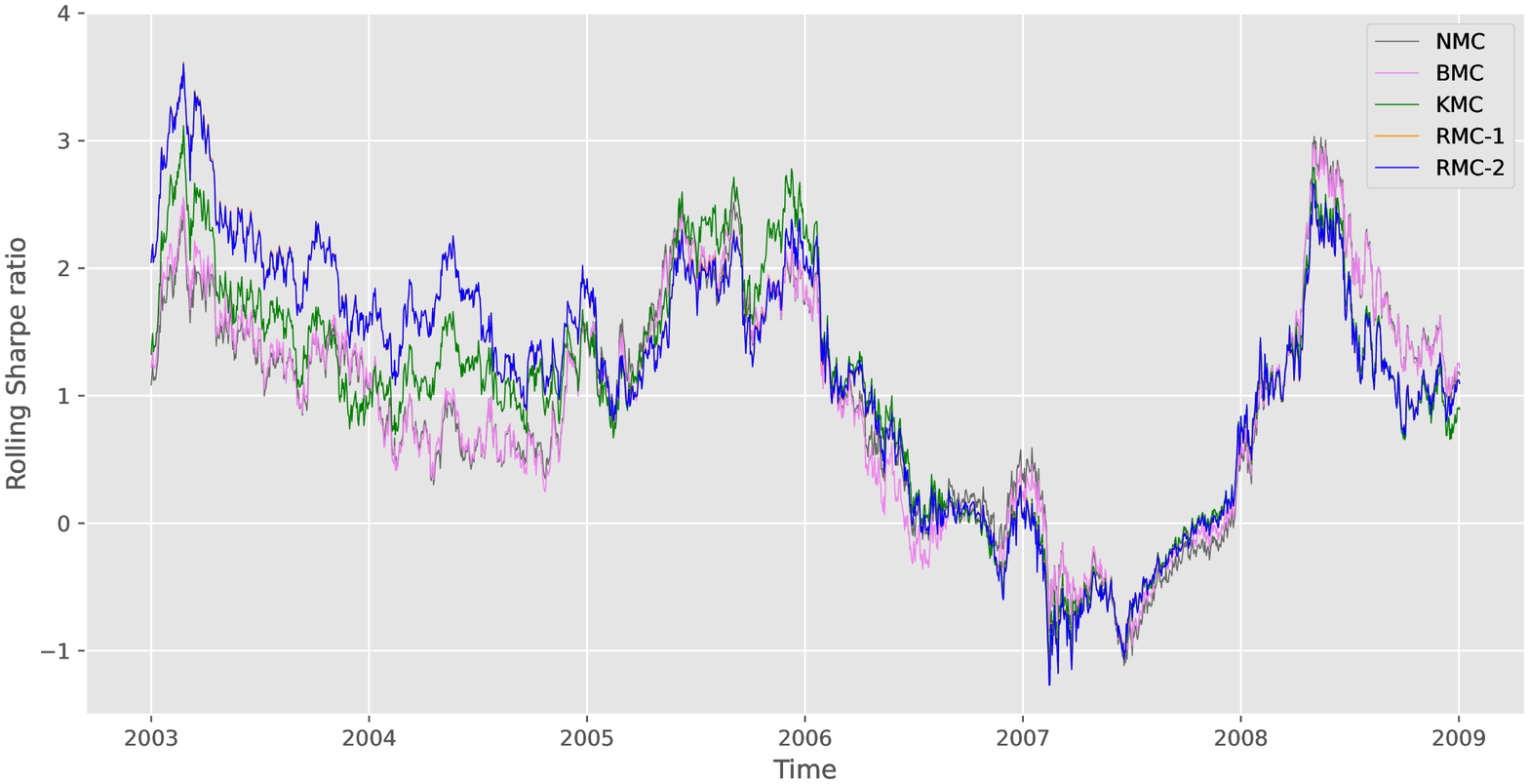}
\caption{Rolling one year Sharpe ratio with transaction costs, starting 2002.02.01}
\label{fig29}
\end{figure}
\begin{figure}[H]
\centering
\includegraphics[width=14cm]{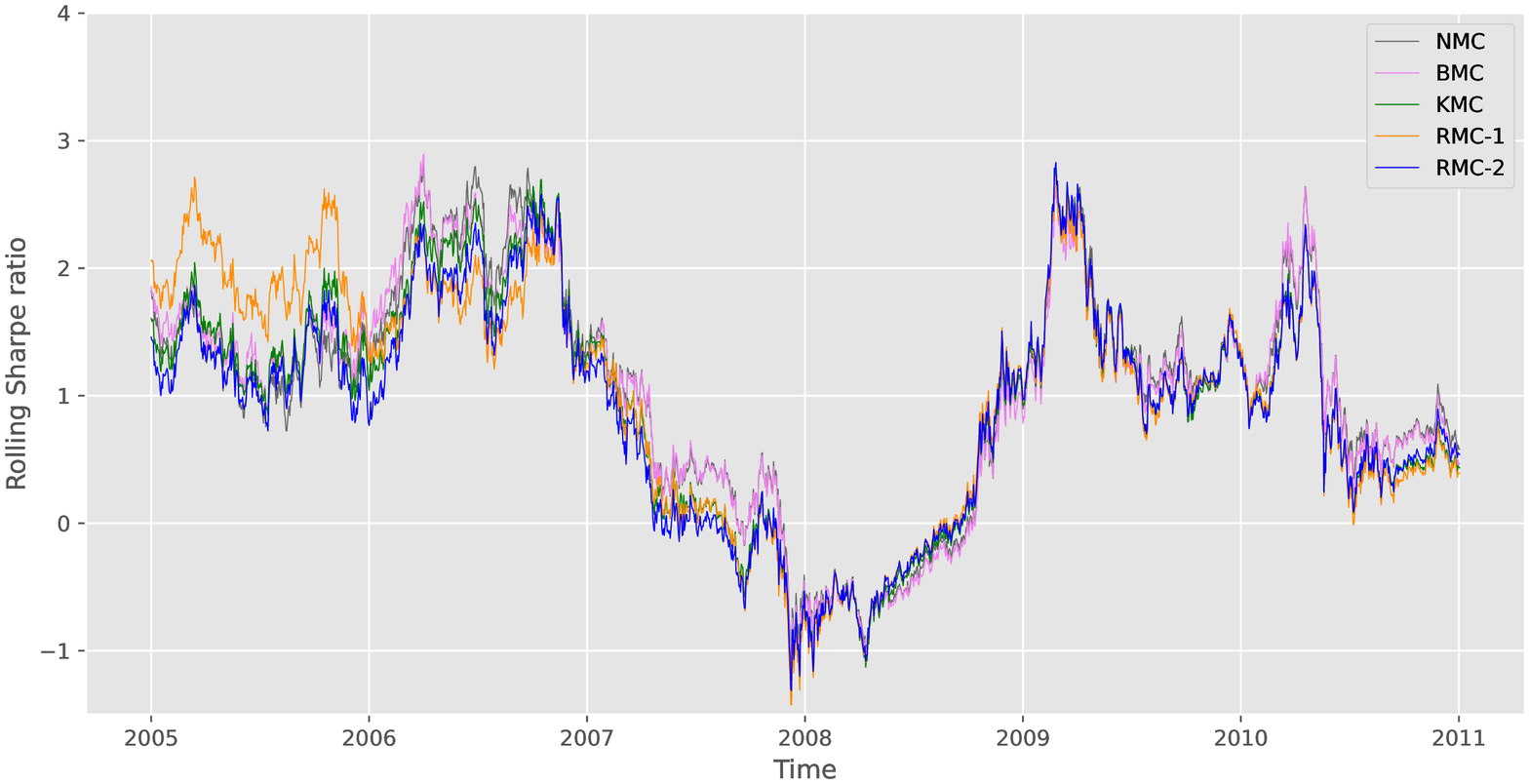}
\caption{Rolling one year Sharpe ratio with transaction costs, starting 2004.06.01}
\label{fig30}
\end{figure}
\begin{figure}[H]
\centering
\includegraphics[width=14cm]{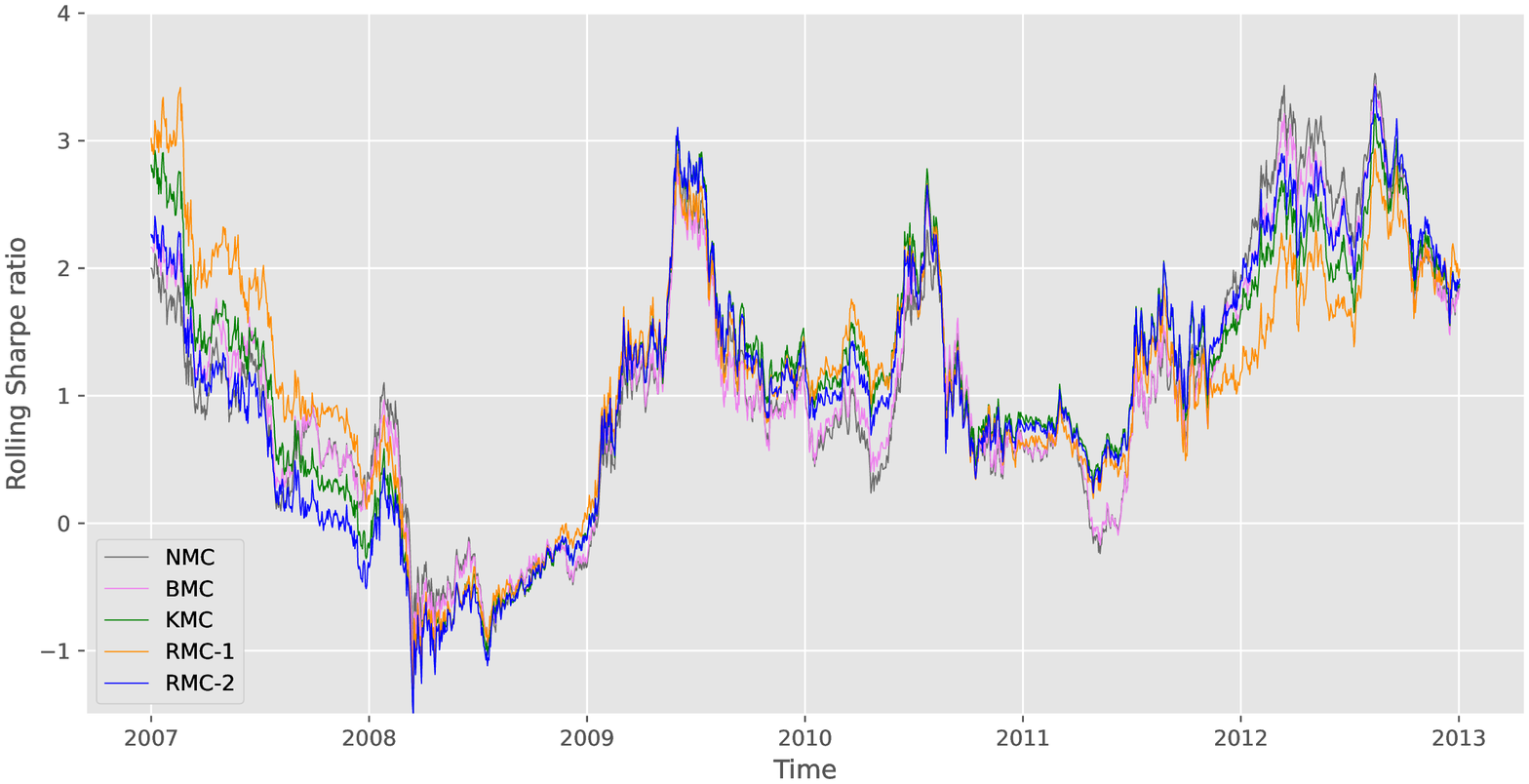}
\caption{Rolling one year Sharpe ratio with transaction costs, starting 2006.06.01}
\label{fig31}
\end{figure}
\begin{figure}[H]
\centering
\includegraphics[width=14cm]{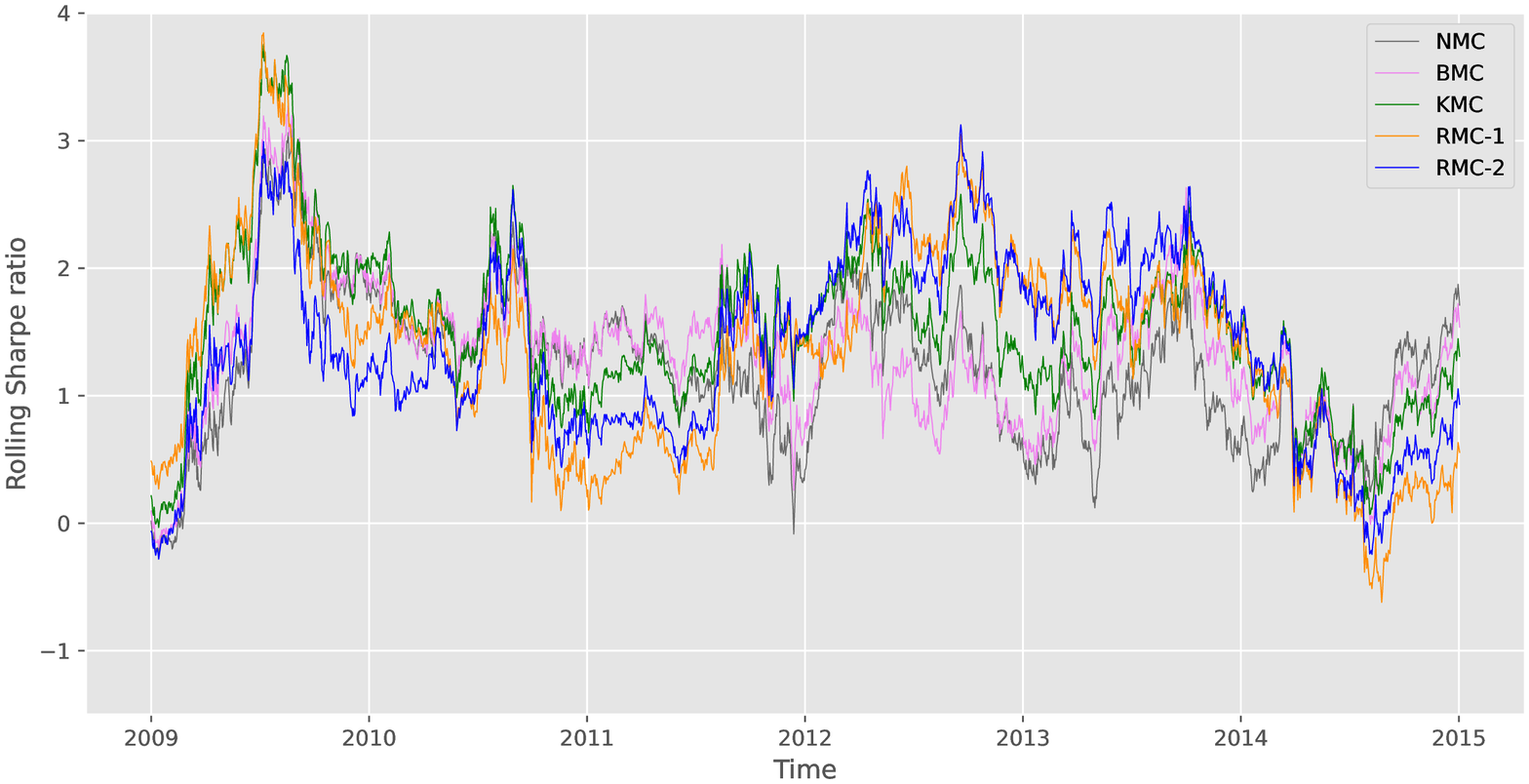}
\caption{Rolling one year Sharpe ratio with transaction costs, starting 2008.08.01}
\label{fig32}
\end{figure}
\begin{figure}[H]
\centering
\includegraphics[width=14cm]{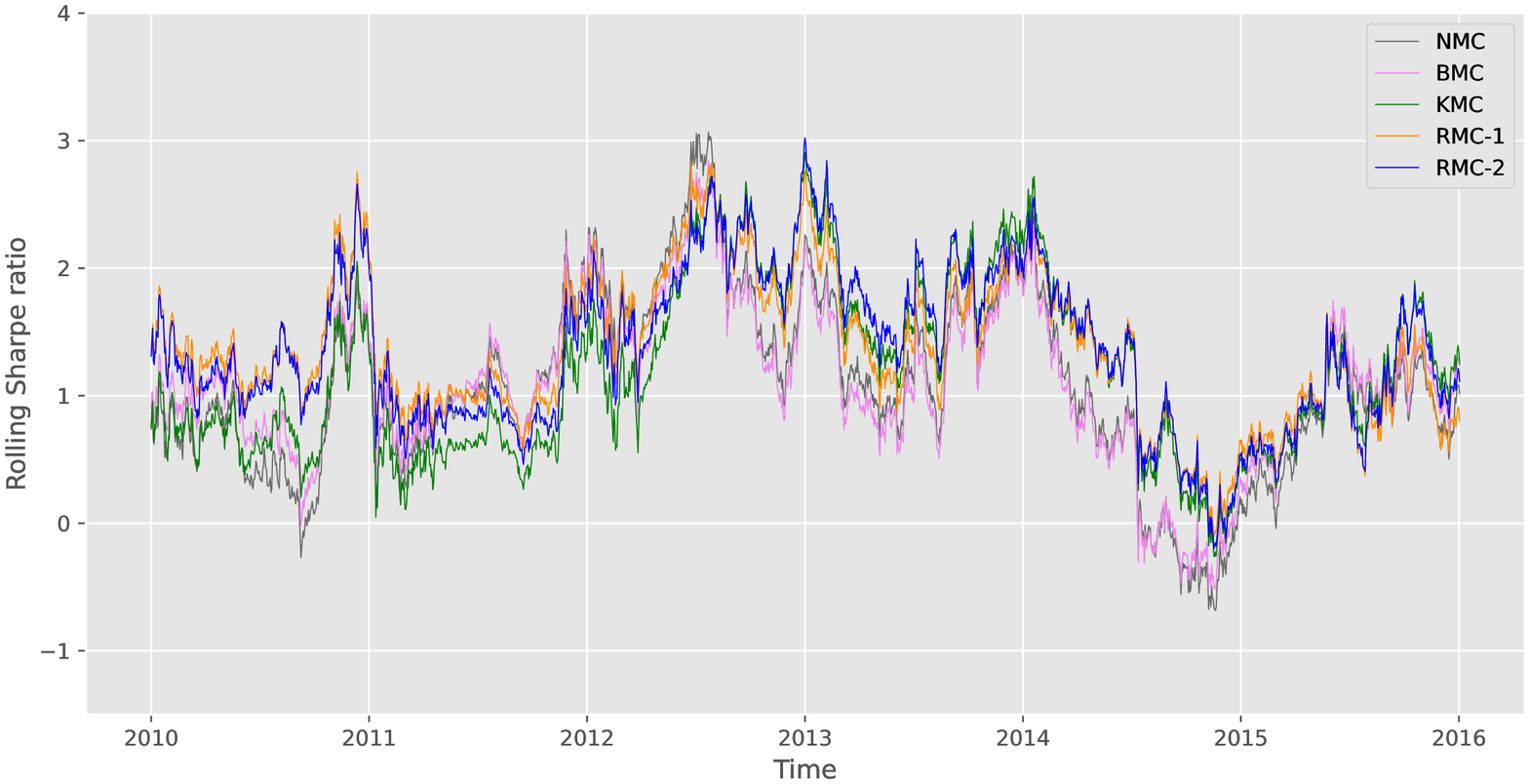}
\caption{Rolling one year Sharpe ratio with transaction costs, starting 2009.06.01}
\label{fig33}
\end{figure}

\section*{Tables}

\begin{table}[H]\centering
\small 
\begin{tabular}{ |c||c|c|c|c|c| }
 \hline
2002.02.01 & Mean (Daily) & Std Dev (Daily) & CVaR$_{0.95}$ & Sharpe (Annualized) & Mean/CVaR \\
 \hline
NMC  & 0.000414214 & 0.010989046 & 0.029964765 & 0.598363528 & 0.013823368\\ 
BMC & 0.000435181 & 0.011647557 & 0.031060781 & 0.593111407 & 0.014010626\\
KMC & 0.000455863 & 0.012299988 & 0.034368950 & 0.588342707 & 0.013263803\\
RMC-1 & 0.000552188 & 0.013090301 & 0.035034627 & 0.669634198 & 0.015761207\\
RMC-2 & 0.000553807 & 0.013134099 & 0.032437378 & 0.669358810 & 0.017073134\\
 \hline
2004.06.01 & Mean (Daily) & Std Dev (Daily) & CVaR$_{0.95}$ & Sharpe (Annualized) & Mean/CVaR \\
 \hline
NMC  & 0.000476001 & 0.011184933 & 0.033064265 & 0.675577574 & 0.014396253\\ 
BMC & 0.000465466 & 0.011344004 & 0.032089496 & 0.651361896 & 0.014505261\\
KMC &  0.000471342 & 0.012952865 & 0.028609607 & 0.577658086 & 0.016474965\\
RMC-1 & 0.000495928 & 0.011402213 & 0.033977053 & 0.690446884 & 0.014595986\\
RMC-2 & 0.000466915 & 0.013160472 & 0.036372233 & 0.563205451 & 0.012837134\\
 \hline
2006.06.01 & Mean (Daily) & Std Dev (Daily) & CVaR$_{0.95}$ & Sharpe (Annualized) & Mean/CVaR \\
 \hline
NMC  & 0.000500284 & 0.009616184 & 0.023153042 & 0.825876164 & 0.021607740\\ 
BMC & 0.000517737 & 0.009946102 & 0.023363066 & 0.826336403 & 0.022160507\\
KMC & 0.000548701 & 0.010665874 & 0.036628863 & 0.816656828 & 0.014980018\\
RMC-1 & 0.000519369 & 0.009652195 & 0.023037725 & 0.854182106 & 0.022544296\\
RMC-2 & 0.000521862 & 0.010785479 & 0.026042552 & 0.768098063 & 0.020038827\\
 \hline
2008.08.01 & Mean (Daily) & Std Dev (Daily) & CVaR$_{0.95}$ & Sharpe (Annualized)& Mean/CVaR\\
 \hline
NMC & 0.000551496 & 0.009087272 & 0.023071616 & 0.963405690 & 0.023903657\\
BMC & 0.000584411 & 0.009215616 & 0.027692293 & 1.006687285 & 0.021103741\\
KMC  & 0.000653984 & 0.009559078 & 0.069318921 & 1.086055334 &  0.009434422\\
RMC-1 & 0.000737939 & 0.013222184 & 0.037927259 & 0.885968105 & 0.019456717\\
RMC-2 & 0.000580732 & 0.010365320 & 0.036360908 & 0.889393335 & 0.015971349\\
 \hline
2009.06.01 & Mean (Daily) & Std Dev (Daily) & CVaR$_{0.95}$ & Sharpe (Annualized) & Mean/CVaR \\
 \hline
NMC  & 0.000469028 & 0.007066659 & 0.016025745 & 1.053622161 & 0.029267157\\
BMC & 0.000474973 & 0.006836645 & 0.015247661 & 1.102876796 & 0.030747243\\
KMC & 0.000537719 & 0.007443031 & 0.037543810 & 1.146848850 & 0.014322441\\
RMC-1 & 0.000599987 & 0.007115376 & 0.016198153 & 1.338581097 & 0.037040495\\
RMC-2 & 0.000627472 & 0.007486768 & 0.019298999 & 1.330455427 & 0.032513194\\
\hline
\end{tabular}
\caption{Performances of different robust portfolio strategies without transaction costs.}
\label{table1}
\end{table}

\begin{table}[H]\centering
\small 
\begin{tabular}{ |c||c|c|c|c|c| }
 \hline
2002.02.01 & Mean (Daily) & Std Dev (Daily) & CVaR$_{0.95}$ & Sharpe (Annualized) & Mean/CVaR \\
 \hline
NMC & 0.000403039 & 0.010983332 & 0.030077326 & 0.582523783 & 0.013400094\\
BMC & 0.000423062 & 0.011640696 & 0.031149149 & 0.576933560 & 0.013581815\\
KMC & 0.000441529 & 0.012291726 & 0.034512515 & 0.570225580 & 0.012793301\\
RMC-1 & 0.000538792 & 0.013080727 & 0.035105688 & 0.653867770 & 0.015347727\\
RMC-2 & 0.000540320 & 0.013124600 & 0.032480528 & 0.653530297 & 0.016635216\\
 \hline
 2004.06.01 & Mean (Daily) & Std Dev (Daily) & CVaR$_{0.95}$ & Sharpe (Annualized) & Mean/CVaR \\
 \hline
NMC  & 0.000469330 & 0.011171314 & 0.033155308 & 0.666921925 & 0.014155523\\ 
BMC & 0.000458590 & 0.011331386 & 0.032185383 & 0.642453959 & 0.014248400\\
KMC & 0.000465784 & 0.012945511 & 0.028725431 & 0.571170848 & 0.016215055\\
RMC-1 & 0.000489308 & 0.011389217 & 0.034073714 & 0.682007719 & 0.014360297\\
RMC-2 & 0.000461418 & 0.013152442 & 0.036434588 & 0.556915338 & 0.012664307\\
 \hline
2006.06.01 & Mean (Daily) & Std Dev (Daily) & CVaR$_{0.95}$ & Sharpe (Annualized) & Mean/CVaR \\
 \hline
NMC  & 0.000492675 & 0.009618197 & 0.023242215 & 0.813144297 & 0.021197444\\ 
BMC & 0.000510663 & 0.009949806 & 0.023440742 & 0.814742125 & 0.021785281\\
KMC & 0.000543398 & 0.010655290 & 0.036826561 & 0.809568805 & 0.014755624\\
RMC-1 & 0.000510559 & 0.009652738 & 0.023128845 & 0.839646198 & 0.022074590\\
RMC-2 & 0.000515162 & 0.010774696 & 0.026043978 & 0.758996643 & 0.019780501\\
 \hline
2008.08.01 & Mean (Daily) & Std Dev (Daily) & CVaR$_{0.95}$ & Sharpe (Annualized) & Mean/CVaR\\
 \hline
NMC & 0.000546916 & 0.009078253 & 0.023064818 & 0.956355251 & 0.023712131\\
BMC & 0.000583198 & 0.009206164 & 0.027685907 & 1.005629031 & 0.021064796\\
KMC  & 0.000652256 & 0.009549088 & 0.069674256 & 1.084318071 & 0.009361506\\
RMC-1 & 0.000739663 & 0.013209766 & 0.037904801 & 0.888872286 & 0.019513717\\
RMC-2 & 0.000578638 & 0.010353592 & 0.036338335 & 0.887188993 & 0.015923624\\
 \hline
2009.06.01 & Mean (Daily) & Std Dev (Daily) & CVaR$_{0.95}$ & Sharpe (Annualized) & Mean/CVaR\\
 \hline
NMC & 0.000460546 & 0.007073345 & 0.016076162 & 1.033591132 & 0.028647758\\
BMC & 0.000465774 & 0.006841371 & 0.015483874 & 1.080767852 & 0.030081231\\
KMC & 0.000526968 & 0.007441115 & 0.037741824 & 1.124208648 & 0.013962441\\
RMC-1 & 0.000588551 & 0.007114447 & 0.016233162 & 1.313237698 & 0.036256100\\
RMC-2 & 0.000614752 & 0.007481517 & 0.019297441 & 1.304400901 & 0.031856701\\
\hline
\end{tabular}
\caption{Performances of different portfolio strategies with transaction costs.}
\label{table2}
\end{table}

\end{document}